\immediate\write18{makeindex -s nomencl.ist \jobname.nlo -o \jobname.nls}
\documentclass[12pt]{article}
\usepackage[twoside,bindingoffset=6mm,verbose,marginratio={4:6,5:7},%
textwidth=135mm,height=179.6mm]{geometry}%
\usepackage{amsmath,amssymb,amsfonts}  
\usepackage{amsthm}
\usepackage{graphicx,psfrag,epsf}
\usepackage{enumerate}
\usepackage{natbib}
\usepackage{url} % not crucial - just used below for the URL 
\usepackage{microtype}
\usepackage{xspace}
\usepackage{import}
\usepackage{tikz} 
\usetikzlibrary{arrows,decorations.pathmorphing,backgrounds,positioning,fit,petri, shapes}
\usepackage{hyperref}
\usepackage{graphicx}

\usepackage{xspace}
\usepackage{microtype}

\usepackage{color}
\usepackage{dsfont}
\usepackage[Pseudocode]{algorithm}
\usepackage{algpseudocode}
\usepackage{acronym}
\usepackage{pdfpages} 

\usepackage{nomencl}

\makenomenclature

\usepackage{etoolbox}
\renewcommand\nomgroup[1]{%
	\item[\bfseries
	\ifstrequal{#1}{A}{Sets}{%
		\ifstrequal{#1}{B}{Measures}{%
			\ifstrequal{#1}{C}{Random Variables}{
				\ifstrequal{#1}{D}{Functions}{
					\ifstrequal{#1}{E}{Others}{						
						\ifstrequal{#1}{F}{Distributions}{
			}}}}}}%
	]}
\usepackage{mathrsfs}
\DeclareMathAlphabet{\mathantt}{OT1}{antt}{m}{n}
\graphicspath{{../paper_cred_reg/pictures/},{./pictures/},{./}}

\DeclareMathOperator*{\argmin}{argmin}
\DeclareMathOperator*{\argmax}{argmax}  
  
\newcommand*{\LargerCdot}{{\raisebox{-0.4ex}{$\,$\scalebox{1.8}{$\cdot$}$\,$}}}
\newcommand*{\eqdist}{\raisebox{-0.3ex}{$\:\overset{\raisebox{-0.5ex}{\scalebox{0.7}{d}}}{\sim}\:$}}

\newcommand*{\eunion}{\ensuremath{\mathsf{U}}} 
\newcommand*{\minkcov}{\ensuremath{k\text{\texttt{-minEU}}}\xspace} 
\newcommand*{\knaps}{\texttt{Knapsack}\xspace}
\newcommand*{\SBP}{SBP\xspace} 

\newcommand*{\rand}[1]{\ensuremath{\boldsymbol{\mathbf{#1}}}\xspace} 
\newcommand*{\densslash}{\ensuremath{\boldsymbol{\mathbf{\slash}}}\xspace} 
\newcommand*{\algor}[1]{\texttt{#1}\xspace} 
\newcommand*{\dens}{\ensuremath{\mathantt{p}}\xspace} 
\newcommand*{\densk}{\raisebox{0mm}{\ensuremath{\mathantt{k}}}\xspace} 
\newcommand*{\moveprob}{\ensuremath{\mathantt{q}}\xspace} 
\newcommand*{\cost}{\ensuremath{\mathantt{f}}\xspace} 
\newcommand*{\funcf}{\ensuremath{\mathantt{f}}\xspace} 
\newcommand*{\funcg}{\ensuremath{\mathantt{g}}\xspace} 
 
\newcommand*{\funcb}{\ensuremath{\mathantt{b}}\xspace} 
\newcommand*{\euler}[1]{\ensuremath{\exp\left(#1\right)}\xspace} 
\newcommand*{\euleremp}{\ensuremath{\exp}\xspace}

\newcommand*{\tho}{\ensuremath{\theta_{\text{old}}}\xspace} 
\newcommand*{\thn}{\ensuremath{\theta_{\text{new}}}\xspace} 
\newcommand*{\qold}{\ensuremath{q_{\text{old}}}\xspace} 
\newcommand*{\qnew}{\ensuremath{q_{\text{new}}}\xspace} 
\newcommand*{\semic}{\ensuremath{\mathbin{\boldsymbol;}}\xspace} 
\newcommand*{\cloglik}{\ensuremath{\mcal{L}}\xspace} 
\newcommand*{\ratio}{\ensuremath{\mathantt{r}}\xspace} 
\newcommand*{\prob}[1]{\ensuremath{\mathbb{P}\left(#1\right)}\xspace} 
\newcommand*{\probemp}{\ensuremath{\mathbb{P}\xspace}} 
\newcommand*{\expec}[1]{\ensuremath{\mathbb{E}\left(#1\right)}\xspace} 
\newcommand*{\expecemp}{\ensuremath{\mathbb{E}}\xspace} 
\newcommand{\ind}[1]{\raisebox{-0.3mm}{\scalebox{1.2}{\ensuremath{\mathds{1}}}}\ensuremath{\left\{ #1\right\}}} 
\newcommand{\indemp}{\raisebox{-0.3mm}{\scalebox{1.2}{\ensuremath{\mathds{1}}}}} 
 
\newcommand*{\sub}{\ensuremath{\subseteq}} 
\newcommand*{\bmid}{\ensuremath{\ \Big\vert\ }} 
\newcommand*{\bbmid}{\ensuremath{\ \Bigg\vert\ }} 
\newcommand*{\pows}[1]{\ensuremath{2^{#1}}} 
\newcommand*{\rel}[1]{\ensuremath{\mathantt{r}(#1)}}
\newcommand*{\hyperg}{\ensuremath{\mathcal{H}}} 
\newcommand*{\hypergg}{\ensuremath{\mathcal{G}}} 
\newcommand*{\hg}{\ensuremath{\mathcal{H}}} 
\newcommand*{\hgg}{\ensuremath{\mathcal{G}}} 
\newcommand*{\vertcov}{\ensuremath{\mathcal{D}}}
\newcommand*{\contx}{\ensuremath{D_x}}
\newcommand*{\cardi}[1]{\ensuremath{\##1}\xspace}

\newcommand*{\SampleProb}[2]{\ensuremath{\mathfrak{S}(#1,#2)}}

\newcommand*{\supmarginal}[1]{\ensuremath{\mathfrak{B}(#1)}}
\newcommand*{\primeprob}[1]{\ensuremath{\mathfrak{R}(#1)}}
\newcommand*{\stepR}{\texttt{stepR}\xspace}
\newcommand*{\greedy}{\texttt{Greedy}\xspace}
\newcommand*{\limi}[1]{\ensuremath{\lim\limits_{#1\rightarrow \infty}}} 
 
\newcommand*{\limii}[2]{\ensuremath{\ \overset{#1\rightarrow #2}{\longrightarrow}\ }}
\newcommand*{\cp}{changepoint\xspace}
\newcommand*{\cps}{changepoints\xspace}
\newcommand*{\Cp}{Changepoint\xspace}

\newcommand*{\states}{\ensuremath{\mathcal{S}}\xspace}
\newcommand*{\ustates}{\ensuremath{\mathcal{U}}\xspace}
\newcommand*{\moves}{\ensuremath{\mathcal{M}}\xspace}
\newcommand*{\spaces}{\ensuremath{\mathcal{D}}\xspace}
\newcommand*{\jprior}{\ensuremath{\mathcal{J}}\xspace}

\newcommand*{\heights}{\ensuremath{\mathcal{H}}\xspace}

\newcommand*{\omeasure}{\ensuremath{\psi}\xspace}
\newcommand*{\mcal}[1]{\ensuremath{\mathcal{#1}}\xspace}

\newcommand*{\bo}[1]{\ensuremath{\mathbf{#1}}\xspace}
\newcommand*{\todo}[1]{}

\newcommand*{\obssigma}{\ensuremath{\mathcal{Y}}\xspace}
\newcommand*{\latsigma}{\ensuremath{\mathcal{X}}\xspace}
\newcommand*{\obs}{\ensuremath{\mathbb{Y}}\xspace}
\newcommand*{\lat}{\ensuremath{\mathbb{X}}\xspace}

\newcommand*{\indep}{\:\rotatebox[origin=c]{90}{\ensuremath{\models}}\:}

\newcommand*{\acc}{\ensuremath{\mathantt{a}}\xspace} 
 
\newcommand*{\rcc}{\ensuremath{\mathantt{r}}\xspace} 
\newcommand*{\occ}{\ensuremath{\mathantt{o}}\xspace} 
 
\newcommand*{\ccc}{\ensuremath{\mathantt{c}}\xspace} \newcommand*{\kcc}{\ensuremath{\mathantt{k}}\xspace}

\newcommand*{\scc}{\ensuremath{\mathantt{s}}\xspace} 
\newcommand*{\tccc}{\ensuremath{\mathantt{\tilde c}}\xspace} 
\newcommand*{\bcc}{\ensuremath{\mathantt{b}}\xspace} 
 
\newcommand*{\qcc}{\ensuremath{\mathantt{q}}\xspace} 
\newcommand*{\tqcc}{\ensuremath{\mathantt{\tilde q}}\xspace} 
\newcommand*{\dcc}{\ensuremath{\mathantt{d}}\xspace} 
 
\newcommand*{\Zcc}{\ensuremath{\mathantt{Z}}\xspace} 
\newcommand*{\ecc}{\ensuremath{\mathantt{e}}\xspace} 
\newcommand*{\mcc}{\ensuremath{\mathantt{m}}\xspace} 
\newcommand*{\pcc}{\ensuremath{\mathantt{p}}\xspace}

\newcommand*{\proposal}{\ensuremath{\mathantt{K}}\xspace} 
\newcommand*{\trans}{\ensuremath{\mu}\xspace}
\newcommand*{\ellr}{\ensuremath{{\mathantt{r}_\ell}}\xspace} 
\newcommand*{\kr}{\ensuremath{{\mathantt{r}_k}}\xspace} 
\newcommand*{\dimens}[1]{\ensuremath{\text{dim}\left(#1\right)}\xspace}
\newcommand*{\scalethis}[1]{\ensuremath{\scalebox{0.7}{\ensuremath{#1}}}\xspace}
\newcommand*{\distribution}[1]{\ensuremath{\mathcal{P}(#1)}\xspace} 
\newcommand*{\distributionemp}{\ensuremath{\mathcal{P}}\xspace} 
\newcommand*{\gibbsrel}{\ensuremath{\mathfrak{R}}}

\makeatletter
\newcommand*{\defeq}{\ensuremath{=}\xspace}
\makeatother
\newtheorem*{problem}{Problem}{\bf}{\rm}
\newtheorem{algo}{Algorithm}{\bf}{\rm}
\newtheorem{myremark}{Remark}{\bf}{\rm}
\newtheorem{definition}{Definition}{\bf}{\rm}
\newtheorem{theorem}{Theorem}{\bf}{\rm}
\newtheorem{lemma}{Lemma}{\bf}{\rm}
\newtheorem*{heuristic}{Greedy}{\bf}{\rm}
\newtheorem*{ILP_problem}{ILP}{\bf}{\rm}
\newtheorem{corollary}{Corollary}{\bf}{\rm}
\begin{document}

\def\spacingset#1{\renewcommand{\baselinestretch}%
{#1}\small\normalsize} \spacingset{1}

%%%%%%%%%%%%%%%%%%%%%%%%%%%%%%%%%%%%%%%%%%%%%%%%%%%%%%%%%%%%%%%%%%%%%%%%%%%%%%
	$\quad$\\
	{\centering\Huge\bfseries
  		Bayesian Changepoint Analysis\par
  	}
    \vspace{4cm}
    {\centering\Large\bfseries
	Tobias Siems\\\bigskip\small	
	Department of Mathematics and Computer Science\\\medskip	
	University of Greifswald
	
	\vfill\small\mdseries
	A thesis submitted for the degree of\\\medskip
	\emph{Doctor of Natural Sciences}\\\medskip
	October 25, 2019\par
	}
  
	\vspace{4cm}

\thispagestyle{empty}
\newpage\null\thispagestyle{empty}\newpage
\pagenumbering{Roman} 
\spacingset{1.45} % DON'T change the spacing!
\sloppy
\newpage
\tableofcontents 
\newpage\null\thispagestyle{empty}\newpage
\pagenumbering{arabic} 
\section{The \cp universe (a poetic view)}

\Cp problems unite three worlds: 
the world of counting, the world of placement and the world of change.
These worlds constitute universes, which receive their rough shape from statistical models applied to data and are further refined through parameters.

Even though the defining elements of such universes are grasped so quickly, we can never fully understand them as a composite whole.
Like a physicist, we need to rely on lenses revealing tiny pictures of the incredible truth, such as the astonishing glimpse into the counting world unveiled by Figure \ref{fig:cp_art}.

\begin{figure}[ht]
	\begin{center}
	\includegraphics[width=0.9\linewidth]{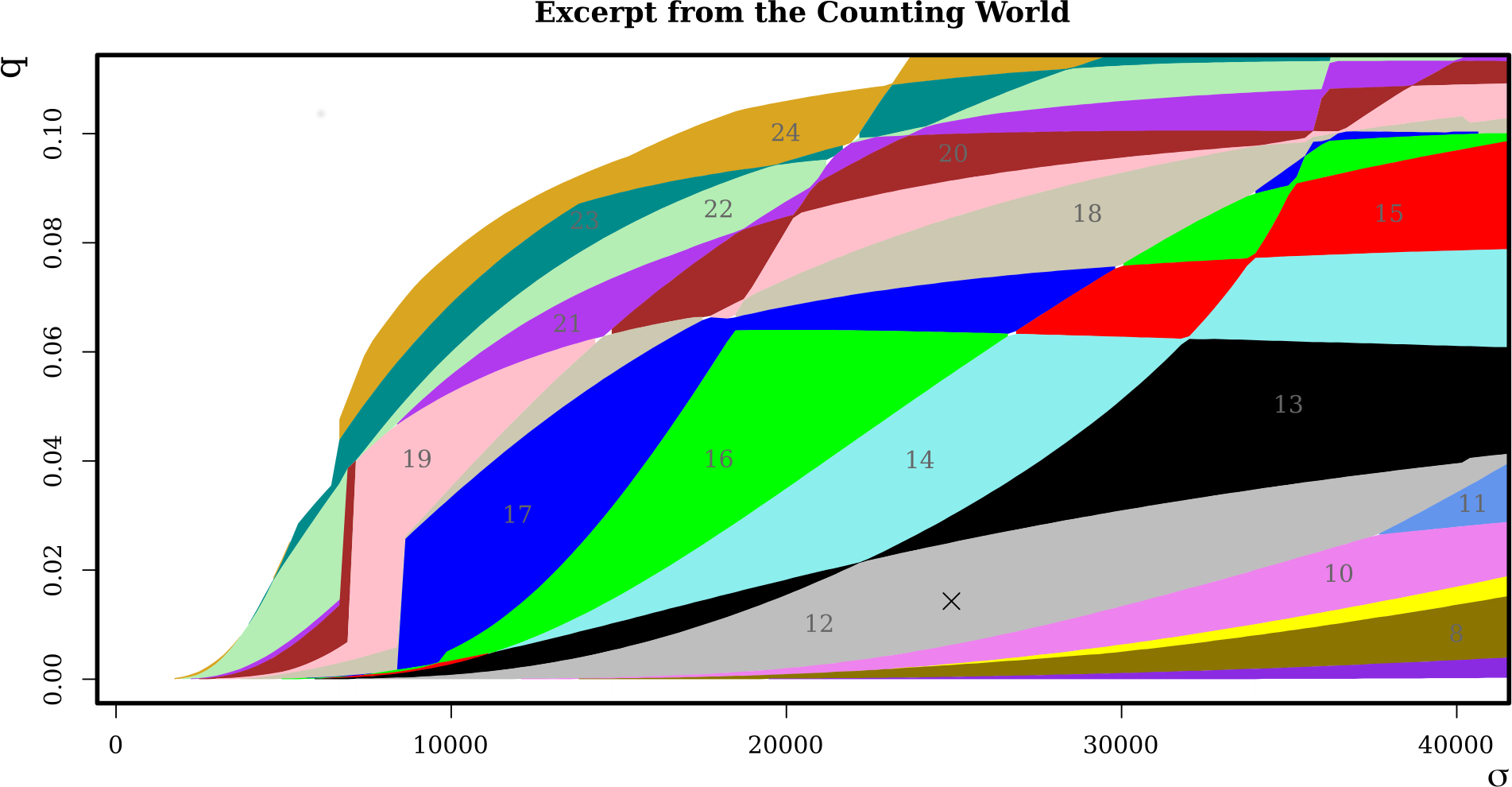}
	\caption[Excerpt from the counting world]{Data-driven partitioning of the parameter space of a \cp model. Each of the colored regions stands for a certain \cp count, obtained by computing the most likely \cp configuration.
		$q$ is the success probability of the length distribution of the segments and $\sigma$ the scale of the data within those segments (see also Section \ref{chap:well-log-para}).
	}
	\label{fig:cp_art}
		\end{center}
\end{figure}

This thesis gives you the unique opportunity to gain an understanding of the fantastic scale of Bayesian \cp universes.
To this end, we will look through a range of enlightening lenses, forged in the mighty fires of mathematics and operated on a highly sophisticated calculator, to discover places no one has known before.

\newpage

\section{Introduction}
\label{chap:intro}

Changepoint analysis deals with time series where certain characteristics undergo occasional changes.
Having observed such a time series, the latent \cp process can be investigated in various ways.
The following four examples provide an insight into the tasks that may arise in practice.

The examination of ion channels is a popular example for \cp analysis.
Ion channels are proteins and part of the cell membrane.
Their purpose is to control the flow of ions, like potassium, into and out of the cell.
New drugs must pass extensive tests for ion channel activity \citep{camerino2007ion}.
Depending on the type of the channel, they can be activated and deactivated through chemical reactions or electrical currents.

\begin{figure}[ht]
	\begin{center}
			\includegraphics[width=0.7\textwidth]{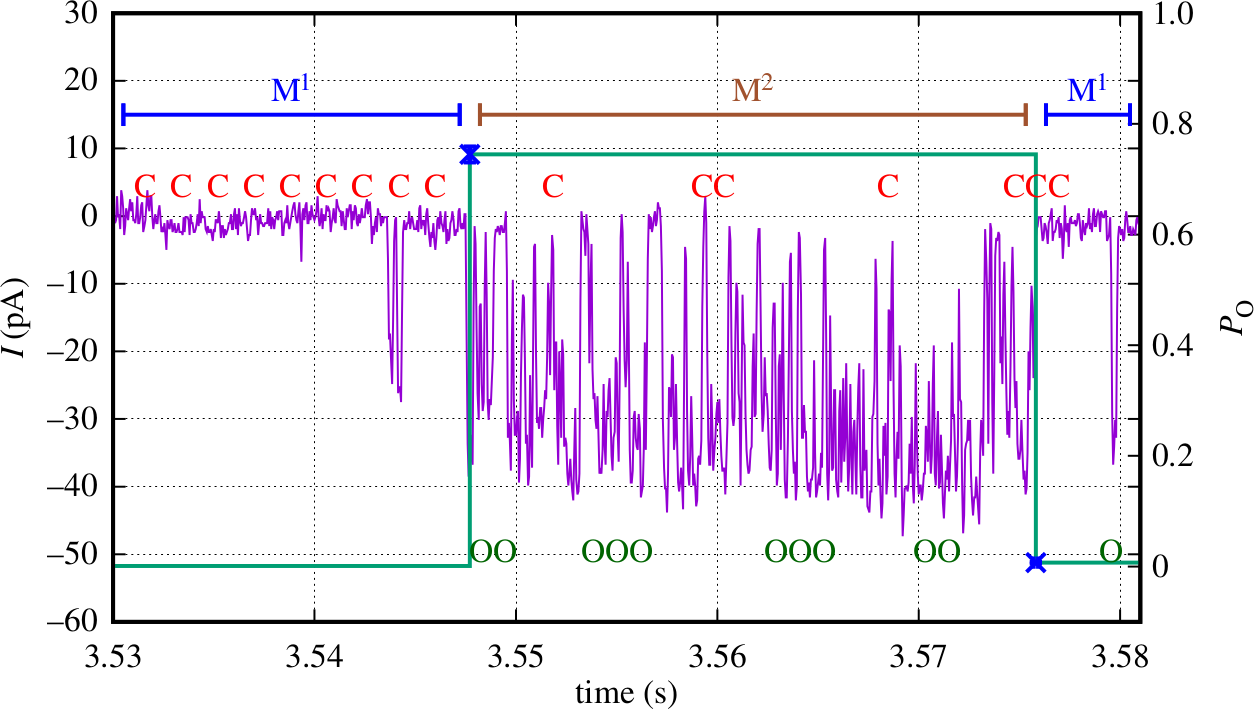}
	\end{center}
	\caption[Ion channel data]{Ion channel data \citep{siekmann2016modelling}.
	}
	\label{fig:ion_channel}
\end{figure}

A Nobel prize wining technique to gather data from single ion channels is the patch clamp technique  \citep{sakmann1984patch}.
It allows to record the electrical current that is generated by the flow of ions.
Figure \ref{fig:ion_channel} shows a data set gathered from an inositol trisphosphate receptor, which controls the flow of calcium ions \citep{siekmann2016modelling}.
The purple lines represent the measured currents, the C's and O's mark the times where the channel was closed or opened, respectively.
This is called gating and can be understood as a \cp process.
However, there is a second superordinate \cp process indicated by $\text{M}^1$ and $\text{M}^2$.
It models the modes of the channel's activity.

These modes reflect the stochastic nature of how the cell controls the flow of ions.
Thus, it is of high interest to find and understand the possible modes an ion channel takes.
Each mode can be represented through a time continuous Markov chain that describes its gating behavior \citep{siekmann2016modelling}.
The transitions between the modes may again be modeled as a Markov chain leading to a hierarchical model.
A challenging problem hereby is the amount of quantities which are to be estimated in order to determine the model and the associated non-identifiability.
Inference in these models may be drawn by means of approximate sampling through Markov chain Monte Carlo (MCMC) \citep{siekmann2011mcmc, siekmann2012mcmc, siekmann2014statistical, siekmann2016modelling}.

\begin{figure}[ht]
	\includegraphics[width=\linewidth]{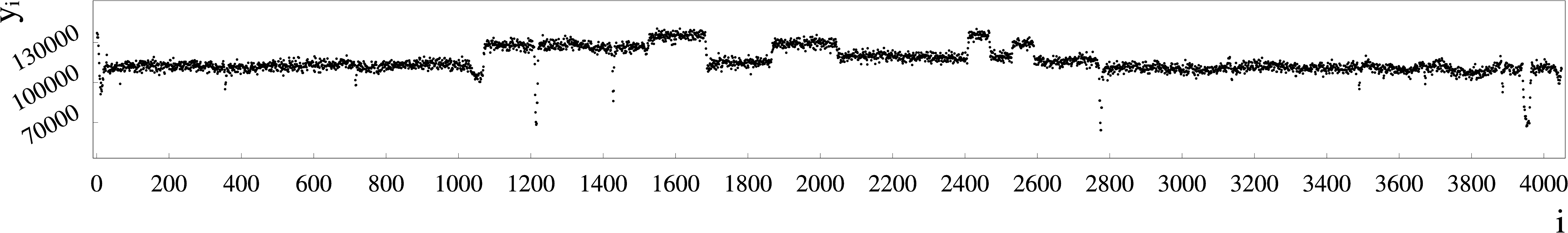}
	\caption[Well-log data]{Well-Log data \citep{exact_fearnhead}.
	}
	\label{fig:well-log}
\end{figure}

A different example for \cp analysis, which arises in geology, is the well-log drilling data in Figure \ref{fig:well-log}. 
It shows the nuclear-magnetic response of underground rocks \citep{exact_fearnhead}.

In order to build a \cp model here, we may regard coherent parts of the data as noisy versions of one and the same state and allow these states to change occasionally.
This can be understood as a change in location problem whereby the states are interpreted as the location parameter of the datapoints, e.g. median or mean.
Sequences of related datapoints are then considered to stem from the same material.

Interestingly, there are challenging perturbations in form of heavy outliers and smaller abrupt, but also gradual location changes.
These irregularities are prone to be confused with the desired location changes, which gives rise to robust \cp inference \citep{exact_fearnhead, fearnhead2019changepoint, weinmannl1}. 
The task here is to maintain a good balance between sensitivity and specificity.

This exciting data set serves as the main demonstration example throughout this thesis.

\begin{figure}[ht]
	\includegraphics[width=\linewidth]{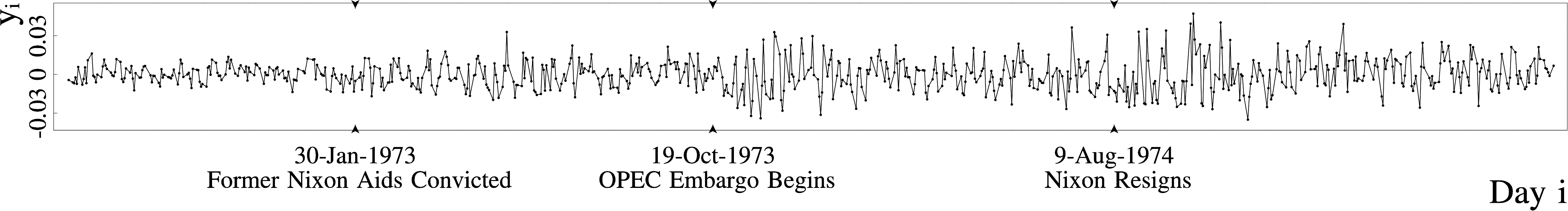}
	\caption[Dow Jones data]{Daily Dow Jones returns between 1972 and 1975 \citep{ocpd}.
	}
	\label{fig:dow_jones_intro}
\end{figure}

Stock markets pose another important \cp application. 
Consider for example Figure \ref{fig:dow_jones_intro}. 
It depicts daily Dow Jones returns observed between 1972 and 1975 \citep{ocpd} and highlights three events on January 1973, October 1973 and August 1974.

This can be understood as a change in variation problem, whereby the data is split into regions of similar spread.
\Cp analysis may be used for a real time detection of high risk or to understand the effect of political decisions on the economy.

\begin{figure}[ht]
	\begin{center}
	\includegraphics[width=0.6\linewidth]{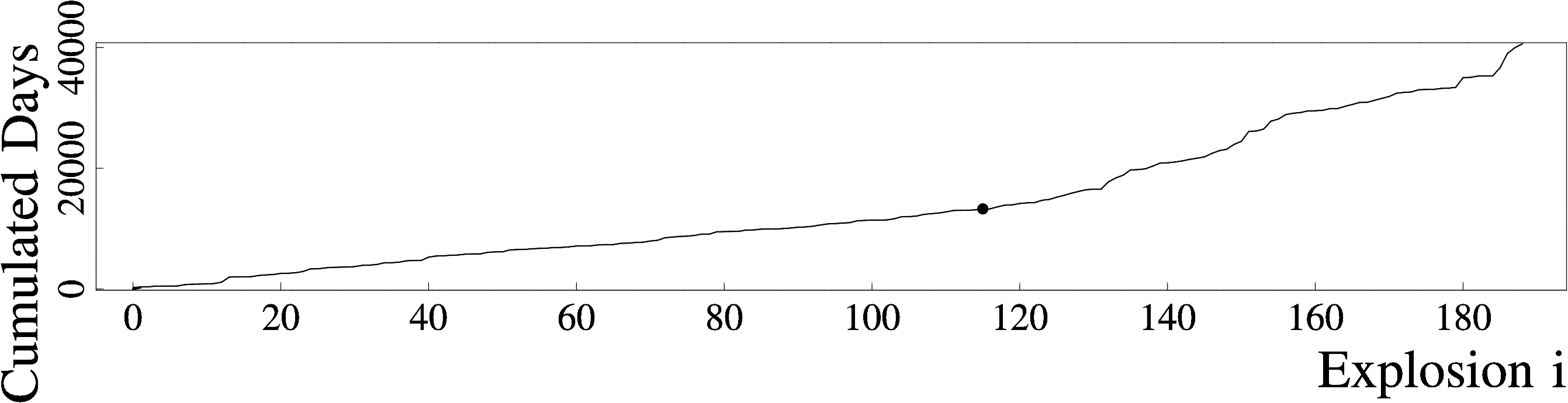}		
	\end{center}
	\caption[Coal mine data]{Coal mine explosions killing ten or more men between March 15, 1851 and March 22, 1962 \citep{jarrett1979note, ocpd}.
	The abscissa shows the explosion count and the ordinate the number of days until the explosions took place.
		The black dot marks the Coal Mines Regulations Act in 1887.
	}
	\label{fig:coal_intro}
\end{figure}

Our last example concerns coal mine explosions between March 15, 1851 and March 22, 1962 that killed ten or more people.
Figure \ref{fig:coal_intro} counts up the days from one explosion to the next.
The black dot points to the Coal Mines Regulations Act in 1887, which also roughly marks a change in the period of time it took from one explosion to the next.

In a \cp approach, the days between consecutive explosions could be modeled as geometrically distributed and changes in success rate would be examined.
The task here is to elaborate the effect of the regulations act by detecting or rejecting one single change.
This kind of single \cp approach can be applied to a large variety of applications where we expose if something has generally changed.
%However, our focus rather lies in multiple \cp problems where the existence of \cps is not the central question.

For the time being, this should be enough to motivate \cp analysis.
However, for further elaborations of different \cp approaches see, for example, \cite{eckley2011analysis}.

\subsection{Time series, segmentations and \cps}

Time series are ordered sequences of data.
This order can, for example, account for successive loci on a DNA sequence or indeed the times the observations were made and allows us to speak of segmentations.

Let $\obs$ be a set and $y=(y_1,\ldots,y_n)\in\obs^n$ be a time series with indexes $1,\ldots,n$.
Any sequence, $(y_j,\ldots,y_i)$ with $1\leq j\leq i\leq n$, of successive datapoints is considered as a \emph{segment} with \emph{segment length} $i-j+1$.
A \emph{segmentation} of $y$ is a collection of non-overlapping segments that covers $y$ as a whole.

A segmentation is uniquely determined by a number $k\in\{0,\ldots,n-1\}$ and a sequence of successive indexes $1<\tau_1<\ldots<\tau_k\leq n$ as follows
\begin{align*}
(y_1,\ldots,y_{\tau_1-1}),(y_{\tau_1},\ldots,y_{\tau_2-1}),\ldots,(y_{\tau_{k-1}},\ldots,y_{\tau_k-1}),(y_{\tau_{k}},\ldots,y_n)
\end{align*}
whereby $k=0$ indicates that the whole time series corresponds to one segment.
In the context of \cps, $k$ represents the number of \cps and the $\tau_1,\ldots,\tau_k$ are the \emph{\cp locations} or simply \emph{\cps}.
For now, \cps at timepoint one are excluded.
However, this will be relaxed later on.

As mentioned earlier, the segmentation features changes in certain characteristics of the time series, e.g. location or scale.
These characteristics are expressed through parameters, the so-called \emph{segment heights} that take values within a set, say $\lat$.
Consequently, we assign each segment of a segmentation a value within $\lat$.
%These values may represent a location parameter, a scale parameter and so forth.
In practice, a change in segment height typically triggers a \cp.

In non-trivial cases, the segment heights are latent and algorithms for finding \cps in a time series trace possible timepoints where they change their value.
%Nevertheless, if appropriate, a \cp approach can do without segment heights or introduce \cps even when the segment heights remain unchanged.
In doing so, they strive for persistency but also selectivity or as the statistician would put it, they govern a trade-off between specificity and sensitivity.
While a too high sensitivity yields too many \cps overfitting the segmentation to the data, a too high specificity disregards valuable \cps and impairs the \cp detection.

The existing literature captures approaches to detect at most one \cp as well as multiple \cps.
Even though single \cp methods can be applied recursively to detect multiple \cps \citep{fryzlewicz2014wild}, in this thesis, we focus solely on techniques that surmise multiple \cps in the first place.

\subsection{\Cp estimation via minimizing a penalized cost}
\label{chap:pencost}

A natural approach for finding multiple \cps in a time series is to assign each segmentation a rating and then select one of the best ones.
This rating should express the goodness of fit, but also regularize the number of \cps to avoid overfitting.
However, instead of speaking about goodness of fit and regularization, it is common to use opposite terms like \emph{cost} and \emph{penalization}.
Thus, \cps are inferred by \emph{minimizing a penalized cost}.

For the sake of computational convenience, the costs and penalizations are often defined individually for the segments instead of the whole segmentation at once.
This can be achieved through a \emph{segmental cost} $\cost:\bigcup_{i=1}^n\obs^i\rightarrow \mathbb{R}_{\geq 0}$ and a penalization constant $\gamma\geq 0$.
For the time being, we may ignore the segment heights.

Minimizing the penalized cost then yields the optimization problem which finds one element of
\begin{align}
\label{eq:pencost}
\argmin\limits_{1<\tau_1<\ldots<\tau_k\leq n\text{ and } k\in\mathbb{N}}\left\{\cost\big(y_{1:\tau_1-1}\big)+\Bigg(\sum_{i=1}^{k-1}\cost\big(y_{\tau_i:\tau_{i+1}-1}\big)\Bigg)+\cost\big(y_{\tau_k:n}\big)+\gamma\cdot k\right\}
\end{align}
This can be solved through a Viterbi algorithm \citep{viterbi1967error}, also known as Bellman recursion \citep{bellman1966dynamic} or dynamic programming.
In oversimplified terms, the algorithm performs a forward search from 1 to $n$, whereby in each iteration $i$ it performs a backward search in order to evaluate a best candidate for the most recent \cp prior to $i$.
Subject to the condition that $\cost$ can be augmented by new observations in constant time, i.e. $\cost(v_{1:i})$ can be derived from $\cost(v_{1:i-1})$ in constant time for all $v_{1:i}\in\obs^i$, this yields a space and time complexity of $\mcal{O}(n^2)$.
A long series of papers explains this algorithm in detail (see for example \cite{jackson_optimal, friedrich} and \cite{pelt}).

Furthermore, an approach developed by \cite{friedrich} computes solutions to (\ref{eq:pencost}) for all possible $\gamma$ at once.
It does so, by exploiting the linearity of the penalized cost w.r.t. $\gamma$.
Almost identical ideas are provided by \cite{haynes2017computationally}.
%If suitable \cp configurations are available beforehand, we could try to pick a $\gamma$ that yields one of these \cp configurations.
%Such an estimation procedure for $\gamma$ can be referred to as supervised learning, in contrast to unsupervised learning.

A discovery that earned considerable popularity is described in \cite{pelt}.
It allows under some circumstances to reduce the expected space and time complexity of solving (\ref{eq:pencost}) to $\mcal{O}(n)$.
We briefly discuss the idea behind it.

Assume that $\cost(v_{1:j})+\cost(v_{j+1:i})\leq \cost(v_{1:i})$ for any sequence $v_{1:i}\in\obs^i$ and $1 \leq j< i$.
This makes sense, because splitting a segment usually results in a tighter fit and thus a lower cost.
Let further $\mcal{F}_i$ be a solution to (\ref{eq:pencost}) w.r.t. $y_{1:i}$, i.e.
\begin{align*}
\mcal{F}_i\in\quad \argmin\limits_{1<\tau_1<\ldots<\tau_k\leq i\text{ and } k\in\mathbb{N}}\Bigg\{\cost\big(y_{1:\tau_1-1}\big)+\Bigg(\sum_{\ell=1}^{k-1}\cost\big(y_{\tau_\ell:\tau_{\ell+1}-1}\big)\Bigg)+\cost\big(y_{\tau_k:i}\big)+\gamma\cdot k\Bigg\}
\end{align*}

If for $1\leq j < i\leq n$ we get $\mcal{F}_j+\cost(y_{j+1:i})\geq \mcal{F}_i$, we can take an arbitrary $\ell> i$ and get 
\begin{align*}
\mcal{F}_j+\cost(y_{j+1:\ell})+\gamma\geq \mcal{F}_j+\cost(y_{j+1:i})+\cost(y_{i+1:\ell})+\gamma\geq \mcal{F}_i+\cost(y_{i+1:\ell})+\gamma%\geq \mcal{F}_\ell
\end{align*}
This states that placing the most recent \cp prior to $\ell$ at $i$, is at least as good as placing it at $j$.
Thus, we may skip $j$ in any backward search past $i$.
An implementation of the corresponding algorithm called \algor{PELT}, can be found in the R-Package \textit{\cp} \citep{pelt}.

The starting point for setting up the segmental cost is usually another function 
$\funcg:\bigcup_i\obs^i\times \lat\rightarrow\mathbb{R}_{\geq 0}$,
which takes the data within a segment and the segment height.
$\funcg$ can, for example, be derived as the negative loglikelihood of a statistical model.
However, M-estimators are also viable \citep{friedrich, fearnhead2019changepoint}.

We may put $\funcg$ directly into (\ref{eq:pencost}) and additionally minimize over the individual segment heights or equivalently, use $\cost(\LargerCdot)=\min_{x\in\lat}\funcg(\LargerCdot,x)$.
The solution can then be interpreted as a piecewise constant function with values corresponding to the optimal segment heights.
Consequently, we can speak of a regression method here.

In contrast, the Bayesian statistician may integrate over the segment heights in order to get rid of them completely.
This involves a more complex approach since we need to put a prior on the segment heights and perform an integration for each considered segment.

\cite{li2012multiple} make use of priors and run a genetic algorithm, which detects multiple \cps by minimizing a functional similar to the penalized cost.
Another non-Bayesian method to infer multiple \cps through a dynamic programming algorithm can be found in \cite{frick_munk_sieling}.

\subsection{Bayesian \cp models}
\label{chap:baycpmodel}

By applying priors, we pave the way for an interpretation of segmentations as random.
This allows for the treatment of the data, \cps and segment heights within the same methodological realm of stochastics.
As a result, a more complex modeling approach is required here.

To begin with, the \cps are represented by continuous or discrete random times, $(\rand T_i)_{i\in\mathbb{N}_+}$ over $\mathbb{R}_{\geq 0}$.
The \cp process starts at a certain timepoint, say $t_0\in\mathbb{R}$, whereby the first \cp is $t_0+\rand T_1$, the second $t_0+\rand T_1+\rand T_2$ and so forth.
Consequently, $\rand T_i$ represents the length of the $i$-th segment.

Furthermore, the segment heights are modeled accordingly through a family of random variables $(\rand \theta_i)_{i\in \mathbb{N}_+}$ over a measurable space $(\lat,\latsigma)$, whereby $\rand \theta_i$ represents the $i$-th segment height.
Finally, the observations constitute a family $(\rand Y_i)_{i=1,\ldots,n}$ over a measurable space $(\obs,\obssigma)$.
Each observation $\rand Y_i$ was measured at a timepoint $t_i$, whereby $t_0\leq t_1<\ldots<t_n$.

An in some ways minimal requirement for a statistical model to be expedient, is the feasibility of its likelihood.
This enables the use of MCMC techniques like Metropolis-Hastings and thus, pioneers an abundance of inference strategies.
However, before we can set up the likelihood, we need to consider an issue that typically results from the data collection.

In practice, $t_0$ often remains unknown.
It is thus reasonable to build the \cp model solely within $[t_1,t_n]$.
Consequently, we interpret the first segment as a clipped random length with a yet to be determined distribution. 
Given the first \cp within $[t_1,t_n]$, the second segment length will follow the same distribution as the $\rand T_i$'s and so forth.
As soon as a segment exceeds $t_n$ the process ends irrespective of its concrete length.

We may account for the lack of knowledge about $t_0$ by considering the limit $t_0\rightarrow -\infty$.
To this end, assume that $(\rand T_i)_{i\in\mathbb{N}_+}$ is a renewal process \citep{doob1948renewal} and independent of all other random variables within the \cp model.
Renewal theory \citep{lawrance1973dependency} teaches us that the distribution of the random length from $t_1$ to the first \cp exhibits the density
$\prob{t_1+\rand T_1> t}\slash\expec{\rand T_1}$ w.r.t. the Lebesgue measure.
These partial lengths are referred to as residual times or forward recurrence times.

The case where $-\infty<t_0<t_1$ turns out to be more challenging since the distribution of the residual times now depends on $t_0$, and can be very difficult to compute especially in the continuous case.

There are two important exceptions here.
The residual times of an exponential renewal process are independent of $t_0$ and correspond to the same exponential distribution again.
Similarly, for a geometric renewal process that exhibits \cps only on a known discrete grid, we can again take the same geometric distribution for the residual times.
%Note that this results from a discretization.

A similar issue arises with the segment heights.
A not too restrictive assumption is that $(\rand \theta_i)_{i\in \mathbb{N}_+}$ is described by a Markov chain.
If $t_0<t_1$, related theory may be applicable to derive the distribution of the segment height of the segment that contains $t_1$.

Having considered this, we can now  set up the joint likelihood of the data, \cp locations and segment heights in order to use MCMC methods like Metropolis-Hastings to perform approximate inference.

\subsection{The purpose of this thesis and its structure}

In this thesis, we elaborate upon Bayesian \cp analysis, whereby our focus is on three big topics: approximate sampling via MCMC, exact inference and uncertainty quantification.
Besides, modeling matters are discussed in an ongoing fashion.
Our findings are underpinned through several \cp examples with a focus on the aforementioned well-log drilling data.

In the following, you find four strongly interrelated, but also widely self-contained sections.
While Section \ref{chap:mcmcfinite} provides an understanding of the mathematical foundations of MCMC sampling, Section \ref{chap:transdim} considers the Metropolis-Hastings algorithm particularly for \cp related scenarios.
Both sections do not intend to develop completely new methods.
Instead, Section \ref{chap:mcmcfinite} explains important convergence theorems and gives very basic and simple proofs.
Furthermore, our elaboration of the Metropolis-Hastings algorithm aims at clarifying the apparently difficult topic of sampling within mixtures of spaces.

Section \ref{chap:inference} talks extensively about exact inference in Bayesian \cp models with independence induced by \cps.
It develops the bulk of existing algorithms, like exact sampling, and builds novel, very efficient algorithms for pointwise inference and an EM algorithm.
At the same time, new scopes for applying approximate \cp samples are opened up.

Finally, Section \ref{chap:simcredreg} develops a notion of credible regions for multiple \cp locations.
They are built from either exact or approximate \cp samples and illustrated by a new type of plot.
This plot greatly facilitates uncertainty quantification, model selection and \cp analysis.
The implementation of our credible regions gives rise to a novel NP-complete optimization problem, which is approached either exactly by solving an ILP or through a fast and accurate greedy heuristic.

The Sections \ref{chap:transdim} and \ref{chap:simcredreg} begin with an outline to summarize their purposes.
Besides, each of the following four sections provides an introductory part, which puts its topic into context with existing research and further clarifies the actual subject.
Sections \ref{chap:transdim}, \ref{chap:inference} and \ref{chap:simcredreg} conclude with a discussion.
Section \ref{chap:conclusion} concludes this whole thesis with final remarks.

Section \ref{chap:contributions} lists related papers that I co-authored and considerable contributions of other researchers to this thesis.

References to the supplementary material can be found in Section \ref{chap:supplement}.
Section \ref{chap:appendix} contains the appendix.
At the very end, we provide lists of symbols, figures, tables, pseudocodes and references.

\newpage\null\newpage
\section{Markov chain Monte Carlo on finite state spaces}
\label{chap:mcmcfinite}

In this section, we elaborate the idea behind Markov chain Monte Carlo (MCMC) methods in a mathematically coherent, yet simple and understandable way.
Therefore, we prove a pivotal convergence theorem for finite Markov chains and a minimal version of the Perron-Frobenius theorem.
%Subsequently, we briefly discuss two fundamental MCMC methods, the Gibbs and Metropolis-Hastings sampler.
Only very basic knowledge about matrices, convergence of real sequences and probability theory is required here.
A convenient summary of analogues results for general state spaces can be found in \cite{tierney1994markov}.

MCMC techniques aim at drawing samples from prespecified distributions. 
They do so in an indirect, approximate fashion through Markov chains.
This is important since the distributions deployed in practice are often too complex to be dealt with directly or even unavailable in closed form.

There exists a tremendous number of scientific articles and books about MCMC.
See, for example, \cite{bishop2014} for a vivid and more comprehensive introduction without mathematical proofs.

Let $\pi$ be a probability distribution over a finite state space $\states$ and $\pi_s\in[0,1]$ the probability of state $s\in\states$ under $\pi$.
By virtue of the strong law of large numbers, independent samples from $\pi$ (so-called \emph{i.i.d. samples}) can be used universally to approximate expectations w.r.t. $\pi$.
Thus, for a set of such samples, say $x_1,\ldots,x_m\in\states$, and an arbitrary function $\funcf:\states\rightarrow \mathbb{R}^\ell$, we get that $\frac{1}{m}\sum_{i=1}^m \funcf(x_i)\approx\sum_{s\in\states}\pi_s\funcf(s)$.
This simple recipe poses one of the most powerful tools in statistics.

An example for $\funcf$ is the indicator function $\ind{s\in A}$ for events $A\sub\states$.
It is one if the condition in the brackets is true and zero otherwise.
Its expectation yields the probability of $A$.

%An approach to gain i.i.d. samples from $\pi$ is the \emph{inversion method}.
%Therefore, assume w.l.o.g. that $\states=\{1,\ldots,k\}$.
%We generate a uniform random number $u$ over $[0,1]$ and obtain the desired sample as $\argmin_{\ell}\left\{\sum_{i=1}^\ell\pi_{i}\geq u\right\}$.

A \emph{Markov chain} over $\states$ is defined through an (arbitrary) initial state $s_0\in \states$ and a \emph{transition kernel} $\trans$.
$\trans$ is a non-negative function over $\states\times\states$ such that $\sum_{s\in \states}\trans_{z s}=1$ for all $z\in\states$.
$\trans_{z}$ can be interpreted as a conditional distribution.

The Markov chain starts in state $s_0$ and evolves according to $\trans$ in an iterative fashion:
the distribution of the first link in the chain is $\trans_{s_0 }$ and given the first link, say $s_1$, the distribution of the second link is $\trans_{s_1 }$ and so forth. 
This results in a potentially infinite sequence of random variables $\rand X=(\rand X_0, \rand X_1,\rand X_2,\ldots)$, whereby $\rand X_n$ represents the $n$-th link in the chain.
Consequently, we get that $\prob{\rand X_n=s\mid \rand X_{n-1}=z}=\trans_{zs}$ for $n>0$ and $\prob{\rand X_0=s_0}=1$.

Later on, we will deal with the (unconditional) distributions of the $n$-th links.
To this end, w.l.o.g. we assume that $\states=\{1,\ldots,k\}$.
Therewith, we describe $\trans$ as a matrix in $[0,1]^{k\times k}$ and $\pi$ as a column vector in $[0,1]^k$.
Any such quadratic matrix with non-negative entries and rows that sum to one is called \emph{stochastic matrix}.
Therewith, $\prob{\rand X_n=s}=(\trans^n)_{s_0s}$, i.e. the $n$-fold matrix product of $\trans$ evaluated at $s_0$ and $s$.
A further generalization is to set the law of $\rand X_0$ to an arbitrary distribution $\pi$. 
This yields $\prob{\rand X_n=s}=(\pi^t\trans^n)_{s}$, whereby the $t$ indicates transposition.

We say that a distribution $\pi$ is an \emph{invariant distribution} of $\trans$ if $\pi^t\trans=\pi^t$. %or equivalently if $\sum_{z\in\states}\pi_z\trans_{zs}=\pi_s$ for all $s\in\states$.
Thus, transitioning according to $\trans$ doesn't affect $\pi$.
Once a link in a Markov chain with transition kernel $\trans$ follows the law $\pi$, all subsequent links will do likewise.
In this case, the chain is considered to be in \emph{equilibrium}.
Equilibrium can be enforced by setting the distribution of $\rand X_0$ to $\pi$, but it may also be reached (approximately) in the long run through convergence.

The foundation of MCMC sampling is that under some circumstances the distributions of the $n$-th links of a Markov chain converge towards an invariant distribution regardless of the initial state. 
Thus, by simulating such a chain until equilibrium is reached sufficiently, we may obtain an approximate sample of this invariant distribution.
On these grounds, MCMC methods provide schemes to build Markov chains with a predefined unique invariant distribution.

%The Gibbs sampler \citep{Geman1984} builds a Markov chain with invariant distribution $\pi$ by decomposing $\pi$ into simpler conditional versions.
%This facilitates sampling of complex joint distributions, but is somewhat restricted in its ability to explore the state space.
%
%
%The well-known Metropolis-Hastings algorithm \citep{metropolis1953, hastings1970} is capable of incorporating user defined proposal distributions, which
%enables the exploration of the state space in any desired fashion.
%It further facilitates the application of complex statistical models to observed data.

\subsection{Convergence towards and existence of invariant distributions}

This section first considers the convergence of the distributions of the $n$-th links of certain Markov chains.
This convergence forms the very basis of MCMC sampling.
Afterwards, we provide a version of the Perron-Frobenius theorem, which gives further insights into the existence of invariant distributions.
The following two properties of $\trans$ play a fundamental role.

$\trans$ is called \emph{irreducible} if for every $z,s\in\states$ there exists an $n>0$ such that $(\trans^n)_{zs}>0$.
Thus, regardless of the state the Markov chain starts in, every state can eventually be reached with positive probability.

$\trans$ is called \emph{aperiodic} if there exists an $N$ such that $(\trans^n)_{zs}>0$ for all $n>N$ and all $z,s\in\states$.
Thus, regardless of the state the Markov chain starts in, in the long run, it can always reach any state in one step with positive probability.

The name irreducible suggests that the Markov chain does not divide $\states$ into separate, mutually inaccessible classes.
In turn, aperiodicity excludes the case where parts of $\states$ are reached in a periodic fashion, for example, only through either an even or odd number of transitions.
Aperiodicity implies irreducibility, but not vice versa.

It is obvious that a periodic behavior may impede convergence of the distributions of the $n$-th links.
For a general convergence theorem, aperiodicity is thus a necessary condition.
We will now see that it is also sufficient.
The following theorem is a simplified version of a convergence theorem for Markov chains over countable state spaces provided inter alia by \cite{konig2005stochastische}.

\begin{theorem}
	\label{lem:invdisc}
	For an aperiodic stochastic matrix $\trans$ with invariant distribution $\pi$, we get that $\limi{n}(\trans^n)_{zs}= \pi_s$ for all $z,s\in\states$.
\end{theorem}
\begin{proof}
	Assume that $\rand X=(\rand X_0, \rand X_1, \ldots)$ is a Markov chain with transition kernel $\trans$ that starts with an arbitrary but fixed state $s_0$.
	Furthermore, consider the Markov chain $\rand{Z}=(\rand Z_0, \rand Z_1,\ldots)$ with transition kernel $\trans$ and initial distribution $\pi$, i.e. $\prob{\rand Z_0=s}=\pi_s$ for all $s\in\states$.
	$\rand X$ and $\rand Z$ are supposed to  be independent of each other.

	Let $\rand T$ be the random variable that represents the first $n$ where $\rand X$ and $\rand Z$ equal, i.e. $\rand T=\min\{n\mid \rand X_n=\rand Z_n\}$.
	We want to show that $\rand T$ is finite with probability one.
	Due to the aperiodicity of $\trans$, we may choose an $N$ such that $\trans^N$ has solely positive entries. 
	Let $\epsilon>0$ be the smallest entry of $\trans^N$ and consider that
	\begin{align*}
	&\prob{nN< \rand T}=\prob{\rand X_{i}\not=\rand Z_{i}\text{ for all } i\leq nN}\leq	\prob{\rand X_{i\cdot N}\not=\rand Z_{i\cdot N}\text{ for all } i\leq n}\\
	&=\sum_{z_0\not=s_0}\pi_{z_0}\sum_{s_1\in\states}\trans^N_{s_0s_1}\sum_{z_1\not=s_1}\trans^N_{z_0z_1}\ldots
%	\sum_{z_2\in\states}\trans^N_{z_1z_2}\sum_{s_2\not=z_2}\trans^N_{s_1s_2}\ldots
	\sum_{s_n\in\states}\trans^N_{s_{n-1}s_n}\underbrace{\sum_{z_n\not=s_n}\trans^N_{z_{n-1}z_n}}_{\leq 1-\epsilon}\leq (1-\epsilon)^n\overset{n\rightarrow \infty}{\longrightarrow}\ 0
	\end{align*}

	Define further $\rand Y_n=\ind{n\leq \rand T}\rand X_n+\ind{n>\rand T}\rand Z_n$.
	The Markov chain $\rand Y=(\rand Y_0, \rand Y_1, \ldots)$ starts with copying $\rand X$ and switches to $\rand Z$ as soon as both equal the first time.
	We are interested in the distribution of $\rand Y$.
	To this end, consider an arbitrary path $s_{1:n}\in\states^{n}$, define $\pcc_{s_{j:i}}=\prod_{\ell=j}^{i} \trans_{s_{\ell-1} s_\ell}$ for $j,i=1,\ldots,n$ and observe that
	\begin{align*}
	&\prob{\rand Y_{0:n}=s_{0:n}}=\prob{\rand Y_{0:n}=s_{0:n}, n< \rand T}+\sum_{\ell=0}^n\prob{\rand Y_{0:n}=s_{0:n}, \rand T=\ell}\\
	&=\pcc_{s_{1:n}}\prob{\rand Z_{0:n}\neq s_{0:n}}
	+\sum_{\ell=0}^n\pcc_{s_{1:\ell}}\prob{\rand Z_\ell=s_\ell, \rand Z_{0:\ell-1}\neq s_{0:\ell-1}}\pcc_{s_{\ell+1:n}}=\pcc_{s_{1:n}}
	\end{align*}
	whereby we used that $\prob{\rand Z_{0:n}\neq s_{0:n}}+\sum_{\ell=0}^n\prob{\rand Z_\ell=s_\ell, \rand Z_{0:\ell-1}\neq s_{0:\ell-1}}=1$.
	This shows that $\rand Y$ is a Markov chain with transition kernel $\trans$ and initial state $s_0$.

	Altogether, we may state that 
	\begin{align*}
	(\trans^n)_{s_0s}&=\prob{\rand Y_n=s}=\prob{\rand Y_n=s, n\leq \rand  T}+\prob{\rand Y_n=s, n>\rand T }\\
	\pi_s&=\prob{\rand Z_n=s}=\prob{\rand Z_n=s, n\leq \rand T}+\prob{\rand Y_n=s, n>\rand T}
	\end{align*} 
	and
	$|(\trans^n)_{s_0s}-\pi_s|=|\prob{\rand Y_n=s, n\leq \rand T }- \prob{\rand Z_n=s, n\leq \rand T}|\limii{n}{\infty} 0$ for all $s_0\in\states$.
\end{proof}

%We say $\rand{X}$ is in equilibrium at $n$ if $\prob{\rand{X}_n=i}=\pi_i$ is satisfied for all $i\in\states$.

Given a distribution $\pi$, MCMC methods seek an aperiodic transition kernel $\trans$ with invariant distribution $\pi$.
Thus, it is possible to sample approximately from $\pi$ by simulating a Markov chain with transition kernel $\trans$ until equilibrium is reached to a sufficient extend.
The last link within this chain is then taken as a single approximate sample from $\pi$.
In particular, this procedure is independent of the state the Markov chain has started in.
The pace by which equilibrium is approached is referred to as the \emph{mixing time}.

Now, we consider a version of the well-known Perron-Frobenius theorem \citep{frobenius1912matrizen}.
It is usually stated in a more general context and corresponding proofs can be fairly complicated.
In turn, we provide our own convenient proof based on simple arithmetics and matrix algebra.

\begin{theorem}[Perron-Frobenius Theorem]
\label{lem:irredmat}
An irreducible transition kernel $\trans$ has a unique invariant distribution $\pi$.
\end{theorem}
\begin{proof}
Since any stochastic matrix has a right eigenvector with corresponding eigenvalue 1, it also has such a left eigenvector.
In particular, any such left 1-eigenvector exhibits non-zero elements.
Let $x$ be a left 1-eigenvector of \trans.
If $x$ has only non-negative or non-positive entries, we can immediately derive an invariant distribution $\pi$ of $\trans$ through normalizing $x$, i.e. $\pi=x\big\slash \sum_{s\in\states} x_s$.

Assume now that $x$ exhibits positive entries for $s\in \bar N$ and negative entries for $s\in N$.
The following applies
\begin{align}
\nonumber
&\sum_{z\in N}x_z \trans_{zs}+\sum_{z\in \bar N}x_z \trans_{zs}=x_s
\ \ \Rightarrow\ \ 
\sum_{z\in N}x_z \sum_{s\in\bar N}\trans_{zs}+\sum_{z\in \bar N}x_z \sum_{s\in \bar N}\trans_{zs}=\sum_{s\in \bar N}x_s\\
\label{eq:perr}
\Leftrightarrow\ \ &\underbrace{\sum_{z\in N}x_z \sum_{s\in\bar N}\trans_{zs}}_{\leq 0}
=\sum_{s\in \bar N}x_s-\sum_{z\in \bar N}x_z\sum_{s\in \bar N}\trans_{zs}
=\sum_{z\in \bar N}x_z\left(1-\sum_{s\in \bar N}\trans_{zs}\right)
=\underbrace{\sum_{z\in \bar N}x_z \sum_{s\not \in \bar N}\trans_{zs}}_{\geq 0}
\end{align}
Hence, the l.h.s and r.h.s. of (\ref{eq:perr}) have to be zero, which implies that $\trans_{zs}=0$ for all $z\in \bar N$ and $s\not\in \bar N$.
Consequently, $(\trans^n)_{zs}=0$ for all $z\in \bar N$, $s\not\in \bar N$ and $n\in\mathbb{N}$.

Since the existence of positive and negative entries implies reducibility, we conclude that irreducibility implies that any left 1-eigenvector has either solely non-positive or non-negative entries.
Thus, an irreducible transition kernel $\trans$ exhibits at least one invariant distribution $\pi$.

Finally, assume that there is a second invariant distribution $\pi'\neq \pi$.
In order to represent a distribution, not all components of $\pi$ can be either larger or smaller than the components of $\pi'$.
Thus, $\pi-\pi'$ must have positive and negative entries.
However, $\pi-\pi'$ is a left 1-eigenvector of $\trans$ and thus, $\trans$ can't be irreducible, which contradicts the existence of $\pi'$.
\end{proof}

The Perron-Frobenius theorem shows that invariant distributions can certainly be found for an abundance of stochastic matrices, especially for aperiodic ones.
In the context of MCMC, it is, however, only a nice-to-have result and not utterly necessary.
In fact, there is great freedom in choosing aperiodic transition kernels that exhibit a prespecified invariant distribution and each MCMC method provides its very own approach to do so.

\newpage
\section{A note on the Metropolis-Hastings acceptance probabilities for mixture spaces}
\label{chap:transdim}

\subsection{Outline}

This work is driven by the ubiquitous dissent over the abilities and contributions of the Metropolis-Hastings and reversible jump algorithm within the context of transdimensional sampling.
In a mathematically coherent and at times somewhat journalistic elaboration, we aim to demystify this topic by taking a deeper look into the implementation of Metropolis-Hastings acceptance probabilities with regard to general mixture spaces.

Whilst unspectacular from a theoretical point of view, mixture spaces gave rise to challenging demands concerning their effective exploration.
An often applied but not extensively studied tool for transitioning between distinct spaces are so-called translation functions.

We give an enlightening treatment of this topic that yields a generalization of the reversible jump algorithm and unveil a further important translation technique.
Furthermore, by reconsidering the well-known Metropolis-within-Gibbs paradigm, we come across a dual strategy to develop Metropolis-Hastings sampler.
We underpin our findings and compare the performances of our approaches by means of a \cp example.
Thereafter, in a more theoretical context, we revitalize the somewhat forgotten concept of maximal acceptance probabilities.
This allows for an interesting classification of Metropolis-Hastings algorithms and gives further advice on their usage.
A review of some errors in reasoning that have led to the aforementioned dissent concludes this whole section.

\subsection{Introduction}

%The foundation of Markov chain Monte Carlo (MCMC) sampling is that under some circumstances Markov chains converge towards their invariant distribution regardless of the initial state \citep{tierney1994markov}. 
%In practice, this convergence is approximated by simulating the chain for an appropriate time.
%Therefore, MCMC methods provide schemes to build Markov kernels with a desired invariant distribution which is typically an intractable marginal or conditional, or highly multivariate distribution.

The Metropolis-Hastings algorithm is one of the most well-known MCMC methods.
It traverses through the state space by means of a user defined proposal distribution.
Each proposed state undergoes an accept-reject step, which decides whether the proposed state or the previous link in the chain is chosen to be the next link.
This step alone secures the invariance of the Metropolis-Hastings Markov kernel towards the target distribution.

%A primal version was first published in \cite{metropolis1953} and then extended in \cite{hastings1970}.
%\cite{tierney1994markov} combines the latest findings about Markov chains and their convergence, and puts them into context with Metropolis-Hastings and MCMC in general.
%%\cite{geyer1994} elaborates the Metropolis-Hastings algorithm for spatial point processes.
%\cite{green1995} presents the reversible jump algorithm to sample across spaces of different dimensions. 
%\cite{tierney1998} generalizes the Metropolis-Hastings algorithm in a reaction to a hype on trans dimensional sampling.
%\cite{roodaki2011} addresses mixtures of proposals.

%There is a tremendous number of scientific articles and books about MCMC available.
%I recommend \cite{bishop2014} for a practical introduction.
%My writings about MCMC were mainly inspired, in accordance to this ordering, by \cite{bishop2014, konig2005stochastische, green1995, tierney1994markov, tierney1998, roodaki2011, higdon1998, godsill2001, mixkalman}.

Originally, Metropolis-Hastings proposals were designed conveniently through kernel densities w.r.t. the same measure that runs the density of the target distribution \citep{metropolis1953equation, hastings1970}.
In this case, the acceptance probability used in the accept-reject step is determined by the likelihood ratio of the transition in backward and forward direction.
However, upcoming applications like variable selection \citep{spike_slab_mitchel}, point processes \citep{geyer1994simulation} and \cp analysis \citep{exact_fearnhead, siems2019simultaneous} raised higher demands.

Common for these applications is that the elements of the state spaces are inhomogeneous in their dimension, i.e. the spaces are transdimensional.
Obviously, such mixtures of different spaces are not straightforwardly accessible through standard techniques like random walk proposals.
The first methods to conduct a change in dimension were plain births and deaths \citep{geyer1994simulation}.
However, this is prone to ignore the relations shared among several coordinates and therewith promotes poor acceptance rates.
Consequently, the exploration of the entire state space performs differently within the same and across the dimensions impairing the overall mixing time.

Subsequently, it became utterly popular to utilize functions that translate between points of different dimensions.
\cite{green1995} pioneered this approach through the reversible jump algorithm.
It was developed for purely continuous spaces and describes a particular class of proposals that act detached from the target space.
As a result, the density of the target distribution and the kernel density of the proposal do not necessarily share the same underlying measure anymore.
Nevertheless, \cite{green1995} was able to derive correct acceptance probabilities through a nifty application of the change of variable theorem.

According to Google scholar, \cite{green1995} has been cited over 5000 times.
Unfortunately, \cite{green1995} and others caused a misperception about the abilities of the Metropolis-Hastings algorithm, which is nowadays ubiquitous.
It is agreed among a significant number of scientific writings that the reversible jump algorithm has somehow made trans dimensional sampling possible and that other MCMC algorithms like Metropolis-Hastings are not applicable in these scenarios.
\cite{green1995}, \cite{carlin1995bayesian} and \cite{chen2012} even substantiate this wrong conclusion, which was acknowledged by \cite{besag2001markov, waagepetersen2001tutorial, green2003trans, green2009reversible, sisson2005transdimensional, sambridge2006trans} and many others.
A thorough search for papers which oppose this stance indirectly, in one way or another, revealed only \cite{geyer1994simulation, tierney1998note, andrieu2001model, godsill2001relationship, jannink2004metropolis} and \cite{roodaki2011}.

This poses a significant division within the MCMC community, which requires a thorough treatment.
To this end, we take up on the original task of exploring mixture spaces by means of the Metropolis-Hastings paradigm, whereby the very focus lies on the computation of acceptance probabilities.
At the appropriate places, we will investigate the results of \cite{tierney1998note} and \cite{roodaki2011} and make conclusions beyond that.
Furthermore, the discussion gives a more detailed consideration of the aforementioned claims of \cite{green1995, carlin1995bayesian} and \cite{chen2012}.
%Unfortunately, the outshining popularity of the reversible jump algorithm makes this inspiring statistical area seem a little monotonous.

We will see that the foundation of the computation of acceptance probabilities in its most theoretical to most practical form is the compatibility of the involved densities towards a common underlying measure.
For example, the generalized Metropolis-Hastings algorithm from \cite{tierney1998note}, which is applicable to virtually any combination of target distribution and proposal, derives the required densities through the Radon-Nikodym theorem \citep{nikodym1930generalisation}.
Unfortunately, the nonconstructive nature of this theorem yields only little practical value.

Therefore, subject to the aforementioned compatibility, the design of proposals comprises a trade-off between mutability and feasibility.
As a consequence, each of the following methods exhibits its very own conditions and possibilities that need to be considered carefully in order to build the right Metropolis-Hastings sampler.

Similar to \cite{green1995}, we examine so-called translation functions, which convey between pairs of spaces.
However, we consider two different ways to apply them: before and after proposing new states and refer to these concepts as ad-hoc and post-hoc translations, respectively.

In accordance with the reversible jump algorithm, the main feature of post-hoc translations is the detachment of the proposal from the space the target distribution acts on.
This entails an integral transformation which imposes strict conditions and requirements on the translation functions.
Consequently, a sophisticated design in practical and mathematical terms is essential here.

In return, post-hoc translations support parsimony regarding the number of random variables that need to be generated.
This enables very tightly adapted and even deterministic transitions.
Furthermore, they grant access to certain proposals defining intractable compound distributions like convolutions, marginal distributions or factor distributions.

In contrast, ad-hoc translations support arbitrary translation functions, as long as they act on the mixture space exclusively.
We may, for example, deploy partial maximizers of the likelihood.
This represents a straightforward way to improve the acceptance rates of difficult transitions.

On the downside, ad-hoc translations require that the proposal kernel densities comply with the underlying measure space.
Thus, the freedom in choosing  translation functions comes at the prize of a smaller adaptability in terms of the generation of random states.

Whilst the name ad-hoc results from the ability to use simple and purposive translation functions, the term post-hoc refers to the moment when the translation function is applied.

%The aforementioned compatibility is established virtually by lifting the involved measures and densities into the mixture space.

%Both translation approaches are concerned with the aforementioned compatibility issues.
%While post-hoc translations achieve this by means of an integral transformation, ad-hoc translations just require the proposal kernel densities and measures to be lifted into the mixture space.
%While the former relies on higher calculus or measure theory, the latter usually remains invisible to the user.
%
%
%This gives rise to a dual approach that detects the measure space first and then deploys a plain proposal kernel density on this space. 

The well-known Metropolis-within-Gibbs approach confines the set of possible transitions through conditioning 
and is therewith able to steer the exploration of the state space in specific ways.
Consequently, it employs conditional densities defined on suitable measure spaces.
It turns out that proposals set up directly on these measure spaces yield straightforward acceptance probabilities.

This shifts the challenges of creating a Metropolis-Hastings sampler to deriving conditional densities of the target distribution.
We consider this in some way novel and innovative perspective to be dual to the traditional one that puts the definition of proposals first.

Metropolis-within-Gibbs can be applied in the usual way to reduce the number of  coordinates that need to be considered and is therewith capable of relaxing the terms of ad-hoc translations.
In addition, it captures far more difficult cases like so-called semi deterministic translations, an important kind of post-hoc translations.

Our methods are scrutinized by means of a \cp example.
For this sake, we compare the acceptance rates of several birth and death like proposals.
It turns out that ad-hoc proposals (combined with Metropolis-within-Gibbs) or post-hoc proposals are vital here.
While the post-hoc translations are derived from delicate properties of the target distribution, the ad-hoc translations just trace for high likelihood values.
Despite their strong differences, both methods achieve comparable acceptance rates.

The notion of maximal acceptance probabilities as introduced by \cite{peskun1973optimum} and \cite{tierney1998note} doesn't seem to have gained a lot of attention so far.
The reason for this might be that there is usually only one choice for the acceptance probability in place.
We revitalize this dusted property and show that post-hoc translations do not necessarily yield maximal acceptance probabilities.
Therewith, they may sometimes differ from their maximal counterpart.

Transitions between pairs of different spaces are usually performed by separate proposals, which are combined to a mixture proposal.
This allows for a clear modular design and is thus pursued primarily here.

As stated in \cite{tierney1998note} already, acceptance probabilities can either be computed based on unique pairs of the individual proposals or from the mixture proposal as a whole.
Albeit unintentionally, \cite{roodaki2011} elaborates conditions upon when both approaches yield maximality.
We will investigate their main result and put them into our context.

This work is structured as follows.
At first, we consider the Metropolis-Hastings approach in applied terms in Section \ref{chap:mh}.
This involves an introduction of the important detailed balance condition together with the primal Metropolis-Hastings algorithm.
Furthermore, in Section \ref{chap:mixspaces} we talk about mixture spaces, mixture proposals and translations.
Section \ref{chap:metgibs} is concerned with the Metropolis-within-Gibbs approach.
Subsequently, we elaborate a \cp example in Section \ref{chap:cpexample}.
The general consideration of Metropolis-Hastings and maximal acceptance probabilities is pursued in Section \ref{chap:genmh}.
Section \ref{chap:mixtures} further transfers these observations to mixture proposals.
We conclude with a discussion and a brief literature review in Section \ref{chap:discussion_trans}.

\subsection{The Metropolis-Hastings algorithm in applied terms}
\label{chap:mh}

In this section, we deal with practical implementations of the Metropolis-Hastings algorithm.
To this end, we develop ways to design proposals based on kernel densities w.r.t. sigma finite measure spaces.
These comprise so-called ad-hoc and post-hoc translations and a new perspective for the Metropolis-within-Gibbs framework.

The Metropolis-Hastings algorithm is usually introduced on the basis of densities w.r.t. a common measure.6
For this sake, we are given a sigma finite measure space $(\states,\mcal{A},\lambda)$ and a probability measure $\pi$ with a density $\dens$ such that $\pi(ds)=\dens(s)\lambda(ds)$.
The user provides a Markov kernel $\proposal$ from $(\states,\mcal{A})$ to $(\states,\mcal{A})$
in form of a kernel density $\densk(s', s)$ such that $\proposal(s', ds)=\densk(s',s)\lambda(ds)$.
This Markov kernel is referred to as the \emph{proposal} and $\densk$ as the \emph{proposal kernel density}.

\begin{algo}[Metropolis-Hastings]
	\label{alg:mh}
	(I) Choose an initial state $s_0\in\states$ (II) In step $k=1,2,\ldots.$, given the previous state $s_{k-1}=s'$, propose a new state $s$ according to $\proposal(s', \LargerCdot)$ and set $s_k=s$ with probability
	\begin{align}
	\label{eq:acc}
	\acc_{s' s}=\min\Bigg\{1, \frac{\dens(s)\densk(s,s')}{\dens(s')\densk(s',s)}\Bigg\}
	\end{align}
	otherwise $s_k=s'$. We set $\acc_{s' s}$ to zero wherever it is not defined.
\end{algo}
We refer to (\ref{eq:acc}) as the \emph{acceptance probability} and to the second argument within the braces as the \emph{acceptance ratio}.

Let further
\begin{align}
\label{eq:mhtranskern}
\mu(s',\LargerCdot)=\int \acc_{s' s}\delta_{s}(\LargerCdot)+\left(1-\acc_{s's}\right)\delta_{s'}(\LargerCdot)\proposal(s', ds)
\end{align}
$\mu$ represents the Markov kernel whose application constitutes step (II) of Algorithm \ref{alg:mh}.	
Consequently, executing Algorithm \ref{alg:mh} corresponds to simulating successive links of a Markov chain with initial state $s_0$ and Markov kernel $\mu$.

The crucial point here is that, subject to some conditions, the unconditional distributions of these links converge in total variation to $\pi$ \citep{tierney1994markov}.
Thus, after simulating the chain for a sufficiently long time, the last link can be regarded as an approximate sample of $\pi$.

In this work, we mainly focus on the most basic of these conditions, the invariance of $\pi$ towards $\mu$.
This invariance can be accomplished through the acceptance probability alone.
To this end, we will examine different families of proposals and how closed solutions for the acceptance probabilities are calculated.
The other conditions to secure convergence rely on the specific form of $\mu$ and have to be checked rather individually from case to case \citep{tierney1994markov}.

The required invariance proofs are facilitated by means of the following condition.
\begin{definition}
	The Markov kernel $\mu$ preserves the \emph{detailed balance} condition w.r.t. $\pi$ if
	\begin{align*}
	\pi\otimes \mu(A\times B)=\pi\otimes \mu(B\times A)
	\end{align*}
	for all $A,B\in\mcal{A}$. 
\end{definition}

If $\mu$ preserves the detailed balance condition w.r.t. $\pi$, $\pi$ is an invariant distribution of $\mu$ since $\pi\otimes \mu(\states\times \LargerCdot)=\pi\otimes \mu(\LargerCdot\times \states)=\pi$.
The opposite implication does not hold in general.

Markov chains build from Markov kernels that preserve the detailed balance w.r.t. another distribution are called \emph{reversible}.
This is because, they exhibit same probabilities in forward and backward direction once a link in the chain follows the law of this distribution.
Consequently, MCMC methods that preserve the detailed balance condition are also called reversible.

\begin{lemma}[\cite{hastings1970}]
	\label{lem:mh}
	The Markov kernel $\mu$ of Equation (\ref{eq:mhtranskern}) with the acceptance probability shown in (\ref{eq:acc})  preserves the detailed balance w.r.t. $\pi$.
\end{lemma}
\begin{proof}
	There is nothing to prove for $A,B\in\mcal{A}$ with $A=B$. 
	Since $A\times B$ is the disjoint union of $(A\cap B)\times (A\cap B)$, $(A\backslash B)\times (B\backslash A)$, $(A\backslash B)\times (A\cap B)$ and $(A\cap B)\times (B\backslash A)$ we can w.l.o.g. assume that $A\cap B=\emptyset$.
	We get
	\begin{align*}
	&\pi\otimes \mu(A\times B)=\int_A \mu(s',B)\pi(ds')=\int_A \int_B \acc_{s' s}\densk(s',s)\dens(s')\lambda(ds)\lambda(ds')\\
	&\overset{*}{=}\int_B \int_A  \acc_{s s'}\densk(s,s')\dens(s)\lambda(ds')\lambda(ds)=\int_B \mu(s,A)\pi(ds)=\pi\otimes \mu(B\times A)
	\end{align*}
	whereby we have used that $\acc_{ s's}\densk(s',s)\dens(s')=\acc_{s s'}\densk(s,s')\dens(s)$ and Fubini's theorem in *.
\end{proof}

By looking at $\acc_{s' s}$, we see that normalization constants w.r.t. $\dens$ cancel out.
Thus, the algorithm may deal with conditional versions of $\dens$ by deploying $\dens$ unchanged.
This greatly facilitates data processing by means of conditioning.

\subsubsection{Mixture state spaces and translations}
\label{chap:mixspaces}
Now, we elaborate sampling across alternating measurable spaces with the help of functions that translate between these spaces.
We consider two naturally arising approaches, the so-called ad-hoc and post-hoc translations.
Both exhibit their very own challenges and advantages in terms of feasibility and adaptability.

Let $\spaces$ be a finite or countable set and let $(\states_i,\mcal{A}_i, \lambda_i)_{i\in\spaces}$ be a family of disjoint sigma finite measure spaces.
The \emph{mixture measure space} $(\states,\mcal{A}, \lambda)$ of $(\states_i,\mcal{A}_i, \lambda_i)_{i\in\spaces}$ is defined through 
\begin{align*}
&\states=\bigcup_{i\in\spaces}\states_i,\quad \mcal{A}=\sigma\Bigg(\bigcup_{i\in\spaces}\mcal{A}_i\Bigg),\quad \lambda=\sum_{i\in\spaces}\lambda_i\circ \text{id}_i^{-1}
\end{align*}
whereby $\mcal{A}$ is the smallest sigma algebra that comprises all $\mcal{A}_i$'s.
The sole purpose of the identities $\text{id}_i:\states_i\rightarrow \states$ with $\text{id}_i(s)=s$ is to lift states from the component spaces into the mixture space.

As before, we assume that $\pi$ exhibits a density $\dens$ w.r.t. $\lambda$.
A Metropolis-Hastings algorithm for $\pi$ can now be build as usual.
A natural way to set up proposals on the mixture space is by pursuing a modular design through a mixture of Markov kernels, whereby each component solely transitions within a pair of the spaces.
We  want to use the term \emph{move} to feature the available transitions.

Let $\moves$ be a finite or countable set, i.e. the set of \emph{moves}.
In order to build a mixture proposal, we define measurable functions $\moveprob_\ell:\states\rightarrow[0,1]$ with $\sum_{\ell\in\moves}\moveprob_{\ell}(s')=1$ for all $s'\in\states$.
$\moveprob_{\ell}(s')$ is the probability of choosing move $\ell$ while being in $s'$.
Furthermore, each move $\ell$ exhibits a kernel density $\densk_\ell$ that transitions between a pair of the spaces and complies with the measure of the target space.

For the sake of technical correctness, we need to lift each $\densk_\ell$ into the mixture space by setting it to 0 for unsupported arguments.
Therewith, the mixture proposal $\proposal$ that operates on the mixture space $(\states,\mcal{A})$ reads
\begin{align*}
\proposal(s',ds)=\sum_{\ell\in\moves}\moveprob_\ell(s')\proposal_\ell(s',ds)
\end{align*}
with $\proposal_\ell(s',ds)=\densk_\ell(s',s)\lambda(ds)$.

This time, the transition kernel for the Metropolis-Hastings algorithm for mixtures reads
\begin{align}
\label{eq:mhtranskernmix}
\mu(s',\LargerCdot)=\sum_{\ell\in\moves}\moveprob_\ell(s')\int \acc_{s' s}^\ell\delta_{s}(\LargerCdot)+\left(1-\acc_{s's}^\ell\right)\delta_{s'}(\LargerCdot)\proposal_\ell(s', ds)
\end{align}
Given the previous link $s'$, this kernel involves choosing a move $\ell$ according to $\left(\moveprob_k(s')\right)_{k\in\moves}$ at first, then proposing an $s$ according to $\proposal_\ell(s',\LargerCdot)$ and finally accepting or rejecting $s$ through $\acc_{s's}^\ell$.

This is in accordance with Algorithm \ref{alg:mh}, however, with the difference that we do not consider the mixture proposal as a whole.
Instead, we incorporate the moves explicitly into the acceptance probability (see also Section \ref{chap:mixtures}).

\begin{lemma}
	\label{lem:mix}
	Let $\ellr:\moves\rightarrow\moves$ be a bijection.
	The transition kernel $\mu$ of (\ref{eq:mhtranskernmix}) with acceptance probability 
	\begin{align*}
	\acc_{s' s}^\ell=
	\min\Bigg\{1,\frac{\dens(s)\densk_{\ellr}(s,s')}{\dens(s')\densk_{\ell}(s',s)}\cdot\frac{\moveprob_{\ellr}(s)}{\moveprob_{\ell}(s')}\Bigg\}
	\end{align*}
	preserves the detailed balance w.r.t. $\pi$.
\end{lemma}
\begin{proof}
	For $A\in\mcal{A}_j$, $B\in\mcal{A}_i$, we get
	\begin{align*}
	&\sum_{\ell\in\moves}\int_A\int_B\acc_{s' s}^\ell\dens(s')\moveprob_{\ell}(s')\densk_{\ell}(s',s)\lambda_i(ds)\lambda_j(ds')\\
	&=\sum_{\ell\in\moves}\int_B\int_A\acc_{s s'}^\ellr\dens(s)\moveprob_{\ellr}(s)\densk_{\ellr}(s,s')\lambda_j(ds')\lambda_i(ds)\\
	&\overset{*}{=}\sum_{\ell\in\moves}\int_B\int_A\acc_{s s'}^\ell\dens(s)\moveprob_{\ell}(s)\densk_{\ell}(s,s')\lambda_j(ds')\lambda_i(ds)
	\end{align*}
	whereby in * we have used the uniqueness of $\ellr$. 
\end{proof}

$\ellr$ stands for the unique backward or reverse move of $\ell$.
Consequently, if move $\ell$ transitions from $\states_j$ to $\states_i$, move $\ellr$ should transition in the opposite direction from $\states_i$ to $\states_j$.	
Please note that the uniqueness of $\ellr$ doesn't imply that the transitions performed by move $\ell$ can only be reversed by move $\ellr$.
In fact, the pairs of spaces are not supposed to be unique to the moves and the possible transitions determined by related move pairs may overlap arbitrarily.

Concerning the design of the moves, it is often not clear how to transition away from a point in one space to a point in another space.
Especially random walk proposals pose a problem if there is no suitable measure of distance between the spaces available.
This gives rise to two distinct approaches that employ functions as a device for translation.

\begin{figure}[ht]
	\center
	\begin{tikzpicture}[scale=0.80]
	\tikzstyle{every node}=[draw,shape=ellipse, line width=1.1pt];
	\node (node1) at (150:3) {$\states_j$};
	\node (node2) at ( 30:3) {$\states_i$};
	\node (node3) at ( 270:0)[rectangle] {$\densk_{\ell}$};
	\node (node4) at ( 270:1.3)[] {$\states_i$};
	\node (node5) at ( 270:-3)[rectangle] {$\bo s_{\ell}$};
	\node[font=\large, draw=none,fill=none] at (270:-4) {(a)};

	\path[->] (node1.north) edge [out=50, in=180, line width=1.1pt] (node5.west);
	\path[->] (node5.east) edge [out=0, in=120, line width=1.1pt] (node2.north);
	\path[->] (node2.south) edge [out=250, in=0, line width=1.1pt] (node3.east);
	\path[->] (node3) edge [, line width=1.1pt] (node4.north);
	\end{tikzpicture}
	$\quad\quad\quad\quad$
	\begin{tikzpicture}[scale=0.80]
	\tikzstyle{every node}=[draw,shape=ellipse, line width=1.1pt];
	\node (node1) at (150:3) {$\states_j$};
	\node (node2) at ( 30:3) {$\ustates_{\ell}$};
	\node (node3) at ( 270:0)[rectangle] {$\bo s_{\ell}$};
	\node (node4) at ( 270:1.3) {$\states_i$};
	\node (node5) at ( 270:-3)[rectangle] {$\densk_{\ell}$};
	\node[font=\large, draw=none,fill=none] at (270:-4) {(b)};

	\path[->] (node1) edge [out=-20, in=120, line width=1.1pt] (node3);
	\path[->] (node1.north) edge [out=50, in=180, line width=1.1pt] (node5.west);
	\path[->] (node5.east) edge [out=0, in=120, line width=1.1pt] (node2.north);
	\path[->] (node2) edge [out=200, in=60, line width=1.1pt] (node3);
	\path[->] (node3) edge [, line width=1.1pt] (node4.north);
	\end{tikzpicture}
	\caption[Ad-Hoc and post-hoc transition graphs]{Transition graphs for a move $\ell\in\moves$ that transitions from $\states_j$ to $\states_i$ in the fashion of (a) an ad-hoc translation and (b) a post-hoc translation}
	\label{fig:mhrjcomp}
\end{figure}
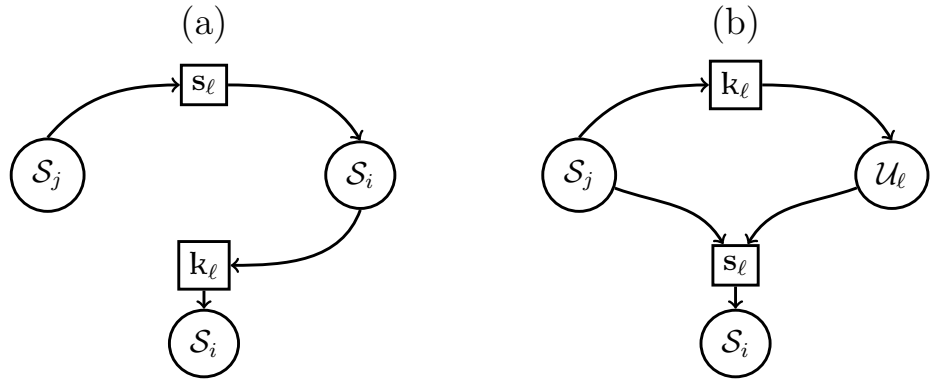

In order to transition from $\states_j$ to $\states_i$ in move $\ell$, our first approach employs a measurable function $\bo s_{\ell}:\states_j\rightarrow\states_i$ that translates between the spaces and a kernel density $\densk_\ell$ from $(\states_i,\mcal{A}_i)$ to $(\states_i,\mcal{A}_i, \lambda_i)$.
Functions like $\bo s_{\ell}$ are referred to as \emph{translation functions}.
Given the previous link in the chain $s'\in\states_j$, new states are proposed by virtue of $\proposal_\ell(s',\LargerCdot)=\int_\LargerCdot\densk_{\ell}(\bo s_{\ell}(s'), s)\lambda_i(ds)$.
We call this procedure \emph{ad-hoc translation}.
Figure \ref{fig:mhrjcomp}(a) illustrates it graphically.

\begin{corollary}
\label{cor:ad-hoc}
The transition kernel $\mu$ of (\ref{eq:mhtranskernmix}) with acceptance probability 
\begin{align*}
\acc_{s' s}^\ell=
\min\Bigg\{1,\frac{\dens(s)\densk_{\ellr}(\bo s_{\ellr}(s),s')}{\dens(s')\densk_{\ell}(\bo s_{\ell}(s'),s)}\cdot\frac{\moveprob_{\ellr}(s)}{\moveprob_{\ell}(s')}\Bigg\}
\end{align*}
preserves the detailed balance w.r.t. $\pi$ in ad-hoc translations.
\end{corollary}

Ad-hoc translations pose a flexible framework that is directly accessible through non-mathematicians with a basic understanding of the Metropolis-Hastings algorithm.
In practice, however, another technically more demanding Metropolis-Hastings method has become the quasi gold-standard for sampling within mixed spaces.
It is commonly referred to as the reversible jump algorithm \citep{green1995}.
We will allocate this approach to another type of translation and give a novel result that slightly generalizes this technique.

In contrast to ad-hoc translations, now, we propose first and then apply a translation to the proposed state.
This allows us to detach the target space from the space where the proposal kernel density acts on.
Thus, in move $\ell\in\moves$ we may transition between a pair of the spaces, say from $\states_j$ to $\states_i$, by detouring over another user-defined sigma finite measure space $(\ustates_{\ell},\mcal{B}_{\ell},\nu_{\ell})$.
To this end, we design a kernel density $\densk_{\ell}$ from $(\states_j,\mcal{A}_j)$ to $(\ustates_{\ell},\mcal{B}_{\ell},\nu_{\ell})$ and a measurable function $\bo s_{\ell}:\states_j\times \ustates_{\ell}\rightarrow \states_i$.
Given the previous link $s'\in\states_j$, a random state $u'\in \ustates_{\ell}$ is drawn by virtue of $\int_\LargerCdot\densk_{\ell}(s', u')\nu_{\ell}(du')$ and the newly proposed state $s\in\states_i$ is obtained as $s=\bo s_{\ell}(s',u')$.
We call this procedure \emph{post-hoc translation}.
Figure \ref{fig:mhrjcomp}(b) illustrates it graphically.

The case where $(\states_j,\mcal{A}_j)$ and $(\states_i,\mcal{A}_i)$ are discrete can be treated straightforwardly since the acceptance probability of move $\ell$ then reads
\begin{align}
\label{eq:accdiscpost}
\acc_{s' s}^\ell=
\min\Bigg\{1,\frac{\dens(s)\int\delta_{\bo s_\ellr(s,u)}\left(\{s'\}\right)\densk_{\ellr}(s,u)\nu_{\ellr}(du)}
{\dens(s')\int\delta_{\bo s_\ell(s',u')}\left(\{s\}\right)\densk_{\ell}(s',u')\nu_{\ell}(du')}\cdot\frac{\moveprob_{\ellr}(s)}{\moveprob_{\ell}(s')}\Bigg\}
\end{align}
This applies, because the proposals in forward and backward direction are summarized as discrete and are thus compatible to $\dens$.
Though, the required feasibility of the integrals in (\ref{eq:accdiscpost}) slightly restricts the set of possible post-hoc translations here.

However, most interesting cases arise when the proposal $\int \delta_{\bo s_\ell(s',u')}(\LargerCdot)\densk_\ell(s',u')\nu_{\ell}(du')$ either doesn't exhibit a kernel density w.r.t. $\lambda_i$ or when it is computationally infeasible.
Unfortunately, our methods so far do rely on such densities calling for a different approach here.

In practice, it would be desirable to derive the acceptance probabilities directly from the given proposal kernel densities.
Thus, the trick is to accept and reject $s=\bo s_{\ell}(s',u')$ within $\states_j\times\ustates_\ell$ yielding the actual post-hoc paradigm.
In order for this procedure to be well-defined, any transition $(s',u')\in \states_j\times\ustates_\ell$ must exhibit a unique backward transition $(s,u)\in \states_i\times\ustates_\ellr$.
Thus, we need to define additional measurable functions $\bo u_\ell:\states_j\times\ustates_\ell\rightarrow \ustates_\ellr$ with
\begin{align*}
(\bo s, \bo u)_\ell(s',u')=(s,u)\quad \Leftrightarrow\quad (\bo s, \bo u)_\ellr(s,u)=(s',u')
\end{align*}
for all $(s',u')\in\states_j\times\ustates_\ell$ and $(s,u)\in\states_i\times\ustates_{\ellr}$.

%As a relaxation of the above, we could narrow down the domain of $(\bo s, \bo u)_\ell$ to a subset of $\states_j\times\ustates_\ell$.
%However, move $\ell$ is only reasonable if this subset is not a null set w.r.t. $\pi\otimes\nu_{\ell}$, which excludes examples like constant post-hoc translations within continuous spaces.

This reasoning shows that the possible options for choosing post-hoc translations this way are somewhat limited.
Therewith, the seemingly restrictive conditions of the next lemma can be considered as fairly universal.

\begin{lemma}
	\label{lem:post-hoc}
	Assume that for move $\ell\in\moves$ with transitions from $\states_j$ to $\states_i$ and corresponding unique backward move $\ellr$, there exist measurable functions $\bo u_{\ell}:\states_j\times \ustates_{\ell}\rightarrow \ustates_{\ellr}$ and $f_{\ell}:\states_j\times \ustates_{\ell}\rightarrow\mathbb{R}_{\geq 0}$ with
	\begin{align*}
	&\textbf{\emph{(I)}}\ \ (\bo s, \bo u)_\ell=(\bo s, \bo u)_\ellr^{-1}\\
	&\textbf{\emph{(II)}}\ \ \lambda_i\otimes\nu_{\ellr}\left((\bo s, \bo u)_{\ell}(ds',du')\right)=f_{\ell}(s',u')\lambda_j\otimes\nu_{\ell}(ds', du')
	\end{align*}
	The transition kernel
	\begin{align*}
	\mu(s',\LargerCdot)=\sum_{\ell\in\moves}\moveprob_\ell(s')\int \acc_{s'u' }^\ell\delta_{\bo s_\ell}(\LargerCdot)+\left(1-\acc_{s'u'}^\ell\right)\delta_{s'}(\LargerCdot)\densk_\ell(s',u')\nu_{\ell}(du')
	\end{align*}
    with
	\begin{align*}
	\acc_{s'u'}^\ell=\min\left\{1,\ \frac{\dens(\bo s_{\ell})\densk_{\ellr}((\bo s, \bo u)_{\ell})}{\dens(s')\densk_{\ell}(s',u')}\cdot\frac{\moveprob_{\ellr}(\bo s_{\ell})}{\moveprob_{\ell}(s')}\cdot f_{\ell}(s',u')\right\}
	\end{align*}
	preserves the detailed balance w.r.t. $\pi$ in post-hoc translations.
\end{lemma}
\begin{proof}
	\textbf{(I)} implies that each $(\bo s, \bo u)_{\ell}:\states_j\times \ustates_{\ell}\rightarrow\states_i\times \ustates_{\ellr}$ is a bijection with inverse $(\bo s, \bo u)_{\ellr}$ and \textbf{(II)} guides through an integral transformation by means of the densities $f_{\ell}$.
	Since 
	\begin{align*}
	&\lambda_j\otimes\nu_{\ell}(ds', du')=\lambda_j\otimes\nu_{\ell}(\bo s_{\ellr}(\bo s_{\ell}(ds', du')))\\
	&=f_{\ell}(\bo s_{\ell}(ds', du'))\lambda_i\otimes\nu_{\ellr}(\bo s_{\ell}(ds', du'))=f_{\ell}(\bo s_{\ell}(s', u'))f_{\ellr}(s',u')\lambda_j\otimes\nu_{\ell}(ds', du')
	\end{align*}
	we see that $f_{\ell}(s',u')f_{\ellr}(s,u)=1$ if $\bo s_{\ell}(s', u')=(s,u)$ and vice versa.
	Furthermore, for $A\in\mcal{A}_j$ and $B\in\mcal{B}_{i}$, we get
	\begin{align*}
	\pi\otimes\proposal_\ellr(B\times A)
	&=\int\delta_{s}(B)\delta_{\bo s_{\ellr}}(A) \dens(s)\densk_{\ellr}(s,u)\lambda_i\otimes\nu_{\ellr}(ds,du)\\
	&=\int\delta_{\bo s_{\ell}}(B)\delta_{s'}(A) \dens(\bo s_{\ell})\densk_{\ellr}((\bo s, \bo u)_{\ell})\lambda_i\otimes\nu_{\ellr}((\bo s, \bo u)_{\ell}(ds',du'))\\
	&=\int \delta_{s'}(A)\delta_{\bo s_{\ell}}(B)\dens(\bo s_{\ell})\densk_{\ellr}((\bo s, \bo u)_{\ell})f_{\ell}(s',u')\lambda_j\otimes \nu_\ell(ds', du')
	\end{align*}
	The rest of the proof now follows the same scheme as before.	
%	\begin{align*}
%	\pi\otimes\mu(A\times B)&
%    =\sum\limits_{\ell\in\moves}\int \delta_{s'}(A)\delta_{\bo s_{\ell}}(B) \acc_{s'u'}^\ell\dens(s')\densk_{\ell}(s',u') \moveprob_{\ell}(s')\lambda_j\otimes \nu_{\ell}(ds',du')\\
%	&\overset{**}{=}\sum\limits_{\ell\in\moves}\int \delta_{\bo s_{\ellr}}(A)\delta_{s}(B) \acc^\ell_{(\bo s, \bo u)_{\ellr}}\dens(\bo s_\ellr)\densk_{\ell}((\bo s, \bo u)_{\ellr}) \moveprob_{\ell}(\bo s_{\ellr})\lambda_j\otimes\nu_{\ell}\left((\bo s, \bo u)_{\ellr}(ds,du)\right)\\
%	&\overset{*}{=}\sum\limits_{\ell\in\moves}\int \delta_{\bo s_{\ellr}}(A)\delta_{s}(B) \acc^\ell_{(\bo s, \bo u)_{\ellr}}\dens(\bo s_{\ellr})\densk_{\ell}((\bo s, \bo u)_{\ellr}) \moveprob_{\ell}(\bo s_{\ellr})f_{\ellr}(s,u)\lambda_i\otimes\nu_{\ellr}(ds,du)\\
%	&=\sum\limits_{\ell\in\moves}\int\delta_{s}(B)\delta_{\bo s_{\ellr}}(A) \acc^\ellr_{su}\dens(s)\densk_{\ellr}(s,u)) \moveprob_{\ellr}(s)\lambda_i\otimes\nu_{\ellr}(ds,du)\\
%	&=\sum\limits_{\ell\in\moves}\int\delta_{s}(B)\delta_{\bo s_{\ell}}(A) \acc^\ell_{su}\dens(s)\densk_{\ell}(s,u)) \moveprob_{\ell}(s)\lambda_i\otimes\nu_{\ell}(ds,du)=\pi\otimes\mu(B\times A)
%	\end{align*}
%	whereby we have substituted $(s',u')$ with $(\bo s,\bo u)_{\ellr}$ in ** and used \textbf{(II)} in *.
\end{proof}

In mathematically precise terms, it is incorrect to refer to $\densk_\ell$ as a proposal kernel density because it is generally not a kernel density of the proposal.
However, in order to maintain a consistent language, we still want to keep this naming.
Thus, the proof of Lemma \ref{lem:post-hoc} applies an integral transformation that secures the compatibility of the (joint) densities w.r.t. the transitions in forward and backward direction.

If all spaces are discrete coordinate spaces, we may choose arbitrary bijections $(\bo s, \bo u)_{\ell}$ that meet \textbf{(I)} and apply
\begin{align*}
\acc_{s'u'}^\ell=\min\left\{1,\ \frac{\dens(\bo s_{\ell})\densk_{\ellr}((\bo s, \bo u)_{\ell})}{\dens(s')\densk_{\ell}(s',u')}\cdot\frac{\moveprob_{\ellr}(\bo s_{\ell})}{\moveprob_{\ell}(s')}
\cdot \frac{\lambda_i\otimes\nu_{\ellr}((\bo s, \bo u)_{\ell}(ds', du'))}{\lambda_j\otimes\nu_{\ell}(ds',du')}\right\}
\end{align*}
\citep{fronk2004} deploy post-hoc translation in this way.

It is common to require that $\dimens{\states_j}+\dimens{\ustates_{\ell}}=\dimens{\states_i}+\dimens{\ustates_{\ellr}}$ if move $\ell$ transitions from $\states_j$ to $\states_i$.
\cite{green1995} pioneered the post-hoc translation approach and introduced this condition as \emph{dimension matching}.
His intention was to embed the spaces into a single superordinate one in order to avoid the inhomogeneity of the dimension.
Around that time the nowadays ubiquitous term \emph{transdimensional sampling} was coined.

The \emph{reversible jump algorithm} of \cite{green1995} employs post-hoc translations over purely continuous coordinate spaces.
It requires the dimension matching condition and that each $(\bo s, \bo u)_{\ell}$ is a diffeomorphism with Jacobi determinant $\bo J_{\ell}$.
If \textbf{(I)} is met, the resulting acceptance probability reads
\begin{align*}
\acc_{s'u'}^\ell=\min\left\{1,\ \frac{\dens(\bo s_{\ell})\densk_{\ellr}((\bo s, \bo u)_{\ell})}{\dens(s')\densk_{\ell}(s',u')}\cdot\frac{\moveprob_{\ellr}(\bo s_{\ell})}{\moveprob_{\ell}(s')}
\cdot |\bo J_{\ell}|\right\}
\end{align*}
This follows from an integration by substitution for multiple variables since $ds'du'|\bo J_{\ell}|=dsdu$ if we substitute $(s,u)=(\bo s, \bo u)_{\ell}(s', u')$.

Post-hoc translations are highly advanced, require very careful implementations and as it stands do only support certain kinds of state spaces.
In return, they give the opportunity to generate random states within spaces of lower dimension than the target space.
Therewith, they pave the way for tightly adapted, parsimonious sampling schemes where even deterministic transitions are viable.

Moreover, post-hoc translations may act on a space of higher dimension than the target space, which grands access to intractable compound distributions.
To see this, consider the following minimal example to ordinarily transition from $\mathbb{R}$ to $\mathbb{R}$ by proposing values in $\mathbb{R}^2$ through the kernel density $\densk$.
By employing the diffeomorphism $(s', u_1, u_2)\mapsto \left(u_1+u_2, s'-u_2, u_2\right)$ we end up with equal forward and backward moves and
\begin{align*}
\acc_{s', u_1, u_2}=\min\left\{1,\ \frac{\densk(u_1+u_2, s'-u_2, u_2)\dens(u_1+u_2)}{\densk(s',u_1, u_2)\dens(s')}\right\}
\end{align*}

The corresponding proposal on the target space, i.e. $\int \delta_{u_1+u_2}(\LargerCdot)\densk(s, u_1, u_2)du_1du_2$, is the convolution of $\densk(s, u_1, u_2)$ w.r.t $u_1$ and $u_2$.
If $\proposal$ exhibits an intractable density on $\mathbb{R}$, the post-hoc paradigm proves functional here.
Nevertheless, there is a little price to pay when we use the post-hoc paradigm this way.
We lose the maximality of the acceptance probability.
This topic will be discussed in Section \ref{chap:genmh}.

\subsubsection{Metropolis-within-Gibbs}
\label{chap:metgibs}

Imagine a coordinate state space, whereby each step of the Metropolis-Hastings algorithm solely modifies one single coordinate.
This evokes the Metropolis-within-Gibbs framework.
Here we investigate this approach further.
We will see that a generalization allows for the identification of new measure spaces to run proposals on.
This creates an innovative perspective for an important class of post-hoc translations.

Given a countable set of moves $\moves$, consider a family of measurable functions $(\funcg_\ell)_{\ell\in \moves}$ from $(\states, \mcal{A})$ into a measurable space $(\mcal{G}_\ell,\mcal{B}_\ell)$ with $\{g\}\in\mcal{B}_\ell$ for all $g\in\mcal{G}_\ell$.
Each $\funcg_\ell$ determines a partition of the state space $\states$ into measurable sets that yield the same values under $\funcg_\ell$. 
Given the previous link $s'\in\states$, the idea is to choose an $\ell$ out of $\moves$ and to propose a new state exclusively within the ensuing set 
\begin{align*}
\gibbsrel_\ell(s')=\{s\in\states\mid \funcg_\ell(s)=\funcg_\ell(s')\}
\end{align*}

This entails the use of conditional versions of $\pi$.
For this sake, assume that $\rand X$ is a random variable with $\rand X\sim \pi$ and $\prob{\rand X\in ds\mid \funcg_\ell(\rand X)=g}=\dens_\ell(g, s)\lambda_\ell(ds)$ with a measure $\lambda_\ell$ over $\mcal{A}$ and a kernel density $\dens_\ell$ from $(\mcal{G}_\ell,\mcal{B}_\ell)$ to $(\states,\mcal{A}, \lambda_\ell)$.
Furthermore, we stipulate that for all $s'\in\states$, the restriction of $\lambda_\ell$ to $\gibbsrel_\ell(s')$ is sigma finite.

On these grounds, we employ proposal kernel densities $\densk_\ell$ from $(\states,\mcal{A})$ to $(\states,\mcal{A}, \lambda_\ell)$ 
such that $\proposal_\ell(s',ds)=\densk_\ell(s',s)\lambda_\ell(ds)$
and $\proposal_\ell(s',\gibbsrel_\ell(s'))=1$ for all $\ell\in\moves$ and $s'\in \states$.
We further define move probabilities $\moveprob_\ell:\states\rightarrow[0,1]$ as before.

\begin{lemma}
	\label{lem:metgibbs}
    The transition kernel $\mu$ of (\ref{eq:mhtranskernmix}) with acceptance probability 
	\begin{align}
	\label{eq:metgibbsacc}
	\acc_{s' s}^\ell=
	\min\left\{1,\frac{\dens_\ell(g, s)\densk_\ell(s, s')}{\dens_\ell(g, s')\densk_\ell(s',s)}\cdot\frac{\moveprob_\ell(s)}{\moveprob_\ell(s')}\right\}
	\end{align}
	preserves the detailed balance w.r.t. $\pi$.
\end{lemma}
\begin{proof}
	Consider the following equality
	\begin{align*}
	\pi(\LargerCdot)=\int\int_{\LargerCdot}\dens_\ell(g, s')\lambda_\ell(ds')\pi\circ\funcg_\ell^{-1}(dg)
	\end{align*}
	For $A,B\in\mcal{A}$ with $A\cap B=\emptyset$ we get
	\begin{align*}
	&\pi\otimes\mu(A\times B)
	=\int\sum\limits_{\ell\in\moves}\int_A\int_B\acc_{s's}^\ell\moveprob_\ell(s')\densk_\ell(s',s) \lambda_\ell(ds)\dens_\ell(g, s')\lambda_\ell(ds')\pi\circ\funcg_\ell^{-1}(dg)\\
	&=\int\sum\limits_{\ell\in\moves}\int_A\int_B\ind{\funcg_\ell(s')=\funcg_\ell(s)=g}\acc_{s's}^\ell\moveprob_\ell(s')\densk_\ell(s',s) \lambda_\ell(ds)\dens_\ell(g, s')\lambda_\ell(ds')\pi\circ\funcg_\ell^{-1}(dg)\\
	&\overset{*}{=}\int\sum\limits_{\ell\in\moves}\int_B\int_A\acc_{ss'}^\ell\moveprob_\ell(s)\densk_\ell(s,s')\lambda_\ell(ds') \dens_\ell(g, s)\lambda_\ell(ds)\pi\circ\funcg_\ell^{-1}(dg)
	=\pi\otimes\mu(B\times A)
	\end{align*}
	whereby we have used Fubini's theorem together with the partial sigma finiteness of $\lambda_\ell$ in *. 
\end{proof}

The major application scenario for Metropolis-within-Gibbs is to simplify sampling within coordinate spaces.
Let $\states=\states_1\times\ldots\times\states_n$ be an arbitrary coordinate space together with a product sigma algebra $\mcal{A}$ and for $\ell=1,\ldots,n$ define $\funcg_\ell(s)=(s_1,\ldots,s_{\ell-1},s_{\ell+1},\ldots,s_n)$.
Move $\ell$ describes transitions where all but the $\ell$-th coordinate remain fixed.

Let further $\pi(ds)=\dens(s)\nu(ds)$ with a product measure $\nu=\nu_1\otimes\ldots\otimes\nu_n$ over $\mcal{A}$. We define
\begin{align*}
&\dens_\ell(\funcg_\ell(s'), s)=\frac{\dens(s)}{\int\dens(s)\nu_\ell(ds_\ell)}\ind{\funcg_\ell(s')=\funcg_\ell(s)}\quad\text{ and }\quad\lambda_\ell(ds)=\nu_\ell(ds_\ell)\prod_{j\not =\ell}\#(ds_j)
\end{align*} 
with the counting measure $\#$. 
$\int\dens(s)\nu_\ell(ds_\ell)$ is negligible here since it cancels out in (\ref{eq:metgibbsacc}).
$\proposal_\ell$ is essentially build from a kernel density w.r.t. $\nu_\ell$.
If we chose plain probabilities $\moveprob_\ell$ for $\ell\in\moves$, we may write 
$\acc_{s' s}^\ell=\min\left\{1,\frac{\dens(s)\densk_\ell(s, s')}{\dens( s')\densk_\ell(s', s)}\right\}$
for transitions $(s',s)$ that differ in coordinate $\ell$ only.

The Metropolis-within-Gibbs approach exhibits a tremendous potential and even captures particular post-hoc translations.
To see this, consider a translation function $\bo s:\states'\rightarrow\states$ between two measure spaces $(\states',\mcal{A}',\lambda')$ and $(\states,\mcal{A},\lambda)$. 
The corresponding post-hoc translation from $\states'$ to $\states$ applies $\bo s$ in a deterministic manner.
However, if $\bo s$ is not injective, the backward move is usually not deterministic.
In the following, we refer to this sort of moves as semi deterministic translations (SDT's).
In practice, SDT's are the most commonly applied post-hoc translations.

Define
\begin{align*}
\funcg(s)=\begin{cases}
\bo s(s)& s\in\states'\\
s& s\in\states
\end{cases}
\end{align*}
This yields Metropolis-within-Gibbs moves that mimic post-hoc translations.
%In order to transition from $\states$ to $\states'$ we rely on the concrete form of $\prob{\rand X\in \LargerCdot\mid \funcg(\rand X)=g}$.
Even though transitions within one space and the same are not per se excluded here, we deem them as unsupported in this particular move.

In contrast to Lemma \ref{lem:post-hoc}, $\bo s$ is not subject to restrictions.
Lemma \ref{lem:metgibbs} ensures that the acceptance probabilities are correct as long as we work within measure spaces that are in line with $\prob{\rand X\in \LargerCdot\mid \funcg(\rand X)=g}$.
%Nevertheless, insensible choices of $\bo s$ yield poor moves.
%If, for example, the image of a non-null set w.r.t. $\pi$ poses a null set, it has to be  nullset under $\prob{\rand X\in \LargerCdot\mid \funcg(\rand X)=g}$ as well.
%set remains inaccessible from outside the domain and image of $\bo s$.
The following lemma represents a reconciliation of Lemma \ref{lem:post-hoc} and \ref{lem:metgibbs}.
\begin{lemma}
	\label{lem:metgibbspost-hoc}	
	Assume that there exist a sigma finite measure space $(\ustates,\mcal{B},\nu)$ and measurable functions $\bo u:\states'\rightarrow \ustates$ and $f:\states\times \ustates\rightarrow\mathbb{R}_{\geq 0}$ such that 
	\begin{align*}
	&\textbf{\emph{(I)}}\ \ (\bo s, \bo u):\states'\rightarrow\states\times \ustates \text{ is a bijection}\\
	&\textbf{\emph{(II)}}\ \ \lambda'\left((\bo s, \bo u)^{-1}(ds, du)\right)=f(s, u)\lambda\otimes\nu(ds, du)
	\end{align*}
	For $s\in\states$, we may write
	\begin{align}
	\label{eq:metgibpost}
	\prob{\rand X\in \LargerCdot\mid \funcg(\rand X)=s}\propto
	\dens(s)\delta_{s}(\LargerCdot)
	+\int \delta_{(\bo s, \bo u)^{-1}}(\LargerCdot)\dens\left((\bo s, \bo u)^{-1}\right) f(s,u)\nu(du)
	\end{align}	
\end{lemma}
\begin{proof}
We have to show that $\int_B \prob{\rand X\in A\mid \funcg(\rand X)=s} \pi\circ \funcg^{-1}(ds)=\pi(A\cap \funcg^{-1}(B))$ for all $B\in\mcal{A}$ and $A\in\mcal{A}\otimes\mcal{A}'$.
To this end, consider that
\begin{align*}
&\pi\circ\funcg^{-1}=\int_{\LargerCdot}\dens(s)\lambda(ds)+\int\delta_{\bo s}(\LargerCdot)\dens(s')\lambda'(ds')\\
&=\int_{\LargerCdot}\dens(s)\lambda(ds)+\int\delta_{s}(\LargerCdot)\dens((\bo s, \bo u)^{-1})f(u,s)\nu\otimes\lambda(du, ds)\\
&=\int_{\LargerCdot}\dens(s)\lambda(ds)
+\int_\LargerCdot \underbrace{\int\dens\left((\bo s, \bo u)^{-1}\right)f(u,s)\nu(du)}_{=\bcc(s)}\lambda(ds)
=\int_{\LargerCdot}\dens(s)+\bcc(s)\lambda(ds)
\end{align*}
This allows us to write
\begin{align*}
&\int_B \prob{\rand X\in A\mid \funcg(\rand X)=s}\left(\dens(s)+\bcc(s)\right)\lambda(ds)\\
&\overset{*}{=}\int_B \dens(s)\delta_{s}(A)+ \int \delta_{(\bo s, \bo u)^{-1}}(A)\dens\left((\bo s, \bo u)^{-1}\right) f(s,u)\nu(du)\lambda(ds)\\
&=\pi(A\cap B)
+\int\delta_s(B)\delta_{(\bo s, \bo u)^{-1}}(A)\dens\left((\bo s, \bo u)^{-1}\right) \lambda'\left((\bo s, \bo u)^{-1}(ds, du)\right)\\
&=\pi(A\cap B)
+\int\delta_{\bo s}(B)\delta_{s'}(A)\dens(s') \lambda'(ds')
=\pi(A\cap B)+\pi(A\cap \bo s^{-1}(B))=\pi(A\cap \funcg^{-1}(B))
\end{align*}
whereby in $*$, we have set the proportionality constant in (\ref{eq:metgibpost}) to $1\slash\left(\dens(s)+\bcc(s)\right)$.
\end{proof}

Consequently, in accordance with Lemma \ref{lem:metgibbs}, the proposal for transitioning from $s'\in\states'$ to $\states$ complies with $\delta_{\funcg(s')}$ and thus applies $\bo s$ deterministically. 
In turn, as in post-hoc translations, the backward move may exploit the form of $\int \delta_{(\bo s, \bo u)^{-1}}(\LargerCdot)\nu(du)$.
To this end, we employ a kernel density $\densk$ from $(\states,\mcal{A})$ to $(\ustates,\mcal{B}, \nu)$.
A small contemplation yields that the acceptance probability for transitions from $s'\in\states'$ to $\states$ then reads
\begin{align*}
\min\left\{1,\frac{\dens(\bo s)\densk(\bo s, \bo u)}
{\dens(s')f(\bo s, \bo u)}\cdot\frac{\moveprob(\bo s)}{\moveprob(s')}\right\}
\end{align*}
with move probability $\moveprob$.
This is the same acceptance probability as in the corresponding post-hoc approach.

Please note that the Metropolis-within-Gibbs algorithm doesn't capture all post-hoc translations.
For translations that are not SDT's, there is hardly a suitable $\funcg$ available that yields the same sampling scheme.
To see this, consider the relation defined by the viable transitions under the two translation functions that belong to a move and its backward move.
$\funcg$ is build from the partition of the two state spaces derived from the implied equivalence relation.
However, for non-SDT's this partition can easily contain the spaces as a whole rendering them useless.

%\begin{corollary}
%The acceptance probability for the transition from $s'\in\states'$ to $\states$ reads
%\begin{align*}
%\acc_{s'}=\min\left\{1,\frac{\dens(\bo s)\densk(\bo s, \bo u)}
%{\dens(s')f((\bo s, \bo u))}\cdot\frac{\moveprob(\bo s)}{\moveprob(s)}\right\}
%\end{align*}
%\end{corollary}
%\begin{proof}
%For $A\in\mcal{A}'$ and $B\in\mcal{A}$ we get
%\begin{align*}
%&\int \moveprob((\bo s, \bo u)^{-1})\delta_{s}(B)\delta_{(\bo s, \bo u)^{-1}}(A)\acc_{(\bo s, \bo u)^{-1}}\dens\left((\bo s, \bo u)^{-1}\right)
%f(s,u)\nu(du)\\
%&=\int \delta_{(\bo s, \bo u)^{-1}}(A)
%\densk(s,u)\nu(du)\moveprob(s)\delta_{s}(B)\acc_{s}\dens(s)
%\end{align*}
%\end{proof}

%Another interesting application that lies beyond the scope of this paper is the Metropolis within Gibbs version of the Hamiltonian MCMC algorithm \citep{duane1987hybrid}.
%Imagine, we use several moves whereby one of them, say the $\ell$-th, employs $\funcg_\ell(s')=\dens(s')$.
%Move $\ell$ solely traverses within the contour line of $\dens$ with height $\dens(s')$ and thus, proposes states of equal quality irrespective of the distance that is covered.

\subsection{A \cp example}
\label{chap:cpexample}

In this section, we scrutinize our theory by means of a small \cp example.
To this end, we consider three different implementations for so-called birth and death moves which either add or remove \cps.
Firstly, we employ very plain and unsophisticated proposals.
Secondly, we apply ad-hoc translations that incorporate maximizers of certain partial likelihoods.
Finally, we use tightly adapted post-hoc translations in the fashion of the reversible jump algorithm.

It turns out that the ad-hoc and post-hoc approaches perform equally well on this example in terms of their computation times and acceptance rates.
In contrast, the plain approach performs poorly, which justifies the need of sophisticated ad-hoc and post-hoc translations.

\begin{figure}[ht]
	\includegraphics[width=1\linewidth]{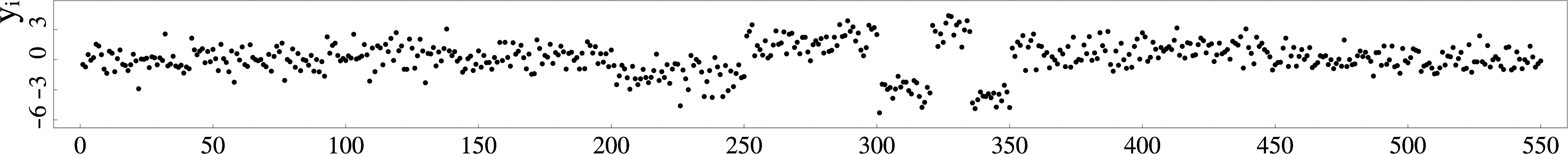}
	\caption[Gaussian change in mean example dataset]{Gaussian change in mean example dataset.}
	\label{pic:guassmhrj}
\end{figure}

Figure \ref{pic:guassmhrj} shows an artificial dataset.
The $n=550$ datapoints were drawn independently from a normal distribution with variance 1 and mean values that where subject to 9 changes.

To build an exemplary Bayesian model here, we choose a prior for the \cp locations and mean values, i.e. the segmentation and its heights.
The time from one \cp to the next is geometrically distributed with $q=3\slash550$ and the segment heights in case of a jump and at the beginning are distributed according to $\mcal{N}(0, 25)$.
The data at timepoint $i=1,\ldots,n$ also follows a normal distribution with its respective mean at $i$ and variance equal to 1.

The sampling starts with no \cps and an overall segment height of 0.
Subsequently, new \cp locations may be found or discarded through \emph{birth} and \emph{death} moves.
Additionally, we employ \emph{shift} moves to shift single \cps and \emph{adjust} moves to adjust single segment heights.

Each move only manipulates a subset of the segment heights or a single \cp and leaves the rest as it is.
The Metropolis-within-Gibbs framework as well as post-hoc translations support such moves and they can be applied directly without any extra effort.

Please note that the backward moves for adjust and shift are again adjust and shift.
Similarly, death and birth will be reversed by birth and death, respectively.
This arises naturally here, since each of the moves has its very own and unique complementary move, e.g. a birth can only be reversed by a death.

An adjust move relocates the old height $h'$ of a uniformly chosen segment according to $\mcal{N}(h', 0.5)$.
In turn, a shift move shifts the location of a uniformly chosen \cp uniformly to a new position within the neighboring \cps.

Let $\phi(x\semic \mu,\sigma^2)$ be the density of the univariate normal distribution with mean $\mu$ and variance $\sigma^2$ evaluated at $x$.
Adjusting the height $h'$ of a segment $S$ to $h$ yields an acceptance probability of
\begin{align*}
\min\left\{1,\quad \frac{\phi(h\semic 0,25)}{\phi(h'\semic 0,25)}\prod\limits_{j \in S}\frac{\phi(y_j\semic h, 1)}{\phi(y_j\semic h', 1)}\quad\right\}
\end{align*}
Note that the probabilities of choosing the shift move, the segment and the density values of the datapoints within the untouched segments, cancel out.
Furthermore, since the normal distribution is symmetric w.r.t. its mean, the proposal densities cancel out as well.
What remains are the density values for the datapoints within the adjusted segment and the priors for the segment heights.

Consider $1$ and $n+1$ to be immutable auxiliary \cps.
In the shift move, we shift an existing \cp, say at location $i$, uniformly to a new location, say $j$, within the neighboring \cps.
Let $\ell$ and $k$ be the \cps to the left and right of the \cp at $i$ and let $h_1$ respectively $h_2$ be the corresponding segment heights.
The acceptance probability for the shift move reads
\begin{align*}
&\min\left\{1,\quad \frac{\prod\limits_{m=\ell}^{j-1} \phi(y_m\semic h_1, 1)\prod\limits_{m=j}^{k-1} \phi(y_m\semic h_2, 1)}{\prod\limits_{m=\ell}^{i-1} \phi(y_m\semic h_1, 1)\prod\limits_{m=i}^{k-1} \phi(y_m\semic h_2, 1)}\quad\right\}
\end{align*}
Here, the probabilities of choosing the move and the \cp, the priors for the heights and the density values of the data within the untouched segments cancel out.

\subsubsection{A plain implementation of birth and death}
\label{chap:plainbd}
In the death move, we choose a \cp uniformly, remove it and propose a new height for the remaining segment according to $\mcal{N}(0, 25)$.
In the birth move, we propose a new \cp uniformly among the timepoints without a \cp.
The two heights left and right of the new \cp are then proposed independently according to $\mcal{N}(0, 25)$.

Let $p_d$ respectively $p_b$ be the probabilities to choose a death or a birth move, respectively.
We perform a death move.
Therefore, let $i$ be the timepoint of the \cp that is to be removed and 
let $\ell$ and $k$ be the \cps to the left and right of the \cp at $i$ and let $h_1$ and $h_2$ be the corresponding segment heights.
Finally, let $h$ be the ensuing single segment height.
The acceptance probability reads
$\min\left\{1,\ \frac{n-c}{c+1}
\cdot\frac{p_b}{p_d}
\cdot L_{\ell, i, k}\right\}$,
whereby $c+1$ is the number of \cps without $1$ and $n+1$ and 
\begin{align*}
L_{\ell,i, k}=\frac{\prod\limits_{m=\ell}^{k-1}\phi(y_m\semic h, 1)}{\prod\limits_{m=\ell}^{i-1} \phi(y_m\semic h_1, 1)\prod\limits_{m=i}^{k-1} \phi(y_m\semic h_2, 1)}
\end{align*}

In order to derive the acceptance probability for the corresponding birth move, we just take the reciprocal of the acceptance ratio.

\subsubsection{An ad-hoc implementation of birth and death}
\label{chap:ad-hocimpl}

Here, we utilize parts of the likelihood of the model.
The likelihood of the data within a single segment can be maximized w.r.t. its segment height just by choosing the mean of the involved datapoints.
Thus, by proposing new segment heights close to these means we might propose sensible states that yield good acceptance rates.
To this end, our sampling approach proposes new segment heights through a normal distribution with variance $0.01$ and mean equal to the empirical mean of the associated data.

Let $\tau^2=0.01$ and for $1\leq \ell\leq k\leq n$ define $\mu_{\ell:k}=\sum_{m=\ell}^{k-1}\frac{y_m}{(k-\ell)}$.
By using the notation of Section \ref{chap:plainbd}, the acceptance probability for the death move reads
\begin{align*}
\min\left\{1,\ 
\frac{n-c}{c+1}
\cdot\frac{p_b}{p_d}
\cdot L_{\ell,i, k}\cdot \frac{\phi(h\semic 0,25)}{\phi(h_1\semic 0,25)\phi(h_2\semic 0,25)}\cdot \frac{\phi(h_1\semic \mu_{l:i}, \tau^2)\phi(h_2\semic \mu_{i:k}, \tau^2)}
{\phi(h\semic \mu_{l:k}, \tau^2)}\quad\right\}
\end{align*}

\subsubsection{A post-hoc implementation of birth and death}
Now, we design birth and death moves by virtue of the reversible jump paradigm.
It is crucial hereby to find a carefully adapted diffeomorphism.
As we have argued already, a good estimator for a segment height is the empirical mean of the observations belonging to it.
If $h$ is this mean and we split the segment, say from $\ell$ to $k$, into two parts through a \cp at $i$, such that the lengths of the resulting segments are $n_1=i-\ell$ and $n_2=k-i+1$ and their associated mean values are $h_1$ and $h_2$, we see that $h=\frac{n_1h_1+n_2h_2}{n_1+n_2}$.

Furthermore, in order to build a diffeomorphism thereof, we introduce an auxiliary variable $u$.
A sense of parsimony is achieved by stipulating that $u\in\mathbb{R}$.
We treat $u$ as part of the target space by setting $u=h_2$.
That way, $(h,u)$ describes a diffeomorphism with Jacobi determinant $\frac{n_1}{n_1+n_2}$.

The opportunity to embed the auxiliary variables directly into the target space often remains unnoticed in practical implementations. 
Instead, $u$ would be sampled independently, which is prone to forfeit acceptance rate by proposing unsuited states.

This yields an SDT, whereby we rely on the reasoning of Section \ref{chap:ad-hocimpl} and draw $u$ from $\mcal{N}(m_{i,k}, \tau)$ in the birth move.
Therewith, the acceptance probability for the death move reads
$\min\left\{1,\ \frac{n_1}{n_1+n_2}
\cdot\frac{n-c}{c+1}\cdot\frac{p_b}{p_d}
\cdot L_{\ell,i, k}\cdot \frac{\phi(h\semic 0,25)}{\phi(h_1\semic 0,25)\phi(h_2\semic 0,25)}\cdot \phi(h_2, m_{i,k}, \tau) \right\}$.

\subsubsection{Convergence}
Here, we want to verify the convergence of our three different Markov chains towards the posterior distribution, say $\pi$, that is determined through the \cp model and data.
For this sake, we check irreducibility and aperiodicity as required by Theorem 1 in \citep{tierney1994markov}.
Additionally, Harris recurrence follows from Corollary 2 there.

We stipulate that the move probabilities are all positive.
Consequently, the acceptance probabilities for all proposed transitions will as well be positive.
Therefore, given an arbitrary state, in all three cases, any \cp configuration can be occupied with positive probability through a finite amount of birth, death and shift moves.

In order to prove irreducibility, at first, we consider events, say $E$, whose elements all exhibit the same \cp configuration.
For $E$ to be positive under $\pi$, the set of possible segment heights for each segment must be positive under an arbitrary univariate normal distribution.
Thus, starting from a state that exhibits the very same \cp configuration as $E$, transitioning into $E$ can be done with positive probability by applying adjust moves to each segment height correspondingly.
Again, this requires only a finite number of steps.

Now, take an arbitrary event that is positive under $\pi$ and partition it by the different \cp configurations present.
Due to the finiteness of this partition, one of the ensuing subsets must be positive under $\pi$.
Given an arbitrary state, we reach the \cp configuration of this subset and then transition into this subset, both with positive probability through a finite amount of moves.
This shows the irreducibility.

To see the aperiodicity, consider the two sets of states that can be reached through one or two successive applications of the adjust move. 
These sets overlap significantly, contradicting periodicity.
Thus, the chain is aperiodic.

Together with the invariance, ensured by the particular forms of the acceptance probabilities, we conclude that all three Markov chains converge in total variation towards $\pi$, irrespective of the state they are started in.

\subsubsection{Results and conclusions}
We want to compare our proposals empirically in terms of their runtime and acceptance rates.
To this end, we set all four move probabilities to $0.25$, except in the boundary cases.
If there is only one segment, birth and adjust have the same probability of 0.75 and if each timepoint accommodates a \cp then death and adjust have the probabilities 0.5 respectively 0.25.

\begin{table}[ht]
	\begin{center}	
		\begin{tabular}{ l || c | c | c | c}
			Acceptance rates & Death move & Birth move & Shift move & Adjust move   \\
			\hline\hline
			Plain proposals & 0.0021 & 0.0022 & 0.0681 & 0.2896\\
			Ad-hoc translations & 0.0588 & 0.0594 & 0.0678 & 0.2904\\			
			Post-hoc translations & 0.0639 & 0.0645 & 0.0681 & 0.2899\\
		\end{tabular}
	\end{center}
	\caption[Acceptance rates for particular moves]{Acceptance rates for particular moves.}
	\label{fig:fractable}
\end{table}

For each of the three Markov chains, we simulated $10^{7}$ successive links.
On my computer (AMD 64 with a 3200 GHZ CPU), the sampling algorithm for the plain proposals and the post-hoc translations needed twelve seconds and the ad-hoc approach took 13 seconds.

Table \ref{fig:fractable} shows the acceptance rates for the individual moves.
It becomes apparent that the plain birth and death moves struggle with finding alternate \cp configurations.
Thus, the Markov chain performs poorly in exploring the state space, which impairs the mixing time.

In contrast, our ad-hoc and post-hoc translations compare well with each other, though both exhibit fundamentally different strategies in proposing new segment heights.
Several experiments revealed that an acceptance rate for birth and death moves of more than 6.5\% is infeasible as long as the new \cp locations are chosen in a plain uniform manner.
Thus, the achieved rates are sound, however, higher rates are usually considered as better \citep{bedard2008optimal}.

This shows that sophisticated approaches like ad-hoc and post-hoc translations are key for sampling within mixed spaces.
Whilst our post-hoc proposals can be regarded as ambitiously developed, the ad-hoc proposals constitute a simple independence sampler that was derived from straightforward ideas.
This indicates the great potential of ad-hoc translations.
In turn, post-hoc translations may reduce the number of random variables to be generated and therewith result in more tightly adapted proposals that exhibit a reduced computational runtime.

\subsection{The Metropolis-Hastings algorithm in abstract terms}
\label{chap:genmh}
So far, we have considered different strategies to implement the Metropolis-Hastings algorithm based on densities and measure spaces.
Conversely, \cite{tierney1998note}, provides the theoretic foundations for the Metropolis-Hastings algorithm to work with arbitrary choices of $\pi$ and $\proposal$ directly.
At the same time, it introduces the notion of a maximal Metropolis-Hastings algorithm.
Even though this concept remained widely unrecognized, it gives important advice on the design of proposals and allows for a classification of Metropolis-Hastings approaches.

In this section, we will summarize the main outcomes of \cite{tierney1998note} and put them in line with the results of \cite{roodaki2011} for mixture proposals.
Subsequently, we examine maximality with regards to our own methods.

%First of all, we restate the detailed balance condition in a slightly more general way.
Given a measurable space $(\states,\mcal{A})$, for a set $A\in\mcal{A}\otimes \mcal{A}$ we define the \emph{transpose} of $A$ as $A^t=\left\{(s,s')\mid (s',s)\in A\right\}$.
We leave it to the reader to show that $A^t\in\mcal{A}\otimes \mcal{A}$ and that $\pi\otimes\mu(\LargerCdot)=\pi\otimes\mu(\LargerCdot^t)$ if and only if $\mu$ preserves the detailed balance w.r.t. $\pi$.
%\begin{corollary}
%	$\pi$ and $\mu$ meet the detailed balance condition if and only if $\pi\otimes\proposal=\pi\otimes\proposal(\LargerCdot^t)$.
%\end{corollary}
%\begin{proof}
%	The proof of ``$\Leftarrow$'' is straightforward.
%	Since the sets $A\times B$ with $A,B\in\mcal{A}$ generate $\mcal{A}\otimes\mcal{A}$ and are closed under intersections, $\pi\otimes\proposal$ is the unique measure with values $\pi\otimes\proposal(A\times B)$ over those sets.
%	The detailed balance thus implies that $\pi\otimes\proposal=\pi\otimes\proposal(\LargerCdot^t)$.
%\end{proof}

In the following, we are given an arbitrary probability space $(\states,\mcal{A},\pi)$ and a Markov kernel $\proposal$ from $(\states,\mcal{A})$ to $(\states,\mcal{A})$.
The idea is to use
\begin{align*}
\ratio_{s's}=\frac{\pi\otimes\proposal(ds,ds')}{\pi\otimes\proposal(ds',ds)}
\end{align*}
as the acceptance ratio.
$\ratio_{s's}$ is the density of $\pi\otimes\proposal(\LargerCdot^t)$ w.r.t. $\pi\otimes\proposal$ evaluated at $(s',s)$.
%Furthermore, in order to prove detailed balance in the same fashion as before, we require that $\ratio_{s's}=1\slash\ratio_{ss'}$.

At first, we need to find a set where $\ratio_{s's}$ is guaranteed to exist.
To this end, consider $\xi(\LargerCdot)=\pi\otimes\proposal(\LargerCdot)+\pi\otimes\proposal(\LargerCdot^t)$.
$\pi\otimes\proposal$ and $\pi\otimes\proposal(\LargerCdot^t)$ are absolutely continuous w.r.t. $\xi$.
Thus, the theorem of Radon-Nikodym \citep{nikodym1930generalisation} ensures the existence of $\frac{\pi\otimes\proposal(ds',ds)}{\xi(ds',ds)}$ and $\frac{\pi\otimes\proposal(ds,ds')}{\xi(ds',ds)}$.
Let further 
\begin{align*}
\gibbsrel=\Bigg\{(s',s)\bbmid\frac{\pi\otimes\proposal(ds',ds)}{\xi(ds',ds)}>0, \frac{\pi\otimes\proposal(ds,ds')}{\xi(ds',ds)}>0\Bigg\} 
\end{align*}
On $\gibbsrel$ the density $\ratio_{s's}$ exists and equals the ratio of $\frac{\pi\otimes\proposal(ds,ds')}{\xi(ds',ds)}$ and $\frac{\pi\otimes\proposal(ds',ds)}{\xi(ds',ds)}$.

%\begin{proof}
%	This holds since for $(s',s)\in \gibbsrel$ we may write
%	\begin{align*}
%	\pi\otimes\proposal(ds',ds)&=\frac{\pi\otimes\proposal(ds',ds)}{\xi(ds',ds)}\xi(ds',ds)\\
%	&=\frac{\pi\otimes\proposal(ds',ds)}{\xi(ds',ds)}\underbrace{\frac{\pi\otimes\proposal^t(ds',ds)}{\xi(ds',ds)}\densslashb\frac{\pi\otimes\proposal^t(ds',ds)}{\xi(ds',ds)}}_{=1}\xi(ds',ds)\\
%	&=\frac{\pi\otimes\proposal(ds',ds)}{\xi(ds',ds)}\densslashb \frac{\pi\otimes\proposal^t(ds',ds)}{\xi(ds',ds)}\pi\otimes\proposal^t(ds',ds)
%	\end{align*}
%\end{proof}
\begin{lemma}[\cite{tierney1998note}]
	\label{lem:genmh}	
	The transition kernel $\mu$ in (\ref{eq:mhtranskern}) with acceptance probability 
	\begin{align}
	\label{eq:accratio}
	\acc_{s' s}=\begin{cases}
	\min\{1,\ratio_{s's}\}& (s',s)\in \gibbsrel\\
	0& \text{otherwise}
	\end{cases}
	\end{align}
	preserves the detailed balance w.r.t. $\pi$.
	Furthermore, $\acc_{s' s}$ is $\xi$ almost surely maximal over all such acceptance probabilities.
\end{lemma}
\begin{proof} 
    By considering the symmetry of $\xi$ and $\gibbsrel$, the detailed balance can be proven in the same fashion as Lemma \ref{lem:mh}.	
    Furthermore, any proper acceptance probability $\hat \acc_{s's}$ satisfies $\frac{\pi\otimes\proposal(ds',ds)}{\xi(ds',ds)}\hat\acc_{s's}=\frac{\pi\otimes\proposal(ds,ds')}{\xi(ds',ds)}\hat\acc_{ss'}$ almost surely w.r.t. $\xi$.
    Thus, on $\gibbsrel$ we get $\hat\acc_{s's}=\ratio_{s' s}\hat\acc_{ss'}\leq\min\{1, \ratio_{s' s}\}=\acc_{s's}$ and a small contemplation yields that outside of $\gibbsrel$ we get $\hat\acc_{s's}=0$.
\end{proof}

Lemma \ref{lem:genmh} allows us to speak of the \emph{maximal Metropolis-Hastings algorithm} for the proposal $\proposal$ when we use $\acc_{s' s}$ as in (\ref{eq:accratio}).

$\gibbsrel$ is $\xi$ almost surely uniquely defined and comprises all transitions an arbitrary Metropolis-Hastings sampler with proposal $\proposal$ and target distribution $\pi$ can carry out.
%If we are to give an illustration of $\gibbsrel$, we could describe it as the set of viable transitions in forward and backward direction.

Take, for example, an SDT with translation function $\bo s$ and a state $s$ such that $\pi(\{s\})=0$ but $\pi\left(\bo s^{-1}(s)\right)>0$.
Let $T=\bo s^{-1}(s)\times \{s\}$.
Since $\xi\left(T\right)>0$ and $\pi\otimes\proposal\left(T^t\right)=0$, we see that $\gibbsrel$ and $T$ are $\xi$ almost surely disjoint.
Thus, the transitions defined by $T$ are not accessible for this SDT.
In order to take advantage of such degenerated translation functions, we should resort to the ad-hoc paradigm instead.
Section \ref{chap:ad-hocimpl} provides an example for this.

\subsubsection{Mixture proposals in general and maximality in particular}
\label{chap:mixtures}

So far, we have build the acceptance probabilities for mixture proposals from unique pairs of forward and backward moves.
However, the computation of the acceptance probability for the maximal Metropolis-Hastings algorithm involves the mixture proposal as a whole.
Consequently, if the two approaches are different, the pairwise one is prone to yield smaller acceptance rates. 
%Unfortunately, lower acceptance rates increase the dependencies of successive links, which, in turn, impairs the mixing time.

Nevertheless, the pairwise approach is an important measure of convenience.
It exonerates us from evaluating the full, perhaps infinite range of component proposals.
This gives rise to the following lemma that provides sufficient conditions for the pairwise approach to satisfy the maximality.

\begin{lemma}[\cite{roodaki2011}]
	\label{lem:mixmh}
	If for all $\ell\in\moves$, there exist disjoint $Z_{\ell}\in\mcal{A}\otimes\mcal{A}$ with
	\begin{align*}
	&\textbf{(I) }\quad\int_{\states\times \states \backslash Z_\ell}\moveprob_\ell(s')\pi\otimes\proposal_\ell(ds',ds)=0\\
	&\textbf{(II) }\quad Z_{\ellr}=Z_\ell^t
	\end{align*}	
	we get 
	\begin{align}
	\label{eq:mixpairdens}
	&\frac{\pi\otimes\proposal(ds, ds')}{\pi\otimes\proposal(ds', ds)}=\frac{\moveprob_{\ellr}(s)}{\moveprob_\ell(s')}\frac{\pi\otimes\proposal_{\ellr}(ds, ds')}{\pi\otimes\proposal_\ell(ds', ds)}	
	\end{align}	
	for $(s',s)\in Z_\ell$.
\end{lemma} 
\begin{proof}
	Define $\ratio_{s's}^\ell=\frac{\moveprob_{\ellr}(s)}{\moveprob_\ell(s')}\frac{\pi\otimes\proposal_{\ellr}(ds, ds')}{\pi\otimes\proposal_\ell(ds', ds)}$. 
	For $(s',s)\in Z_\ell$ we get
	\begin{align*}
	&\ratio_{s's}^\ell\pi\otimes \proposal(ds',ds)
	=\ratio_{s's}^\ell\moveprob_\ell(s')\pi\otimes \proposal_\ell(ds',ds)
	=\moveprob_{\ellr}(s)\proposal_{\ellr}(ds,ds')=\pi\otimes \proposal(ds,ds')
	\end{align*}
	This holds since \textbf{(I)} and \textbf{(II)} state that we can set $\moveprob_k(s')\pi\otimes \proposal_k(ds',ds)=\moveprob_\kr(s)\pi\otimes \proposal_\kr(ds,ds')=0$ for all $k\not =\ell$.
\end{proof}

Strangely enough, \cite{roodaki2011} seem to make a different statement with their theorem.
They point out when the acceptance probabilities can be derived from pairs of component moves.
Therewith, they appear to miss that this is always viable as long as the pairs are unique, a statement that has already been known from \citep{tierney1998note}.

Under the assumption that the conditions \textbf{(I)} and \textbf{(II)} of Lemma \ref{lem:mixmh} are met, we may ask for the maximality of the approaches considered in this work.
As a rule, it is sufficient to show that for each move $\ell$ the acceptance ratios correspond to the r.h.s. of (\ref{eq:mixpairdens}).
This is straightforward for the primal Metropolis-Hastings algorithm and ad-hoc translations, and also for (\ref{eq:accdiscpost}) and (\ref{eq:metgibbsacc}).

However, post-hoc translations that accept within the auxiliary space are deviant.
In the context of Lemma \ref{lem:post-hoc}, we define 
\begin{align*}
\ratio_{s' u'}^\ell=\frac{\dens(\bo s_{\ell})\densk_{\ellr}((\bo s, \bo u)_{\ell})}{\dens(s')\densk_{\ell}(s',u')}\cdot f_{\ell}(s',u')
\end{align*}
Furthermore, let $\proposal_{\ell}(s',\LargerCdot)=\int \delta_{\bo s_{\ell}}(\LargerCdot)\densk_{\ell}(s', u') \nu_\ell(du')$ be the proposal of move $\ell$.

\begin{lemma}
\label{lem:rvjproposdens}
For move $\ell$ with transitions from $s'\in\states_j$ to $s\in\states_i$, we get
\begin{align*}
\frac{\pi\otimes\proposal_{\ellr}(ds,ds')}{\pi\otimes\proposal_{\ell}(ds',ds)}=\expec{\ratio_{s' \rand U}^\ell\bmid \bo s_{\ell}(s',\rand U)=s}
\end{align*}
whereby $\rand U$ is a random variable distributed according to $\int_\LargerCdot\densk_{\ell}(s', u')\nu_\ell(du')$.
\end{lemma}
\begin{proof}
This holds since for $A\in\mathcal{A}_j$ and $B\in\mathcal{A}_i$
\begin{align*}
&\int_A \int_B \expec{\ratio_{s' \rand U}^\ell\bmid \bo s_{\ell}(s',\rand U)=s}\proposal_{\ell}(s',ds)\pi(ds')\\
&=\int_A \int \expec{\ratio_{s' \rand U}^\ell\bmid \bo s_{\ell}(s',\rand U)=\bo{s}_{\ell}(s', u')}\delta_{\bo{s}_{\ell}}(B)\densk_{\ell}(s',u')\nu_\ell(du')\pi(ds')\\
&\overset{*}{=}\int_A \int \delta_{\bo{s}_{\ell}}(B)\ratio_{s' u'}^\ell\densk_{\ell}(s',u')\dens(s')\nu_\ell(du')\lambda(ds')\\
&=\int_A\delta_{\bo s_{\ell}}(B) \int_B \dens(\bo s_{\ell})\densk_{\ellr}((\bo s, \bo u)_{\ell})f_{\ell}(s',u')\nu_\ell(du')\lambda(ds')\\
&\overset{**}{=}\int_B \int \delta_{\bo s_{\ellr}}(A)\densk_{\ellr}(s,u)\dens(s)\nu_\ellr(du)\lambda(ds)=\pi\otimes\proposal_{\ellr}(B\times A)
\end{align*}
whereby we have used the basic properties of the conditional expectation in $*$ and condition \textbf{(I)} and \textbf{(II)} of Lemma \ref{lem:post-hoc} in **.
\end{proof}

Thus, the acceptance probability for move $\ell$ of the maximal Metropolis-Hastings algorithm for the post-hoc translations of Lemma \ref{lem:post-hoc} reads
\begin{align*}
\acc_{s's}^\ell=\min\left\{1,\ \frac{\moveprob_{\ellr}(s)}{\moveprob_{\ell}(s')}\cdot\expec{\ratio_{s' \rand U}^\ell\bmid \bo s_{\ell}(s',\rand U)=s}\right\}
\end{align*}
This states that given a previous link $s'$, we accept or reject the newly proposed state $s$ based on an expectation drawn from $\ratio_{s' u'}$ over those values of $u'$ that yield $\bo s_{\ell}(s',u')=s$.
If $\bo s_{\ell}(s',\LargerCdot)$ is injective, $\ratio_{s' u'}^\ell=\expec{\ratio_{s' \rand U}^\ell\bmid \bo s_{\ell}(s',\rand U)=\bo s_{\ell}(s',u')}$ holds and the post-hoc approach for this move is guaranteed to be maximal.
This is the case for all SDT's in the fashion of Lemma \ref{lem:post-hoc} and \ref{lem:metgibbspost-hoc}.

On the downside, consider a transition $(s',s)$ where the maximal Metropolis-Hastings algorithm exhibits an acceptance probability of one, i.e. $\int \ratio_{s' u'}\prob{\rand U\in du'\mid \bo s_{\ell}(s',\rand U)=s}\geq 1$.
The corresponding post-hoc translation approach accepts $(s',s)$ with probability one if and only if $1=\prob{\ratio_{s' \rand U}\geq 1\mid \bo s_{\ell}(s',\rand U)=s}$.
In non-injective cases, these two conditions can be very distinct and lead to significantly different results.

\subsection{Discussion and literature review}
\label{chap:discussion_trans}

In this section, we examined several practical and theoretical approaches to derive acceptance probabilities for the Metropolis-Hastings algorithm.
Due to their tricky characteristics, post-hoc translations received the most attention.
In accordance with \cite{green1995}, key to their implementation is that the accept-reject step is performed within the auxiliary space rather than the target space.
As a consequence, the translation functions need to be supplemented by auxiliary functions determining for each transition its unique backward transition.
On top of this bijective completion, further conditions are imposed in order to enable a certain integral transformation.

In an abstract consideration of the Metropolis-Hastings algorithm, we came across the notion of maximality as introduced by \cite{peskun1973optimum} and \cite{tierney1998note}.
Post-hoc translations may lose this maximality if they are not carefully implemented.
This renders SDT one of the most vital post-hoc approaches.
Interestingly, we found that SDT's are also accessible through the Metropolis-within-Gibbs paradigm.

\cite{tierney1998note} was able to identify the superset of possible transitions a Metropolis-Hastings sampler is able to carry out.
Therewith, we understood the limits and possibilities of SDT's better.
We found that translation functions that exhibit a certain degeneracy are ineffective in post-hoc approaches, a point that underlines the utter relevance of ad-hoc translations.

%In our \cp example, the Metropolis within Gibbs algorithm was used in the traditional way to reduce the dimension of the proposals.
%Additionally, we unveiled its innovative potential with regard to the design and implementation of proposals.
%This opened a new perspective on SDT's.

All considered methods made intense use of mixture proposals.
They feature a modular design and are therefore of high practical value.
It is further beneficial to compute the acceptance probabilities solely based on unique pairs of moves.
However, we should avoid an overlap of their supports since this may forfeit maximality.

So far, we solely considered mixture proposals w.r.t. a countable set of moves.
However, the set of moves can alternatively be part of any sigma finite measure space, whereby the move probabilities are replaced by a kernel density.
Consequently, summation over mixture components is replaced by integration.
As before, we need to define unique pairs of forward and backward moves.
An analogy to this is a post-hoc translation having an auxiliary variable that represents the move with an associated auxiliary function that determines the corresponding backward move.

Our purposive contemplation of ad-hoc as well as post-hoc translations, and Metropolis within Gibbs opens up innovative ways of choosing sampling schemes that satisfy various demands regarding simplicity and sophistication.
The author of this work hopes that this helps to eliminate the ubiquitous confusions surrounding transdimensional sampling.

In the remainder of this section, we discuss a sequence of far-reaching errors in reasoning.
They all share the conclusion that ordinary MCMC methods struggle with transdimensional sampling.

\cite{green1995} concludes wrongly that we have to use dimension matching  in order to pursue transdimensional moves.
On page 715 in \cite{green1995}, the acceptance probability solely utilizes single components of the mixture proposal.
This implies that the density used to perform a move must also be used for the corresponding backward move.
For transdimensional moves, this is obviously not possible and made it inevitably necessary to introduce the dimension matching condition.
Thus, Green's undeniably great idea appears to be the result of a simple error in reasoning.

\cite{carlin1995bayesian} claim, by referring to \cite{tierney1994markov}, that transdimensional moves create absorbing states and therefore violate the convergence of Markov chains.
Unfortunately, \cite{tierney1994markov} doesn't seem to provide this statement.
In order to circumvent this putative problem, they propose to work on the product space $\prod_{i\in\spaces}\states_i$ instead.
Frankly speaking, this adds a huge burden just to avoid the inhomogeneous nature of the spaces.
According to Google scholar, \cite{carlin1995bayesian} was cited over 1000 times, though a few of them question the necessity of this approach, e.g. \cite{godsill2001relationship, green2003trans, green2009reversible, green1998model}.

In a more measure theoretic context, \cite{chen2012} comes to the conclusion that usual MCMC cannot transition across spaces of different dimensions.
They argue with the lack of dominating measures and therewith ignore the fact that each countable collection of sigma finite measures indeed exhibits a common dominating measure.

As we have seen, the acceptance probability of the reversible jump algorithm is not necessarily maximal in the sense of Lemma \ref{lem:genmh}.
Interestingly, \cite{green1995} refers to \cite{peskun1973optimum} in order highlight the maximality of his choice of acceptance probability.

These misperceptions have introduced a significant burden to the field of transdimensional sampling.
The concerning literature is difficult to overview due to inconsistent claims and full of poor mathematical language due to the absurd overemphasis on dimensionality.
As a result, it sometimes feels like reading some sort of star trek novel.

Finally, the name ``reversible jump" apparently stems from Green's wrong conviction that he constructed a method that overcomes the non-reversibility of transdimensional moves in the Metropolis-Hastings algorithm.
Thus, this name lacks a proper meaning and collides awkwardly with ``reversible MCMC" making these topics even harder to grasp for newcomers.

%Due to the current lack of alternatives, the unprepared user of mixed spaces is pushed to understand and apply diffeomorphisms and their Jacobi determinants.
%Since the diffeomorphisms need to fit very closely to the employed model and application in order to be rewarding, it is by no means clear if this is per se appropriate in all application scenarios 

\newpage\null\newpage
\section{Exact inference in Bayesian \cp models}
\label{chap:inference}

A sequence of posterior samples obtained from a \cp model and data allows us to derive various quantities like the expected number of \cps, a \cp histogram and many more.
\cite{siekmann2011mcmc} gives an impressive demonstration of what can be inferred from posterior samples based on ion channel data.

However, the number of possible \cp configurations grows exponentially w.r.t. the data size and thus, only fractions of them can be covered by samples.
This might often be sufficient as there are usually comparatively few relevant ones.
Nonetheless, we can only be sure of the accuracy of our approaches if we apply exact inference algorithms, which is indeed feasible at times.

In this section, we elaborate exact inference strategies for a certain, but still broad, class of Bayesian \cp models.
A range of such algorithms already exists.
However, across those papers the mathematical notations are quite diverse and in parts they convey similar contents.

The following list provides the Bayesian \cp papers that are most relevant to this section. 
\cite{fearnhead_online, ocpd} and \cite{lai2011simple}, elaborate, inter alia, online \cp detection algorithms.
We can also find retrospective inference strategies like exact sampling \citep{exact_fearnhead}, \cp entropy \citep{guedon2015segmentation}, \cp histograms \citep{nam2012quantifying, Rigaill2012, Turner_gaussianprocesses, aston2012implied}, and approaches for parameter estimation \citep{lai2011simple, yildirim2013online, bansal2008application}.

On this basis, we will develop novel inference algorithms.
Though, in order to provide a comprehensive understanding of Bayesian \cp analysis and its capabilities, we will implement the bulk of the existing approaches into our own lightweight and simple notational framework.
This framework resembles that of \cite{fearnhead_online} most.

Meanwhile, we will keep an eye on inference in hidden Markov models \citep{rabiner1986introduction}, hidden semi-Markov models \citep{hsmm} and the Kalman-filter \citep{kalman_filter, rauch_striebel_tung}.
This helps clarifying that inference in Bayesian \cp models is closely related to inference in these more basic models.
General terms for corresponding inference strategies comprise Bayes filtering and smoothing \citep{sarkka2013bayesian}.

The basic innovation provided in this section concerns the development of very efficient pointwise inference.
This paves the way for the computation of a wide range of pointwise expectations w.r.t. all timepoints at once with an overall linear complexity.

These pointwise statements greatly facilitate the understanding of posterior distributions and enable a novel EM algorithm \citep{emalgorithm}.
This algorithm finds local maximizers of the likelihood marginalized over the unknown segmentation and is used for parameter estimation or model fitting in general.
Under good conditions, each EM step exhibits a linear space and time complexity.

The reader will learn how to derive the formulas needed to build an EM algorithm based on a broad class of \cp models.
These formulas employ elementary expectations, which may also be approximated through sampling via MCMC.
By this means, we are even able to maximize a pivotal likelihood locally without knowing its functional form.
Approximate sampling strategies related to \cp models via MCMC are discussed in Section \ref{chap:transdim}.

This whole section is accompanied by a Laplacian change in median example.
Though computationally and statistically attractive, the Laplace distribution appears to be fairly unconsidered within the Bayesian \cp community.
We make intense use of vivid visualization schemes to reveal its strengths.
It turns out that the Laplace distribution can be employed efficiently and suits \cp data that exhibits strong outliers.
In our example, these outliers even appear in a highly asymmetric fashion.
This poses no problem since the resulting posterior distributions address this by developing skewness.

The majority of the considered algorithms are dynamic programming algorithms \citep{bellman1966dynamic} implemented in purely discrete terms.
This fact accounts for their efficiency.
Under good conditions they exhibit a quadratic complexity w.r.t. the data size. 
By a slight approximation through a pruning approach, similar to that of \cite{exact_fearnhead}, the complexity may further reduce to linear.
This renders our algorithms fit for big data purposes.

In turn, the non-discrete part of our algorithms can be traced back completely to elementary integrals.
Therewith, our approach is accessible to all sorts of methods for recursive Bayesian estimation or numerical integration.

Our inference algorithms are applicable to a broad class of distributions, most notably the exponential family.
However, in order to draw inference in an exact manner, we need to put limitations on the \cp model.

We assume that the segment heights and lengths are mutually independent and that the observations are independent given the segment heights.
The independence assumptions imply that \cps partition the whole random process into independent parts.
This greatly reduces the computational complexity of the algorithms.

The structure of this section is as follows.
In Section \ref{chap:modas}, we put our Baysian \cp model and its random variables in concrete terms and display the corresponding graphical model.
Section \ref{chap:forwinf} develops the forward inference and discusses it in connection with the exponential family and the Laplace distribution. 
It further provides a pruning scheme and applies the results to the well-log data set.
Section \ref{chap:backward_inference} is concerned with backward inference.
Therefore, it develops algorithms to compute the MAP estimator, posterior samples and entropy.
Section \ref{chap:pointwise_spike} elaborates novel pointwise inference and applies its results to the well-log example.
Section \ref{chap:param} elaborates parameter estimation in \cp models via the likelihood.
To this end, it develops an EM algorithm and provides a thorough approach to estimate certain parameters in the well-log example.
Thereafter, the EM algorithm will be scrutinized in a more frequentist fashion.
Section \ref{chap:discussion_spike} concludes with a discussion.

\subsection{Model assumption}
\label{chap:modas}

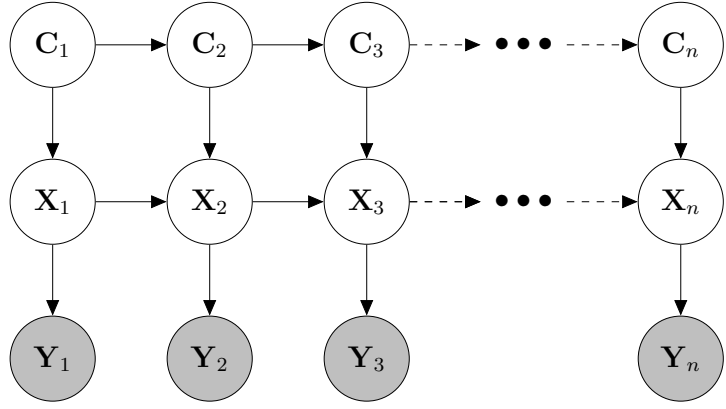
\begin{figure}[ht]
	\centering
	\resizebox{0.6\textwidth}{!}{%
		\begin{tikzpicture}[inner sep=-1mm, minimum size=12mm,
		obs/.style={circle,draw,fill=gray!50},
		rand/.style={circle,draw},
		arrow/.style={arrows={-triangle 45}} 
		]
		]
		\node (R1) [rand] {$\rand C_1$};
		\node (X1) [rand, below= of R1] {$\rand X_1$};
		\node (Y1) [obs, below= of X1] {$\rand Y_1$};
		\node (R2) [rand, right=of R1] {$\rand C_2$};
		\node (X2) [rand, below=of R2] {$\rand X_2$};
		\node (Y2) [obs, below=of X2] {$\rand Y_2$};
		\node (R3) [rand, right=of R2] {$\rand C_3$};
		\node (X3) [rand, below=of R3] {$\rand X_3$};
		\node (Y3) [obs, below= of X3] {$\rand Y_3$};
		\node (RR) [circle, right=of R3] {$\enspace\bullet\bullet\bullet\enspace$};
		\node (XX) [circle, right=of X3] {$\enspace\bullet\bullet\bullet\enspace$};
		%		\node (YY) [circle, right=of Y3] {$\enspace\bullet\bullet\bullet\enspace$};
		\node (Rn) [rand, right= of RR] {$\rand C_n$}; 
		\node (Xn) [rand, below= of Rn] {$\rand X_n$};
		\node (Yn) [obs,below=of Xn] {$\rand Y_n$};
		\draw [->, arrow] (X1.east) to (X2.west);
		\draw [->, arrow] (R1.east) to (R2.west);
		\draw [->, arrow] (X2.east) to (X3.west);
		\draw [->, arrow] (R2.east) to (R3.west);
		\draw [dashed,->, arrow] (X3.east) to (XX.west);
		\draw [dashed,->, arrow] (R3.east) to (RR.west);
		\draw [->, arrow] (X1.south) -- (Y1.north);
		\draw [->, arrow] (X2.south) -- (Y2.north);
		\draw [->, arrow] (X3.south) -- (Y3.north);
		\draw [->, arrow] (Xn.south) -- (Yn.north);
		\draw [->, arrow] (Rn.south) -- (Xn.north);
		\draw [->, arrow] (R1.south) -- (X1.north);
		\draw [->, arrow] (R2.south) -- (X2.north);		
		\draw [->, arrow] (R3.south) -- (X3.north);		
		%		\draw [->] (R1.south east) -- (Y1.north west);
		\draw [dashed,->, arrow] (X3.east) -- (XX.west);
		\draw [dashed,->, arrow] (XX.east) to (Xn.west); 
		\draw [dashed,->, arrow] (RR.east) to (Rn.west); 
		\end{tikzpicture}
	}
	\caption[Graphical representation of the \cp model]{The Bayesian network of our \cp model. 
		Each node represents a random variable and the arrows determine how the model factorizes into the initial distributions.
		The gray nodes highlight the random variables that are observed.}
	\label{fig:inferencecpmodel}
\end{figure}% 	 

Let $(y_1,\ldots,y_n)$ represent our observed data and $t_1<\ldots<t_n$ the corresponding observational timepoints.
For the sake of simplicity, we assume that \cps can only occur within the set $\{t_1,\ldots,t_n\}$.
Consequently, we refer to these timepoints solely by indexes from 1 to n.

For the timepoints $i=1,\ldots,n$, the position of the most recent \cp prior to $i$ is represented by the discrete random variable $\rand C_i\in \{0,\ldots,i\}$, whereby $\rand C_i=0$ represents the case where no \cps occur within $\{1, \ldots, i\}$.

For $0\leq j< i\leq n$, we assume that there exist $0\leq\qcc_{ji}\leq 1$ with
\begin{align}
\label{eq:ciprobs}
&\prob{\rand C_i=j}=\begin{cases}
\qcc_{01}& i=1 \text{ and } j=0\\
\prob{\rand C_{i-1}=j}\qcc_{ji}&\text{otherwise}
\end{cases}
\end{align}
$\qcc_{ji}$ corresponds to $\prob{\rand C_i=j\mid \rand C_{i-1}=j}$ and $\prob{\rand C_i=i}$ ensues indirectly therefrom. 
The $\qcc_{ji}$'s can be derived from discrete length distributions as the complement of the hazard rate.
Those length distributions are allowed to be non-stationary w.r.t. the time and may also depend on the distances between successive timepoints.
When necessary, the $\qcc_{0i}$'s implement residual times.

Let $(\lat,\latsigma)$ and $(\obs, \obssigma)$ be standard Borel spaces.
The segment height at $i$ is modeled through the $(\lat,\latsigma)$-valued random variable $\rand X_i$ with
\begin{align}
\label{eq:segheightdist}
(\rand X_1\mid \rand C_1)\sim\jprior\quad\text{ and }\quad
\big(\rand X_i\mid \rand C_i=j, \rand X_{i-1}=x\big)\sim\begin{cases}
\jprior& j=i\\
\delta_x & j<i
\end{cases}
\end{align}
whereby $\jprior$ is a probability measure over $\latsigma$.
$\jprior$ may change its form w.r.t. $i$, but we leave this out to keep the notation simple.

Finally, the random variable $\rand Y_i$ which is defined over the sigma finite measure space $(\obs, \obssigma, \omeasure)$ represents the observation at timepoint $i$.
Its distribution is determined through the conditional density
\begin{align*}
\dens(\rand Y_{i}=y\mid\rand X_i=x)=\prob{\rand Y_i\in dy\mid \rand X_i=x}\densslash \psi(dy)
\end{align*}
for $y\in\obs$ and $x\in\lat$.

We will later use these densities to derive conditional distributions based on the observations $(y_1,\ldots,y_n)$, whereby we
assume that $\int \prod_{\ell=j}^{i}\dens(\rand Y_\ell=y_\ell\mid \rand X_\ell=x)\jprior(dx)>0$ for all $0<j\leq i\leq n$.
This sort of elementary integrals are pivotal for the feasibility of the inference algorithms in this whole section.

The graphical model depicted in Figure \ref{fig:inferencecpmodel} reiterates how the initial distributions are composed in order to form the complete \cp model.  
It is of high value, because it allows us to quickly grasp conditional independence properties \citep{pearl1998graphical, bishop2014}.
In the following, it is assumed that the reader is capable of deducing these properties himself.

Nevertheless, context specific independence properties \citep{boutilier1996context} result from the concrete distributions in use and can't be read from this graph.
In our case, the most important is the independence induced by \cps, i.e.
\begin{align*}
(\rand X_{1:i-1}, \rand Y_{1:i-1}, \rand C_{1:i-1}\indep \rand X_{i:n}, \rand Y_{i:n}, \rand C_{i+1:n}\mid \exists\ \ell\geq i:\ \rand C_\ell=i)
\end{align*}
for $i=2,\ldots,n$.

\subsection{Forward inference}
\label{chap:forwinf}

Forward inference processes the data in a one by one fashion, examples are the forward algorithm in hidden Markov models \citep{rabiner1986introduction} or update and prediction in the Kalman-filter \citep{kalman_filter}.
These algorithms can be summarized by the term Bayes filter \citep{sarkka2013bayesian}.
\cite{ocpd, fearnhead_online} and \cite{lai2011simple} elaborate forward inference in \cp models.

It is used when data comes available in real time and needs to be processed straight away, e.g. for online detection of \cps.
In addition, it provides the foundation for the remaining inference.

The following two quantities are pivotal for forward inference. 
For $0<i\leq n$ and $0\leq j\leq i$ define
\begin{align}
\label{eq:segc}
&\ccc_{ji}=\prob{\rand C_i=j\mid \rand Y_{1:i}=y_{1:i}}\\
\label{eq:heightdens}
&\heights_{ji}=\prob{\rand X_i\in\LargerCdot\mid \rand C_i=j, \rand Y_{1:i}=y_{1:i}}
\end{align}
$\ccc_{ji}$ as well as $\heights_{ji}$ make statements concerning the $i$-th timepoint, whereby they incorporate all available observations hitherto.
$\ccc_{ji}$ expresses the probability that the most recent \cp prior to $i$ is $j$.
$\heights_{ji}$ is the distribution of the segment height at $i$ given that the segment starts at $j$.
In both cases, the actual right end of the segment that contains $i$ and all the subsequent observations remain unconsidered.

At first, we deal with the $\ccc_{ji}$'s.
They rely on the important auxiliary quantities 
\begin{align}
\label{eq:sego}
&\rcc_{ji}=\prob{\rand Y_i\in dy_i\mid \rand C_i=j,\rand Y_{j:i-1}=y_{j:i-1}}\densslash\psi(dy_i)
\end{align}
for $0<i\leq n$ and $0\leq j\leq i$.
We will justify their existence and elaborate their computation in Lemma \ref{lem:bayesstepfor} and Lemma \ref{lem:segheight} later on.
Given the $\rcc_{ji}$'s, the computation of the $\ccc_{ji}$'s is subject to a convenient iteration scheme as the following lemma shows.

\begin{lemma}
\label{lem:online}
For $0<i\leq n$ and $0\leq j \leq i$ we conclude that
\begin{align*}
&\ccc_{ji}=\frac{\rcc_{ji}}{\Zcc_i}\begin{cases}
\qcc_{ji}\ccc_{ji-1}& j<i\\
1-\sum_{j=0}^{i-1}\qcc_{ji}\ccc_{ji-1}& j=i
\end{cases}\\
&\Zcc_i=\Bigg(\sum_{j=0}^{i-1} \rcc_{ji}\qcc_{ji}\ccc_{ji-1}\Bigg)+\rcc_{ii}\Big(1-\sum_{j=0}^{i-1}\qcc_{ji}\ccc_{ji-1}\Big)
\end{align*}
with $\ccc_{00}=1$.
\end{lemma}
\begin{proof}
For $0< j<i\leq n$, since $\qcc_{ji}=\prob{\rand C_i=j\mid \rand C_{i-1}=j}$ and $(\rand C_i \indep \rand Y_{1:i-1}\mid \rand C_{i-1})$, we see that
\begin{align*}
&\prob{\rand C_i=j\mid \rand Y_{1:i-1}}
=\prob{\rand C_{i-1}=j\mid \rand Y_{1:i-1}}\prob{\rand C_i=j\mid \rand C_{i-1}=j}=\ccc_{ji-1}\qcc_{ji}
\end{align*}
The proof now follows from $\prob{\rand C_i=j, \rand Y_i\in dy_i\mid \rand Y_{1:i-1}=y_{1:i-1}}\densslash \omeasure(dy_i)
=\ccc_{ji-1}\qcc_{ji}\rcc_{ji}$, whereby we have used that $(\rand Y_i \indep \rand Y_{1:j-1}\mid \rand C_{i}=j)$.
\end{proof}

The proof of Lemma \ref{lem:online} reveals that the computation of $\ccc_{ji}$ can be divided into two steps.
In the first step, we go from $i-1$ to $i$ through $\qcc_{ji}$ without considering $y_i$.
In the second step, we incorporate $y_i$ through $\rcc_{ji}$ and normalize with $\Zcc_i$.
These two steps are an analogy of the prominent prediction and update in the Kalman-filter.

\begin{myremark}
\label{rem:marglikelihood}
The value of $\Zcc_i$ in Lemma \ref{lem:online} corresponds to 
\begin{align*}
\prob{\rand Y_i\in dy_i\mid \rand Y_{1:i-1}=y_{1:i-1}}\densslash\psi(dy_i)
\end{align*}
It can be used to compute the \emph{marginal loglikelihood} of the observations marginalized over all segment heights and \cp positions, i.e.
\begin{align*}
\ln\Big(\prob{\rand Y_{1:n}\in dy_{1:n}}\densslash\psi(dy_{1:n})\Big)
=\sum_{i=1}^n\ln(\Zcc_i)
\end{align*}
Their existence follows from the existence of the $\rcc_{ji}$'s.
\end{myremark}

\begin{algorithm}[ht]
	\begin{algorithmic}[1]
		%\scriptsize
		\footnotesize 
		\State{\textbf{Input:} Data $(y_1,\ldots,y_n)$}
		\State{\textbf{Output:} $\Zcc_i$ and $\ccc_{ji}$ for $0<i\leq n$ and $0\leq j\leq i$}
		\State{$\ccc_{00}=1$}
		\For{$i=1,\ldots,n$}		
		\State{$\ell=0$, $\Zcc_i=0$}
		\For{$j=0,\ldots,i-1$}
		\State{Compute $\rcc_{ji}$}
 		\Comment{Complexity $\mathcal{O}(h)$}
 		\label{pseudo:forward:bayestep1}
		\State{$\ell=\ell+\ccc_{ji-1}\qcc_{ji}$, $\quad\dcc_{ji}=\ccc_{ji-1}\qcc_{ji} \rcc_{ji}$, $\quad \Zcc_i=\Zcc_i+\dcc_{ji}$ } \label{pseudo:forward:update}
		\EndFor
		\State{Compute $\rcc_{ii}$}
 		\Comment{Complexity $\mathcal{O}(h)$}
 		\label{pseudo:forward:bayestep2}
		\State{$\dcc_{ii}=(1-\ell)\rcc_{ii}$, $\quad \Zcc_i=\Zcc_i+\dcc_{ii}$}
		\For{$j=0,\ldots,i$}
		\State{$\ccc_{ji}=\dcc_{ji}\slash \Zcc_i$}\label{pseudo:forward:obspost}
		\EndFor\label{pseudo:forward:laststep}	
		\EndFor	
		%\EndFunction
	\end{algorithmic}
	\caption{Forward inference}
	\label{pseudo:forward}
	%{Implementation of \greedy with a runtime complexity of $\mathcal{O}(n^2m\ln(m))$.}
\end{algorithm}

Pseudocode \ref{pseudo:forward} demonstrates the resulting forward inference algorithm.
Assuming that the computation of $\rcc_{ji}$ and  $\rcc_{ii}$ in step \ref{pseudo:forward:bayestep1} and \ref{pseudo:forward:bayestep2} has in both cases a space and time complexity of $\mcal{O}(h)$, the overall runtime complexity corresponds to $\mcal{O}(hn^2)$ and the space complexity to $\mcal{O}(n^2+h)$.

Now, we elaborate two different approaches to compute the $\heights_{ji}$'s and $\rcc_{ji}$'s.
As we have seen in Lemma \ref{lem:online}, $\rcc_{ji}$ enables inference over \cp locations.
In turn, $\heights_{ji}$ can be used for inference over segment heights.
The first approach takes advantage of certain elementary integrals, i.e. the segmental marginal likelihoods, as the following lemma shows.

\begin{lemma}
	\label{lem:segheight}
	For $0< j\leq i\leq n$ the following holds
	\begin{align}
	\label{eq:integrateocc}
	\rcc_{ji}&=\frac{\int \prod_{\ell=j}^{i}\dens(\rand Y_\ell=y_\ell\mid \rand X_\ell=x)\jprior(dx)}{\int \prod_{\ell=j}^{i-1}\dens(\rand Y_\ell=y_\ell\mid \rand X_\ell=x)\jprior(dx)}\\
	\label{eq:integratehcc}
	\heights_{ji}(dx)&=\frac{\prod_{\ell=j}^i\dens(\rand Y_\ell=y_\ell\mid \rand X_\ell=x)\jprior(dx)}{\int\prod_{\ell=j}^i\dens(\rand Y_\ell=y_\ell\mid \rand X_\ell=x)\jprior(dx)}
	\end{align}
	whereby $\rcc_{0i}=\rcc_{1i}$. 
\end{lemma}
\begin{proof}
	(\ref{eq:integrateocc}) follows from
	\begin{align*}
	\rcc_{jj}\cdot\ldots\cdot\rcc_{ji}=\prob{\rand Y_{j:i}=dy_{j:i}\mid \rand C_i=j}\densslash\psi(dy_{j:i}) 
	=\int \prod_{\ell=j}^{i}\dens(\rand Y_\ell=y_\ell\mid \rand X_\ell=x)\jprior(dx)
	\end{align*} 
	for all $0<i\leq n$ and $0\leq j\leq i$.
	Hereby, we have used the requirement of Section \ref{chap:modas} that $\int \prod_{\ell=j}^{i}\dens(\rand Y_\ell=y_\ell\mid \rand X_\ell=x)\jprior(dx)>0$ for all $0<j\leq i\leq n$.
	This is necessary to ensure the existence of the $\rcc_{ji}$'s.

	Equation (\ref{eq:integratehcc}) can now be proved through induction by means of Lemma \ref{lem:bayesstepfor}. 
	In the base case, where $i=j-1$, we set $\heights_{jj-1}$ to $\jprior$.
	In the induction step from $i-1$ to $i$ we divide $\dens(\rand Y_i=y_i\mid \rand X_i=x)\heights_{ji-1}(dx)$ by $\rcc_{ji}$ to obtain $\heights_{ji}(dx)$.
\end{proof}

The integrals employed in Lemma \ref{lem:segheight} cannot per se be reused in an iterative fashion and might have to be computed every time anew. 
Since they involve all the observations from $j$ and $i$, their computation then exhibits a non-constant complexity w.r.t. $n$ increasing the complexity of the forward algorithm by at least a factor of $n$.

The second approach utilizes the following relation in order to compute the $\heights_{ji}$'s and $\rcc_{ji}$'s in an iterative fashion.

\begin{lemma}
\label{lem:bayesstepfor}
For $0<i\leq n$ and $0\leq j\leq i$ 
\begin{align}
\label{eq:forbayestheo}
&\rcc_{ji}\cdot \heights_{ji}(dx)=\dens(\rand Y_i=y_i\mid \rand X_i=x)\cdot \heights_{ji-1}(dx)
\end{align}	
with $\heights_{ii-1}=\heights_{00}=\jprior$.
\end{lemma}
\begin{proof}
At first, we show that $(\rand X_{i}\mid \rand C_{i}=j, \rand Y_{1:i-1}=y_{1:i-1})\sim \heights_{ji-1}$ for $0<j<i\leq n$.
Therefore, consider that $(\rand X_{i-1}\indep \rand C_{i}\mid \rand C_{i-1}, \rand Y_{1:i-1}=y_{1:i-1})$ and 
\begin{align*}
&(\rand X_{i}\mid \rand C_{i}=j, \rand Y_{1:i-1}=y_{1:i-1})
\eqdist(\rand X_{i-1}\mid \rand C_i=j, \rand C_{i-1}=j, \rand Y_{1:i-1}=y_{1:i-1})\\
&\eqdist(\rand X_{i-1}\mid \rand C_{i-1}=j, \rand Y_{1:i-1}=y_{1:i-1})\sim\heights_{ji-1}
\end{align*}
With this and $(\rand Y_i\indep \rand C_i, \rand Y_{1:i-1}\mid \rand X_i)$, we may conclude that 
\begin{align*}
&\dens(\rand Y_i=y_i\mid \rand X_i=x)\heights_{ji-1}(dx)=\prob{\rand X_i\in dx, \rand Y_i\in dy_i\mid \rand C_{i}=j, \rand Y_{1:i-1}=y_{1:i-1}}\densslash\omeasure(dy_i)\\
&=\prob{\rand Y_i\in dy_i\mid \rand C_{i}=j, \rand Y_{1:i-1}=y_{1:i-1}}\densslash\omeasure(dy_i)\cdot \heights_{ji}(dx)=\rcc_{ji}\cdot \heights_{ji}(dx)
\end{align*}
for $0<i\leq n$ and $0\leq j\leq i$.
\end{proof}

Lemma \ref{lem:bayesstepfor} reveals that $\heights_{ji}$ and $\rcc_{ji}$ can be derived from $\dens(\rand Y_i=y_i\mid \rand X_i=x)$ and $\heights_{ji-1}$.
For this sake, in step \ref{pseudo:forward:bayestep1} and \ref{pseudo:forward:bayestep2} of Pseudocode \ref{pseudo:forward}, we need to store $\heights_{ji}$.
It will then be reused to compute $\heights_{ji+1}$ and $\rcc_{ji+1}$.
$\heights_{ii}$ and $\rcc_{ii}$, in turn, are derived from $\dens(\rand Y_i=y_i\mid \rand X_i=x)\jprior(dx)$.

Many distributions, most namely the members of the exponential family, support these operations in closed form, for example, through conjugacy \citep{raiffa1961applied}.
The stepwise evolution of the posterior distributions leaves their functional form invariant and only requires a transformation of the parameters.
Consequently, these distributions are stored by means of their parameters.
This transformation can usually be done in constant time w.r.t. $n$ leading to a forward algorithm with an overall quadratic complexity.

Finally, the joint distribution of the segment height at $i$ and most recent \cp prior to $i$ conditioned on the observations up to $i$ can be expressed through
\begin{align*}
%\label{eq:jointonline}
&\prob{\rand X_i\in dx, \rand C_i=j\mid \rand Y_{1:i}=y_{1:i}}= \ccc_{ji}\heights_{ji}(dx)
\end{align*}

\subsubsection{Pruning}
\label{chap:pruning}

With pruning we refer to the task of discarding certain computation steps w.r.t. $j$ in Pseudocode \ref{pseudo:forward}.
Having skipped the computation of, say $\ccc_{ji}$, from then on we will also skip all further computation steps that involve a \cp at $j$, i.e. we skip the computation of $\ccc_{j\ell}$ for all $i<\ell\leq n$.
In the optimal case, the computational complexity reduces thereby to linear w.r.t. $n$.

Imagine that at timepoint $i$ it is certain that the most recent \cp prior to $i$ is $j$.
This means that the most recent \cp prior to $i$ cannot be at a timepoint prior to $j$.
Hence, all further computations involving a \cp prior to $j$ may safely be skipped.
In this case, pruning doesn't even cause inaccuracies at all.
This resembles the \algor{PELT} method as discussed in Section \ref{chap:pencost}.

However, since \cps are normally not certain, a worthwhile pruning approach introduces inaccuracies to the distribution of \cps.
In order to keep these inaccuracies small, we develop a method that adjusts the intensity of pruning dynamically.
Our approach is similar to the method proposed in \cite{exact_fearnhead}.

In accordance to the particles in particle filter approaches \citep{doucet2001introduction}, we call the quantity $\ccc_{ji}$ the \emph{$(j,i)$-th particle}.

At first, we introduce a certain minimal removal age $T$, i.e. particles $(j,i)$ with $i-j<T$ must be retained.
Therewith, we address the following issue.
A timepoint $j$ that doesn't appear to be the most recent \cp prior to $i$ can later gain more credibility.
This is because inconspicuous changes in segment height require more datapoints to be processed in order to reveal their existence.

Particles whose age is at least $T$ can be removed on the fly.
Therefore, we employ another threshold $T'$, say $T'=10^{-15}$.
If the relative contribution of a particle to the current mass is smaller than $T'$, it will be removed.

\begin{algorithm}[ht]
	\begin{algorithmic}[1]
		%\scriptsize
		\footnotesize 
		\State{\textbf{Input:} Data $(y_1,\ldots,y_n)$ and thresholds $T>1$ and $T'\geq 0$}
		\State{\textbf{Output:} $\Zcc_i$ and a sparse representation of $\ccc_{ji}$ for $0<i\leq n$ and $0\leq j\leq i$}
		\State{$\ccc_{00}=1$}
		\For{$i=1,\ldots,n$}		
		\State{$\ell=0$, $\Zcc_i=0$}
		\For{$j$ such that $c_{ji-1}$ exists}
		\State{Compute $\rcc_{ji}$}\label{pseudo:forwardprun:bayestep1}
		\Comment{Complexity $\mathcal{O}(h)$}
		\State{$\dcc_{ji}=\ccc_{ji-1}\qcc_{ji} \rcc_{ji}$}
		\If{$i-j< T$ or $\dcc_{ji}\geq \Zcc_i\cdot T'$}\label{pseudo:forwardprun:prun}		
		\Comment{Pruning condition}		
		\State{Declare $\ccc_{ji}$ and set $\Zcc_i=\Zcc_i+\dcc_{ji}$ and $\ell=\ell+\ccc_{ji-1}\qcc_{ji}$}		
		\EndIf
		\EndFor\label{pseudo:forwardprun:endfor}
		\State{Compute $\rcc_{ii}$ and declar $\ccc_{ii}$}
		\Comment{Complexity $\mathcal{O}(h)$}
		\label{pseudo:forwardprun:bayestep2}
		\State{$\dcc_{ii}=(1-\ell)\rcc_{ii}$ and \textbf{}$\Zcc_i=\Zcc_i+\dcc_{ii}$}
		\For{$j$ such that $\ccc_{ji}$ exists}
		\State{$\ccc_{ji}=\dcc_{ji}\slash \Zcc_i$}\label{pseudo:forwardprun:obspost}
		\EndFor\label{pseudo:forwardprun:laststep}	
		\EndFor	
		%\EndFunction
	\end{algorithmic}
	\caption{Forward inference with pruning}
	\label{pseudo:forwardprun}
	%{Implementation of \greedy with a runtime complexity of $\mathcal{O}(n^2m\ln(m))$.}
\end{algorithm}
Pseudocode \ref{pseudo:forwardprun} shows the pseudocode of the pruned forward inference algorithm.
The pruning conditions are implemented in Step \ref{pseudo:forwardprun:prun}.
$\dcc_{ji}\geq \Zcc_i\cdot T'$ states that the $(j,i)$-th particle is retained if the relative contribution of $\dcc_{ji}$ to $\Zcc_i$ is at least $T'$.

Thus, our approach is capable of removing particles directly when they are dealt with justifying the term \emph{on the fly pruning}.
Therewith, Step \ref{pseudo:forwardprun:prun} adjusts the number of pruned particles dynamically.
If there is sufficient certainty over the locations of the last \cps, earlier particles will not pass anymore.
Hence, the efficiency of pruning is improved each time a new \cp location pops up.
However, it will refuse to prune if no new \cp locations are spotted anymore.

In contrast to the non-pruned algorithm, here we use a data structure that stores exclusively the existing particles.
It needs to be capable of iterating through the particles $(j,i)$ w.r.t. $j$, but an index based access is not required.
This can be achieved efficiently through lists or a sparse matrix implementation.

The actual complexity of the pruning scheme relies on how the number of \cps increases with $n$.
In the worst case, where the number of \cps stays constant or the model is extremely uncertain about the \cp locations, this algorithm may not improve the complexity whatsoever.

However, in practice the number of \cps usually grows with $n$.
In this case, if we hold $j$ and increase $i$, more and more reasonable \cp locations between $j$ and $i$ will pop up.
This will eventually sort out $j$ and also any timepoint before $j$ to be the most recent \cp prior to $i$.

Hence, a sufficient growth in the number of \cps in relation to $n$, e.g. linear, yields an upper bound for the segment length w.r.t. the retained particles.
In this case, even if the complexity of computing $\rcc_{ji}$ and $\heights_{ji}$ is a function of the segment length, $i-j$, our pruning approach is able to achieve a time and space complexity of $\mathcal{O}(n)$ for Pseudocode \ref{pseudo:forwardprun}.

\subsubsection{The exponential family}
\label{chap:conjugacy}

Now, we show how to compute the $\rcc_{ji}$'s and $\heights_{ji}$'s for a certain but rich family of distributions.
This amounts to improved versions of Pseudocode \ref{pseudo:forward} and \ref{pseudo:forwardprun}.

Let $\obs\sub\mathbb{R}^r$ and $\lat\sub\mathbb{R}^s$, and $\funcf(y\mid x)=\dens(\rand Y_i=y\mid \rand X_i=x)$ for all $i=1,\ldots,n$.
$\funcf$ may also depend on $i$, but we skip this to keep the notation simple.
Further assume that $\jprior$ exhibits a density $\funcg$ w.r.t. a measure $\xi$ over $\latsigma$.
The functional form of $\funcf$ and $\funcg$ is defined through
\begin{align*}
&\funcf(y\mid x)=h(y)\euler{\eta(x)^t T(y)-A(x)}\\
&\funcg(x\semic \nu, \chi)=\tilde h(x) \euler{\chi^t\eta(x)-\nu A(x)-\tilde A(\nu,\chi)}
\end{align*}
whereby $T:\obs\rightarrow\mathbb{R}^u$, $\eta:\lat\rightarrow\mathbb{R}^u$, $h:\obs\rightarrow\mathbb{R}$, $\tilde h:\lat\rightarrow\mathbb{R}$ and $A:\lat\rightarrow\mathbb{R}$ are measurable functions.
$\nu\in\mathbb{R}$ and $\chi\in\mathbb{R}^u$ are pivotal parameters.
The functions $A:\mathbb{R}^s\rightarrow\mathbb{R}$ and $\tilde A:\mathbb{R}^{u+1}\rightarrow\mathbb{R}$ ensure that the densities $\funcf$ and $\funcg$ are correctly normalized.
Additionally, we require that $h(y_i)>0$ for all $i=1,\ldots,n$.

$\funcf$ and $\funcg$ are members of the exponential family w.r.t. to $x$ and $(\nu, \chi)$, respectively.
Moreover, $\funcg$ is conjugate to $\funcf$.
This means that applying Bayes' theorem \citep{bayes1763lii} to these densities results in a density that is of the same form as $\funcg$, but with updated parameters.
More precisely
\begin{align}
\nonumber
&\funcf(y\mid x)\cdot\funcg(x\mid \nu,\chi)\\
\label{eq:onlinebayesstep}
&=\funcg(x\semic \nu+1,T(y)+\chi)\cdot h(y)\euler{-\tilde A(\nu,\chi)+\tilde A(\nu+1,T(y)+\chi)}
\end{align}

\begin{lemma}
For $0<j\leq i\leq n$ there exist $\chi_{ji}\in\mathbb{R}^s$ and $\nu_{ji}\in\mathbb{R}$ with 
\begin{align*}
\heights_{ji}(dx)=\funcg(x\semic \nu_{ji}, \chi_{ji})\xi(dx)
\end{align*}
These parameters can be computed recursively through $\nu_{ji}=\nu_{ji-1}+1$ and $\chi_{ji}=\chi_{ji-1}+T(y_i)$, whereby $\nu_{ii-1}=\nu_{00}=\nu$ and $\chi_{ii-1}=\chi_{00}=\chi$.
Furthermore, 
\begin{align*}
&\rcc_{ji}=h(y_i)\euler{-\tilde A(\nu_{ji-1},\chi_{ji-1})+\tilde A(\nu_{ji-1}+1,T(y_i)+\chi_{ji-1})}
\end{align*}
\end{lemma}
\begin{proof}
These inference formulas are a direct consequence of Lemma \ref{lem:bayesstepfor}.
The existence of $\heights_{ji}(dx)\densslash\xi(dx)$ follows from the fact that $\int \funcf(y_i\mid x)\funcg(x\mid \nu,\chi)\xi(dx)>0$ for all $\nu\in\mathbb{R}$, $\chi\in\mathbb{R}^u$ and $i=1,\ldots,n$.
\end{proof}

This allows for a recursive computation of $\rcc_{ji}$ that can easily be incorporated into Pseudocode \ref{pseudo:forward}. 
There, we just need to modify Step \ref{pseudo:forward:bayestep1} and \ref{pseudo:forward:bayestep2} to compute and safe $\nu_{ji}$, $\chi_{ji}$.
This has a constant complexity w.r.t. $n$ and thus, the overall space and time complexity of the unpruned forward algorithm becomes $\mcal{O}(n^2)$.

There are useful summaries for concrete members of the exponential family available, e.g. \cite{wiki:conjugacy}.
In this context, the quantities involved in Lemma \ref{lem:bayesstepfor} receive particular names.
$\heights_{ji-1}$ is called (conjugate) prior, \funcf likelihood, $\heights_{ji}$ posterior and $\rcc_{ji}$ posterior predictive or compound.

\subsubsection{The Laplacian change in median model}
\label{chap:laplacian}

Let now $\obs=\lat=\mathbb{R}$ and $\omeasure(dy)=dy$, and
\begin{align*}
\jprior(dx)&=\frac{1}{2\tau}\euler{-\frac{|x-\mu|}{\tau}}dx\\
\dens(\rand Y_i=y\mid \rand X_i=x)&=\frac{1}{2\sigma}\euler{-\frac{|y-x|}{\sigma}}
\end{align*}
with $\sigma, \tau>0$, $\mu\in\mathbb{R}$.
Therewith, the segment heights at \cp locations and observations follow a Laplace distribution with scale $\tau$ and $\sigma$, respectively.
Each observational distribution ties its median to the corresponding segment height and the segment heights themselves employ $\mu$ as their median.
Thus, the \cp model infers changes in median, or more generally speaking, changes in location.
All parameters are allowed to depend on $i$, but we leave this out to keep the notation simple.

$\jprior$ effectively contributes an observation at $\mu$ to each segment, however with a particular scale $\tau$. 
These two parameters should be chosen such that $\jprior$ easily covers all likely segment heights, but also inhibits unlikely ones.
Similarly, $\sigma$'s task is to promote and prevent joint accommodation of observations within single segments and across multiple segments, respectively.

\cite{jafari2016bayesian} considers the Laplacian change in median model for the single \cp case.
Apart form that, this model seems to be disregarded within the Bayesian \cp literature.

However, it is of interest, because it is more robust against outliers than, for example, the normal distribution, 
but computationally more attractive than, for example, Student's t-distribution.
The Laplace distribution serves as an example for a distribution that is not part of the exponential family w.r.t. its location parameter.

Now, we elaborate an analytical way to compute the $\rcc_{ji}$'s and $\heights_{ji}$'s in this model.
Therefore, let 
\begin{align}
\label{eq:laplaceexp}
&\Zcc_{ji}^m=\int x^m\euler{-\frac{|x-\mu|}{\tau}-\sum_{\ell=j}^i\frac{|y_\ell-x|}{\sigma}}dx
\end{align}
for $0<j\leq i\leq n$.
Lemma \ref{lem:segheight} shows that for $0<j\leq i\leq n$ 
\begin{align}
\label{eq:laplaceo}
&\rcc_{ji}=\frac{\Zcc_{ji}^0}{2\sigma\Zcc_{ji-1}^0}\quad \text{and}\quad
\heights_{ji}(dx)=\frac{1}{\Zcc_{ji}^0}\euler{-\frac{|x-\mu|}{\tau}-\sum_{\ell=j}^i\frac{|y_\ell-x|}{\sigma}}dx
\end{align}
whereby $\Zcc^0_{00}=2\tau$.
Hence, the computation of the ratio $\rcc_{ji}$ and $\heights_{ji}$ comes down to computing the integral in Equation (\ref{eq:laplaceexp}) for $m=0$. 
Key to this computation is the fact that $-\frac{|x-\mu|}{\tau}-\sum_{\ell=j}^i\frac{|y_\ell-x|}{\sigma}$ can be described through a continuous piecewise linear function.

To this end, let $(z_1, \ldots, z_r)$ be a sequence of real numbers which is sorted in ascending order and $(\sigma_1,\ldots,\sigma_r)$ a sequence of positive real numbers.
By decomposing each of the absolute value functions into their positive and negative parts we see that
\begin{align}
\label{eq:sumofabsfunc}
&-\sum_{\ell=1}^r\frac{|x-z_\ell|}{\sigma_\ell}
=\begin{cases}\dcc_{0}\cdot x+\ecc_{0}&x< z_1\\
\dcc_{j}\cdot x+\ecc_j&x\in[z_j,z_{j+1})\\
\dcc_{r}\cdot x+\ecc_r&x\geq z_r
\end{cases}
\end{align}	
with
\begin{align*}
\dcc_{j}=-\sum_{\ell=1}^j1\slash\sigma_\ell+\sum_{\ell=j+1}^r1\slash\sigma_\ell\quad \text{ and } \quad \ecc_j=\sum_{\ell=1}^jz_\ell\slash\sigma_\ell-\sum_{\ell=j+1}^rz_\ell\slash\sigma_\ell
\end{align*}
Moreover, due to the continuity mentioned above, we see that $\dcc_{\ell-1} z_{\ell}+\ecc_{\ell-1}=\dcc_{\ell} z_{\ell}+\ecc_{\ell}$ for all $\ell=1,\ldots,r$.

\begin{corollary}
	\label{cor:intlaplace}
	Let $(z_1,\ldots,z_r)$ and $(\sigma_1,\ldots,\sigma_r)$ be as before.
	For $m\in\mathbb{N}$, the following holds
	\begin{align}
	\label{eq:laplacesum}
	&\int x^m\euler{-\sum_{\ell=1}^r\frac{|x-z_\ell|}{\sigma_\ell}}dx
	=\sum_{\ell=1}^r\euler{\dcc_\ell z_\ell+\ecc_\ell}(\bcc_{\ell-1}^\ell-\bcc_{\ell}^\ell)
	\end{align}
	whereby for $j=0,\ldots,r$ and $k=1,\ldots,r$ we define
	\begin{align*}
	&\bcc_j^k=\begin{cases}
	\sum_{\ell=0}^m\frac{(-1)^{\ell}z_k^{m-\ell}m!}{\dcc_j^{\ell+1}(m-\ell)!}&\dcc_j\neq 0\\
	z_k^{m+1}\slash (m+1)&\dcc_j=0	
	\end{cases}
	\end{align*}
\end{corollary}
\begin{proof}
	Equation (\ref{eq:laplacesum}) follows from a piecewise integration w.r.t. to the partitioning of $\mathbb{R}$ that is implied by $(-\infty, z_1,\ldots,z_r, \infty)$.
	The $\bcc_j^k$'s are the result of a repeated integration by parts conducted on these pieces.
\end{proof}

Corollary \ref{cor:intlaplace} requires a sorting of the sequence $(\mu, y_1,\ldots,y_n)$.
Sorting can be done by successively adding observations.
Each insertion costs $\mcal{O}(\ln(n))$.
Alternatively, we may use the approach of \cite{weinmannl1}, which was developed in another context. It utilizes a certain data structure to set up all the sorted sequences beforehand.
This preprocessing step has a time complexity of $\mcal{O}(n^2)$ instead of $\mcal{O}(n^2\ln(n))$.
However, since this preprocessing thwarts the benefits of pruning, we won't consider it further.

The main drawback in using the Laplace distribution is that the integration concerning the $(j,i)$-particle involves a sum of $i-j+2$ summands.
Unfortunately, a stepwise computation in the fashion of Lemma \ref{lem:bayesstepfor} does not lead to any simplification here.
This is because the piecewise linear functions as described in (\ref{eq:sumofabsfunc}) change their appearance dramatically with new observations.
Hence, in each step we need to reconsider all previous observations within the concerned segment and not just the new one.
Therefore, the complexity of the unpruned forward algorithm is $\mcal{O}(n^3\ln(n))$ making \cp problems with $n>10^5$ a very challenging task.

However, the pruning scheme of Section \ref{chap:pruning} is able to reduce this complexity tremendously.
If we assume that there is a constant $\ell<\infty$ with $i-j<\ell$ for each existing particle $(j,i)$, the time and space complexity reduces to $\mcal{O}(n)$.

It is essential to improve the numerical properties of the integration by pulling out $\mcc_{ji}=\max_{x\in\mathbb{R}}\big\{-\frac{|x-\mu|}{\tau}-\sum_{\ell=j}^i\frac{|y_\ell-x|}{\sigma}\big\}$ or a similar value from the exponent of the r.h.s. of Equation (\ref{eq:laplacesum}).
It would then be reintroduced to $\rcc_{ji}$ by multiplying with $\exp(\mcc_{ji-1}-\mcc_{ji})$ in Equation (\ref{eq:laplaceo}).
$\heights_{ji}$ does not a require a reintroduction of $\mcc_{ji}$.
This is very effective in avoiding division by values close to zero.

Within the existing literature you can find several generalizations of the Laplace distribution, most notably a multivariate and an asymmetric \citep{kotz2012laplace}.
However, these two are different to the univariate extension that is represented by $\heights_{ji}$.

In order to facilitate the understanding in later chapters, we want to elaborate some of its properties briefly by means of an example.
Therefore, consider the unnormalized density $\exp\big(-\sum_{i=1}^5|x-z_i|\big)$ with 
$(z_1,\ldots,z_5)=(-7, -5, 0, 1.2, 1.3)$.

\begin{figure}[ht]
	\includegraphics[width=\linewidth]{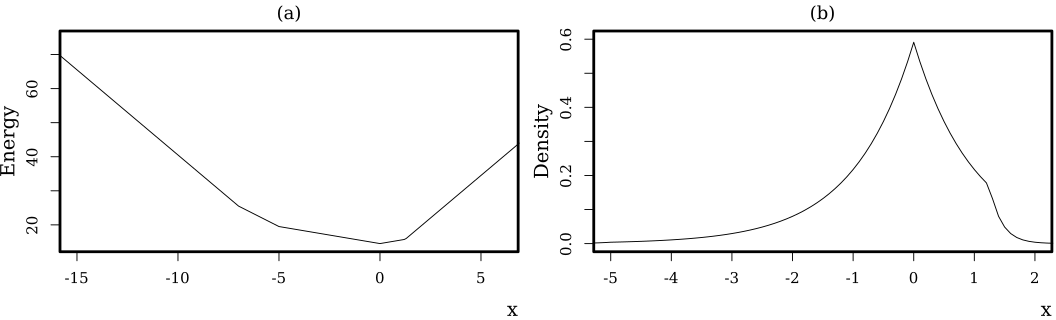}
	\caption[Graph of energy and density of generalized Laplace distribution]{(a) The graph of $\sum_{i=1}^5|x-z_i|$. (b) Normalized version of $\exp\big(-\sum_{i=1}^5|x-z_i|\big)$.
	}
	\label{fig:gen_laplace}
\end{figure}

$\sum_{i=1}^5|x-z_i|$ is depicted in Figure \ref{fig:gen_laplace}(a).
In the context of Gibbs measures \citep{georgii2011gibbs}, we would call it energy function.
In accordance with Equation (\ref{eq:sumofabsfunc}), it has the form of a continuous piecewise linear function being monotonically decreasing from $-\infty$ to the median of the $z_i$'s and monotonically increasing thereafter.
Before and after the possibly non-unique medians, we even get strong monotony.

If we follow the graph from the median to minus or plus $\infty$, the slope decreases or increases, respectively, each time we pass a $z_i$.
This yields a progressive increase of energy in both directions, which is controlled by the occurrences of the $z_i$'s.

However, the higher the energy the lower the likelihood, i.e. the density values.
The corresponding density, i.e. the normalized version of $\exp\big(-\sum_{i=1}^5|x-z_i|\big)$, is depicted in Figure \ref{fig:gen_laplace}(b).
Its single mode matches the median of $z_1,\ldots,z_5$.
We see that the quick occurrence of $z_4$ and $z_5$ immediately starts slashing the likelihood.
In turn, $z_1$ and $z_2$ are comparatively widespread, which yields a more gradual reduction of the likelihood.
However, in the long run, the left and right tails decay according to the same exponential distribution.

The standard deviation w.r.t. the density is 1.0208203.
Standard deviation measures the extend to which the distribution scatters around its expectation, but doesn't reveal anything about the possibly asymmetric spread.

For this purpose, we may use skewness, which amounts to -1.2708606.
Its negativity confirms the apparent tendency to produce larger values to the left than to the right.

Thus, we conclude that the $\heights_{ji}$'s define skewed distributions over $\mathbb{R}$.
The skewness is owed by the irregular spread of the observations around their median. 
Moderate observations significantly slash the likelihood in the direction of extremer observations and thus, thwart their impact.
$\heights_{ji}(dx)\densslash dx$ can be composed continuously by exponential functions and exhibits a single mode or a single plateau without further modes.

\subsubsection{Well-log example}
\label{chap:well_log_forward}

%\begin{align*}
%&(1-p)r\slash p=4049/k\\
%&r=p(4049/k+r)\\
%&p=r\slash (4049/k+r)
%\end{align*}

\begin{figure}[ht]
	\includegraphics[width=\linewidth]{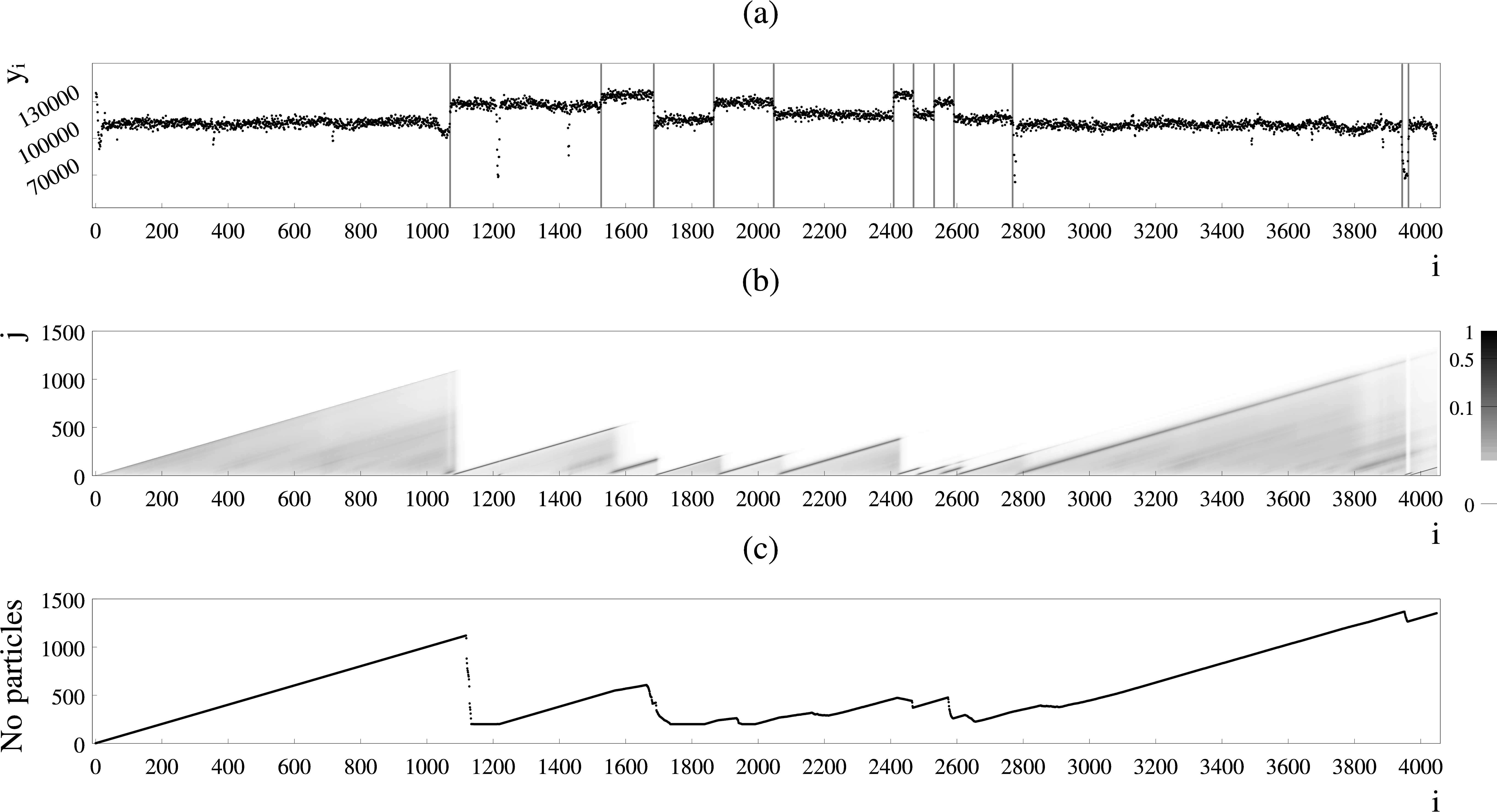}
	\caption[Well-log data, run-lengths and number of particles after pruning]{(a) Well-log data and \cp estimates (through MAP estimation) as gray vertical lines. (b) Density plot of the run-lengths. For each point $(i, j)$ it displays the value of $\ccc_{i-j i}$. (c) Illustration of the number of particles retained by the pruned forward algorithm.
	}
	\label{fig:laplace_forward}
\end{figure}
Now we consider a well-log dataset (see Figure \ref{fig:laplace_forward}(a)) that stems from nuclear-magnetic response of underground rocks \citep{fearnhead2003line, exact_fearnhead, fearnhead2019changepoint}.
Besides strong outliers, there are small and large changes in location present.
Our task is to build a \cp model that is sensitive towards the significant location changes that are marked by the gray vertical lines and specific otherwise.
Therefore, we employ the Laplacian change in median model of Section \ref{chap:laplacian}.
This example will further accompany us at several other occasions, i.e. Sections \ref{chap:well_log_forward}, \ref{chap:well-log-pointw}, \ref{chap:well-log-para}, \ref{chap:emsigmalap} and \ref{chap:well-logsimcred}.

We choose $\mu=113854$, $\tau=6879$ and $\sigma=25000$ and the time from one \cp to the next follows a negative binomial distribution with number of trials $r=3$ and success probability $q=0.01430724$.
In this section, we solely discuss the values of $\mu$ and $\tau$ in more detail.
The reasoning behind the other parameters requires a deeper understanding of \cp inference and will be carried out in Section \ref{chap:well-log-para}.

$\mu$ represents the median of all the possible segment heights.
However, not all datapoints are plausible candidates, in particular those heights indicated by outliers should be excluded.
Thus, the median of the whole data set, which is 113854, is a good choice here since it gives only little weight to outliers.

In a similar fashion, we choose the scale, $\tau$, as the mean absolute deviation around the median w.r.t. the whole data set.
As well as the median, it exhibits some robustness against outliers.
Thus, our approach to choose $\mu$ and $\tau$ functions poses a supportive mean against strong perturbations in the data.

For reasons to be elaborated in detail in Section \ref{chap:geomem}, we modify the length distribution of the first segment slightly.
Let $\text{NB}(\LargerCdot\semic q,r)$ represent the probability mass function of the negative binomial distribution with success probability $q$ and number of trials $r$.
We assume that $q\leq r\slash(r+1)$.
For $0<i\leq n$ we choose
\begin{align*}
\qcc_{ji}=\begin{cases}
\sum_{\ell=i-j}\text{NB}(\ell\semic q,r)\slash \sum_{\ell=i-j-1}\text{NB}(\ell\semic q,r)& 0<j<i\\
1-q'& j=0
\end{cases}
\end{align*}
whereby $q'=q\slash (r(1-q))$.
Thus, the length of the first segment now follows a geometric distribution with success probability $q'$.

Figure \ref{fig:laplace_forward}(b) visualizes the so-called \emph{run-lengths}, first introduced by \cite{ocpd},  as a density plot. 
For $(i,j)$ it depicts the probability that the most recent \cp prior to $i$ is $j$ timepoints away, that is $\ccc_{i-j i}$.
At any timepoint $i$, by following the oblique lines down to zero, we may encounter the location of the most recent \cp prior to $i$.

These locations usually appear slightly shifted to the right.
This is because, in the forward algorithm, at each timepoint $i$, we only process the datapoints until $i$ and not further.
It may thus take a little while before the model pays attention to a \cp.
Our visualization of the run-lengths can be of great help when \cps are to be detected on time.

Besides, this visualization technique can be used to get a rough picture of the models sensitivity and specificity towards \cps at interesting spots.
By this means, we may conclude that our Laplacian change in median model is pleasantly specific towards the outliers and small nuisance mean shifts present.
Though, some spots, most notably the perturbation at around 3800 that takes the form of a dip, could be of concern.

In our computation we used the pruning scheme of Section \ref{chap:pruning} with $T=200$ and $T'=10^{-15}$ to compute the $\ccc_{ji}$'s.
We implemented our algorithms in C++ and the whole computation takes around 26 seconds on an AMD FX-8350 with 4.2 GHz.
In comparison, for $T=4000$ the computation takes around 361 seconds.

Figure \ref{fig:laplace_forward}(c) depicts the evolution of the particles.
For each timepoint $i$, we see the number of particles that were retained by the pruning procedure. 
It shrinks dramatically in the presence of \cps, which gives our pruning approach its convincing efficiency.

\subsection{Backward inference}
\label{chap:backward_inference}

In the following, we consider retrospective or offline inference algorithms.
Intrinsic to these algorithms is that they reprocesses the quantities inferred by the forward inference algorithm in reverse order, starting at $n$.
Therewith they are an analogy to the backward algorithms in hidden Markov models \citep{rabiner1986introduction}.

Let
\begin{align*}
&\tccc_{ji}=\prob{\rand C_i=j\mid \rand C_{i+1}=i+1, \rand Y_{1:i}=y_{1:i}}\quad&&\text{ for }0<i< n\text{ and }0\leq j\leq i\\
&\tccc_{jn}=\ccc_{jn}&&\text{ for }j=0,\ldots,n
\end{align*}
Since \cps induce independence, we see that
\begin{align*}
\prob{\rand C_i=j\mid \rand C_{i+1}=i+1, \rand Y_{1:i}=y_{1:i}}=\prob{\rand C_i=j\mid \rand C_{i+1}=i+1, \rand Y_{1:n}=y_{1:n}}
\end{align*}
Thus, the $\tccc_{ji}$'s pose primal retrospective quantities that are used as a basic building block in our backward inference algorithms.

\begin{lemma}
For $0<i<n$ and $0\leq j\leq i$ we get
\begin{align*}
\tccc_{ji}=\frac{\ccc_{ji}(1-\qcc_{ji+1})}{\sum_{\ell=0}^i\ccc_{\ell i}(1-\qcc_{\ell i+1})}
\end{align*}
\end{lemma}
\begin{proof}
Since $(\rand C_{i+1}\indep \rand Y_{1:i}\mid \rand C_{i})$	we may conclude that
\begin{align*}
&\prob{\rand C_i=j, \rand C_{i+1}=i+1\mid \rand Y_{1:i}=y_{1:i}}\\
&=\prob{\rand C_i=j\mid \rand Y_{1:i}=y_{1:i}}\prob{\rand C_{i+1}=i+1\mid \rand C_i=j}
=\ccc_{ji}(1-\qcc_{j i+1})
\end{align*}	
A normalization w.r.t. $i$ yields the proof.
\end{proof}

In the pruned case, we solely use the existing particles to compute the $\tccc_{ji}$'s.
Thereby, we declare $\tccc_{ji}$ only if $\ccc_{ji}$ exists.
Their computation has thus a time and space complexity which is linear in the number of particles.

Similarly, we deal with the segment height at $i$ conditioned on the whole segment that contains $i$ and all observations. 
However, their distribution does not differ from those of the segment heights with unknown right end, as the next lemma shows.

\begin{lemma}
\label{lem:segheightsretro}
For $0<i\leq n$ and $0\leq j\leq i$ the following holds
\begin{align*}
(\rand X_i\mid \rand C_{i}=j, \rand C_{i+1}=i+1, \rand Y_{1:n}=y_{1:n})\sim \heights_{ji}
\end{align*}
\end{lemma}
\begin{proof}
Since \cps induce independence we get
\begin{align*}
(\rand X_i\mid \rand C_{i}=j, \rand C_{i+1}=i+1, \rand Y_{1:n}=y_{1:n})\eqdist(\rand X_i\mid \rand C_{i}=j, \rand C_{i+1}=i+1, \rand Y_{1:i}=y_{1:i})
\end{align*}
and finally $(\rand X_i\indep \rand C_{i+1}\mid \rand C_i, \rand Y_{1:i})$ yields
\begin{align*}
(\rand X_i\mid \rand C_{i}=j, \rand C_{i+1}=i+1, \rand Y_{1:i}=y_{1:i})\eqdist(\rand X_i\mid \rand C_{i}=j, \rand Y_{1:i}=y_{1:i})\sim\heights_{ji}
\end{align*}
\end{proof}

\subsubsection{Likelihood of \cp locations and MAP estimator}
\label{chap:condlikandmap}

The likelihood of a given set of \cps $0<\tau_1< \ldots<\tau_k\leq n$ conditioned on the whole set of observed data is expressed through 
\begin{align}
\label{eq:likelihoodcp}
\tccc_{0\tau_1-1}\Bigg(\prod_{i=1}^{k-1}\tccc_{\tau_i\tau_{i+1}-1}\Bigg)\tccc_{\tau_kn}
\end{align}
whereby $\tccc_{00}=1$.

The MAP (maximum a posteriori) estimator for the \cp locations can be obtained with the Viterbi style algorithm of Section \ref{chap:pencost}.
Hereby, we would simply choose $\cost(y_{j:i})=-\ln(\tccc_{ji})$ for $0\leq j\leq i\leq n$ with a dummy observation $y_0$ and $\gamma=0$.
This provides an element of
\begin{align}
\label{eq:map}
\argmax_{c_{1:n}\in\mathbb{N}^{n}}\Big\{\prob{\rand C_{1:n}=c_{1:n}\mid \rand Y_{1:n}=y_{1:n}}\Big\}
\end{align}
The gray vertical lines in Figure \ref{fig:laplace_forward}(a) show the result of this estimation in the well-log example.

After estimating the \cp locations we may also estimate the corresponding segment heights by means of a Bayes or MAP estimator w.r.t. the $\heights_{ji}$ (see also Lemma \ref{lem:segheightsretro}).

\subsubsection{Sampling}
\label{chap:sampling}
Now, we show how to obtain i.i.d. samples from $(\rand C_{1:n}\mid \rand Y_{1:n}=y_{1:n})$. 
Later on in Section \ref{chap:simcredreg} we deploy these samples for a novel approach to analyze \cp models and \cp data.

\begin{algorithm}[ht]
	\begin{algorithmic}[1]
		%\scriptsize
		\footnotesize 
		\State{\textbf{Input:} A possibly sparse representation of $\tccc_{ji}$ for $0<i\leq n$ and $0\leq j\leq i$.}
		\State{\textbf{Output:} Sample $S$ }
		\State{$i=n$}		
		\While{$i>0$}
		\State{$s=0$}		
		\State{Sample an $u$ from the uniform distribution over $[0,1]$}
		\For{$j=i,\ldots,0$ such that $\tccc_{ji}$ exists}
		\State{$s=s+\tccc_{ji}$}\label{pseudocode:sampling1}
		\If{$s\geq u$}
		\State{$S\ll j$}
		\State{$i=j-1$}
		\State{break the for loop}		
		\EndIf
		\EndFor
		\EndWhile		
	\end{algorithmic}

	\caption{Sampling}
	\label{pseudo:sampling}
	%{Implementation of \greedy with a runtime complexity of $\mathcal{O}(n^2m\ln(m))$.}
\end{algorithm}

Sampling from $(\rand C_{1:n}\mid \rand Y_{1:n}=y_{1:n})$ is done in a stepwise manner by means of the $\tccc_{jn}$'s.
We start at $n$ and recursively obtain new \cp locations until we arrive at $0$.
Pseudocode \ref{pseudo:sampling} shows a possible implementation of the sampling procedure also dealing with a sparse representation of the $\tccc_{ji}$'s.
This algorithm exhibits an overall complexity of $\mcal{O}(n)$, because Step \ref{pseudocode:sampling1} can be executed $n+1$ times at most.

Given a \cp sample, we may further sample the segment heights conditioned on the whole data set and \cp locations according to the $\heights_{ji}$'s (see also Lemma \ref{lem:segheightsretro}).
This finally yields a joined sample of the \cp locations and segment heights conditioned on the whole data set, i.e. a sample of $(\rand C_{1:n}, \rand X_{1:n}\mid \rand Y_{1:n}=y_{1:n})$.

\subsubsection{Entropy}
\label{chap:entropy}

\cite{guedon2015segmentation} is concerned with the entropy of \cp locations.
Entropy can be used as a measure of dispersion.
We want to discuss its computation briefly.
The entropy of $(\rand C_{1:n}\mid \rand Y_{1:n}=y_{1:n})$ is defined as 
\begin{align*}
-\sum_{c_{1:n}\in\mathbb{N}^n}\prob{\rand C_{1:n}=c_{1:n}\mid \rand Y_{1:n}=y_{1:n}}\ln\big(\prob{\rand C_{1:n}=c_{1:n}\mid \rand Y_{1:n}=y_{1:n}}\big)
\end{align*}
whereby we agree that $0\cdot \ln(0)=0$.

Its computation is done by a dynamic programming algorithm.
Therefore let 
\begin{align*}
&\funcb(c_{1:i})=\begin{cases}
\prob{\rand C_{1:i}=c_{1:i}\mid \rand C_{i+1}=i+1, \rand Y_{1:i}=y_{1:i}}& 0< i< n\\
\prob{\rand C_{1:n}=c_{1:n}\mid \rand Y_{1:n}=y_{1:n}}& i=n
\end{cases}
\end{align*}
and for $0<i\leq n$ let $\ecc_{i}=\sum_{c_{1:i}\in\mathbb{N}^i}\funcb(c_{1:i})\ln\big(\funcb(c_{1:i})\big)$.
Ultimately, $-\ecc_n$ corresponds to the entropy of $(\rand C_{1:n}\mid \rand Y_{1:n}=y_{1:n})$.
The following lemma shows that $\ecc_n$ can be derived by computing all the $\ecc_i$'s one by one in a recursive fashion.

\begin{lemma}[Gu{\'e}don]
For $i=1,\ldots,n$ the following holds
\begin{align*}
\ecc_i=\sum_{j=0}^i \tccc_{ji}\big(\ecc_{j-1}+\ln(\tccc_{ji})\big)
\end{align*}
whereby we define $\ecc_{-1}=\ecc_0=0$.
\end{lemma}
\begin{proof}
\begin{align*}
&\ecc_i=\sum_{c_{1:i}\in\mathbb{N}^i}\funcb(c_{1:i})\ln\big(\funcb(c_{1:i})\big)=\sum_{j=0}^i\sum_{c_{1:i}\in\mathbb{N}^i\text{ with }c_{i}=j}\funcb(c_{1:i})\ln\big(\funcb(c_{1:i})\big)\\
&=\tccc_{0i}\ln(\tccc_{0i})+\tccc_{1i}\ln(\tccc_{1i})+\sum_{j=2}^i\tccc_{ji}\underbrace{\sum_{c_{1:j-1}\in\mathbb{N}^{j-1}}\funcb(c_{1:j-1})\ln\big(\funcb(c_{1:j-1})\big)}_{\ecc_{j-1}} +\tccc_{ji}\ln(\tccc_{ji})\\
&=\sum_{j=0}^i \tccc_{ji}\big(\ecc_{j-1}+\ln(\tccc_{ji})\big)
\end{align*}
\end{proof}

Since each $\tccc_{ji}$ is used only twice, the complexity of computing $\ecc_n$ is linear in the number of particles.

\subsection{Pointwise inference}
\label{chap:pointwise_spike}
In the following three sections, we utilize
\begin{align*}
\tqcc_i=\prob{\rand C_i=i\mid \rand Y_{1:n}=y_{1:n}}\quad\text{and}\quad(\rand X_i\mid \rand Y_{1:n}=y_{1:n})
\end{align*}
for $i=1,\ldots,n$.
These are pointwise quantities conditioned on the whole data set.
The first expresses the probability of seeing a \cp at $i$ and the latter represents the random segment height at $i$.
The consideration of these quantities can be referred to as smoothing (see, for example, \cite{rauch_striebel_tung} and \cite{sarkka2013bayesian}).

\cite{Rigaill2012} and \cite{aston2012implied} employ the $\tqcc_i$'s for model selection purposes.
However, they can also be used to pursue more sophisticated tasks.
We can, for example, use the standard deviation of $(\rand X_i\mid \rand Y_{1:n}=y_{1:n})$ to derive meaningful bands for the segment heights or display the trajectory of their expectations.
Moreover, they pave the way for a novel EM algorithm for \cp models (Section \ref{chap:em}).

\begin{lemma}
\label{lem:margcpprob}
For $0<i\leq n$ the following applies
\begin{align*}
&\tqcc_i=\sum_{\ell=i}^{n-1}\tccc_{i\ell}\tqcc_{\ell+1} + \tccc_{in}
\end{align*}
\begin{proof}
For $0<i<n$ we see that
\begin{align*}
&\prob{\rand C_i=i\mid \rand Y_{1:n}=y_{1:n}}=\sum_{\ell=i}^{n-1}	\prob{\rand C_\ell=i, \rand C_{\ell+1}={\ell+1}\mid \rand Y_{1:n}=y_{1:n}}+\tccc_{in}\\
&=\sum_{\ell=i}^{n-1}	\prob{\rand C_\ell=i\mid \rand C_{\ell+1}={\ell+1}, \rand Y_{1:n}=y_{1:n}}\tqcc_{\ell+1}+\tccc_{in}=\sum_{\ell=i}^{n-1}	\tccc_{i\ell}\tqcc_{\ell+1}+\tccc_{in}
\end{align*}
There is nothing to show for $\tqcc_n$.
\end{proof}
\end{lemma}

\begin{myremark}
\label{rem:probseg}
The quantity $\tccc_{i\ell}\tqcc_{\ell+1}$ represents the probability of seeing a segmentation with a segment that starts in $i$ and ends in $\ell$ conditioned on the whole set of observations.
\end{myremark}
\begin{myremark}
\label{rem:expcp}
$\sum_{i=1}^n\tqcc_i$ corresponds to the expected number of \cps given the data.
\end{myremark}
%\cite{Rigaill2012} shows a plot of those probabilities, which demonstrates that the \cp locations are usually highly dependent and only few combinations are likely.
%A novel tool to visualize combinations of more than two \cps is elaborated in Section \ref{chap:simcredreg}.

The computation of the $\tqcc_i$'s yields a dynamic programming algorithm that
starts with $\tqcc_n$ and proceeds with $\tqcc_{n-1}$ and so forth until $\tqcc_1$ is reached.
However, this algorithm requires a particular access to specific particles.
Whilst in the previous algorithms, we iterate through the existing particles $(j,i)$ w.r.t. $j$, now we need to iterate over $i$.
Thus, in order to benefit from pruning here, we need to take care of this fact, e.g. by storing a second matrix or using an appropriate sparse representation.
The computation of $\tqcc_i$ for $i=1,\ldots,n$ touches each particle only once and its complexity is thus linear in the number of particles.

\begin{lemma}
\label{lem:margsegmentheights}
For $0<i\leq n$ the following holds
\begin{align*}
&(\rand X_i\mid \rand Y_{1:n}=y_{1:n})\sim \sum_{j=0}^i\heights_{jn}\tccc_{jn}+\sum_{\ell=i}^{n-1}\heights_{j\ell}\tccc_{j\ell}\tqcc_{\ell+1}
\end{align*}
\end{lemma}
\begin{proof}
A marginalization over all segments that contain $i$ yields for $0<i<n$
\begin{align*}
\prob{\rand X_i\in\LargerCdot\mid \rand Y_{1:n}=y_{1:n}}
&=\sum_{j=0}^i\prob{\rand X_i\in\LargerCdot,\rand C_n=j\mid \rand Y_{1:n}=y_{1:n}}\\
&+\sum_{\ell=i}^{n-1}\prob{\rand X_i\in\LargerCdot, \rand C_\ell=j, \rand C_{\ell+1}=\ell+1\mid \rand Y_{1:n}=y_{1:n}}\\
&\overset{*}{=}\sum_{j=0}^i\heights_{jn}\tccc_{jn}+\sum_{\ell=i}^{n-1}\heights_{j\ell}\tccc_{j\ell}\tqcc_{\ell+1}
\end{align*}
In $*$ we have used that the segment heights stay constant within a segment.
\end{proof}

In the unpruned case, $\prob{\rand X_i\in\LargerCdot\mid \rand Y_{1:n}=y_{1:n}}$ is a mixture with $\mcal{O}(n^2)$ components.
Therefore, in contrast to the marginal \cp probabilities, it can pose a huge burden to store these quantities directly.
However, Lemma \ref{lem:margsegmentheights} indicates that $\prob{\rand X_i\in\LargerCdot\mid \rand Y_{1:n}=y_{1:n}}$ shares many components with $\prob{\rand X_{i+1}\in\LargerCdot\mid \rand Y_{1:n}=y_{1:n}}$, a fact that can be exploited nicely.

\begin{lemma}
\label{lem:itersegmentheight}	
For $0<i<n$ the following holds
\begin{align*}
&\prob{\rand X_i\in\LargerCdot\mid \rand Y_{1:n}=y_{1:n}}
=\prob{\rand X_{i+1}\in\LargerCdot\mid \rand Y_{1:n}=y_{1:n}}\\
&-\sum_{\ell=i+1}^{n-1}\heights_{i+1\ell}\tccc_{i+1\ell}\tqcc_{\ell+1}-\heights_{i+1n}\tccc_{i+1n}
+\sum_{\ell=0}^{i}\heights_{\ell i}\tccc_{\ell i}\tqcc_{i+1}
\end{align*}
\end{lemma}

The proof of Lemma \ref{lem:itersegmentheight} is straightforward.
The intuition behind it is that the sets of segments that contain either $i$ or $i+1$, only differ in the segments that start at $i+1$ and end at $i$.
This is because, these are the only segments which do not include $i$ and $i+1$ simultaneously.
These segments need to be either added or removed in order to transition from the set of segments that contain $i+1$ to those that contain $i$.
The initial timepoint to start with is $n$ since $\prob{\rand X_n\in\LargerCdot\mid \rand Y_{1:n}=y_{1:n}}=\sum_{j=0}^n \tccc_{jn}\heights_{jn}$ is already known.

In practice, we are less inclined to compute these probability distributions, we mostly utilize expectations w.r.t. to measurable functions $\funcf:\lat\rightarrow \mathbb{R}^u$, i.e.
\begin{align}
\label{eq:expsegheight}
\expec{\funcf(\rand X_i)\mid \rand Y_{1:n}=y_{1:n}}
\end{align}
Since the expectation w.r.t. a mixture distribution also decomposes into the sum of corresponding expectations, the basic requirement for computing (\ref{eq:expsegheight}) is the feasibility of 
\begin{align}
\label{eq:hoff1}
&\heights_{ji}^{\funcf}=\int \funcf(x)\heights_{ji}(dx)
\end{align}
for $0<i\leq n$ and $0\leq j\leq i$.
Together with Lemma \ref{lem:itersegmentheight}, this allows for an iteration scheme that touches each of the segments only twice, once when it is added and once when it is removed.

\begin{corollary}
\label{cor:exp}
For $0<i<n$ the following holds
\begin{align*}
&\expec{\funcf(\rand X_i)\mid \rand Y_{1:n}=y_{1:n}}\\
&=\expec{\funcf(\rand X_{i+1})\mid \rand Y_{1:n}=y_{1:n}}
-\sum_{\ell=i+1}^{n-1}\heights_{i+1\ell}^{\funcf}\tccc_{i+1\ell}\tqcc_{\ell+1}
-\heights_{i+1n}^{\funcf}\tccc_{i+1n}
+\sum_{\ell=0}^{i}\heights_{\ell i}^{\funcf}\tccc_{\ell i}\tqcc_{i+1}
\end{align*}
\end{corollary}

\begin{theorem}
\label{theo:expcomp}
Given a possibly sparse representation of $\tccc_{ji}$ for $0<i\leq n$ and $0\leq j\leq i$, the computation of $\tqcc_i$ for all $0<i\leq n$ amounts to a time complexity of $\mathcal{O}$(number of particles) and a space complexity of $\mathcal{O}(n)$.
Further assume that we have computed $\heights_{ji}^{\funcf}$ whenever particle $\tccc_{ji}$ exists.
The time and space complexity of computing $\expec{\funcf(\rand X_i)\mid \rand Y_{1:n}=y_{1:n}}$ for all $0<i\leq n$ at once reads $\mathcal{O}$(number of particles) and $\mathcal{O}(n)$, respectively.
\end{theorem}

This poses a huge improvement to the straightforward approach that applies Lemma \ref{lem:margsegmentheights} directly.
Under good conditions, with pruning in place, we are now able to compute arbitrary marginal expectations, and therewith variances, skewnesses and many others, for all $i=1,\ldots,n$ at once with linear complexity w.r.t. $n$.
In particular, this renders the EM algorithm of Section \ref{chap:em} feasible for big data purposes.

These inference formulas for segment height appear to be new within the \cp literature.
However, after developing these formulas, similar approaches dedicated to hidden semi-Markov models came to my attention.
Ultimately, a careful comparison with \cite{hsmm} revealed a close algorithmic relationship between \cp models and hidden semi-Markov models.
The main difference between these models is that the latter employs only discrete distributions for the segment height.

\subsubsection{Pointwise statements in the well-log example}
\label{chap:well-log-pointw}
\begin{figure}[ht]
	\includegraphics[width=\linewidth]{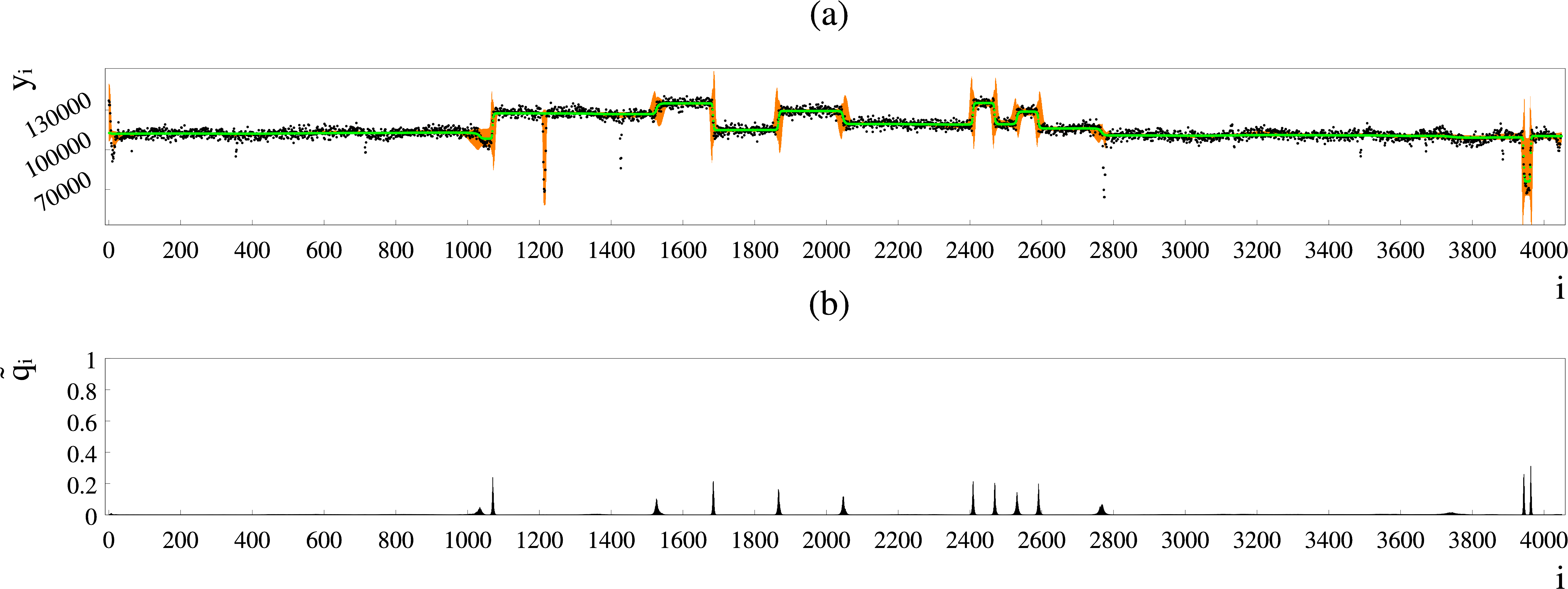}
	\caption[Well-log data with posterior expectations, variances, and \cp histogram]{(a) Well-log data and posterior expectation of the segment heights in green and a depiction of the standard deviation and skewness in orange. (b) Marginal \cp probabilities.
	}
	\label{fig:laplace_marginals}
\end{figure}
We are now interested in several posterior pointwise quantities with regard to our well-log example in Section \ref{chap:well_log_forward}.
To this end, we compute the first three moments of the segment heights.

For $m\in\mathbb{N}$ we define $\funcf_m(x)=x^m$.
According to Corollary \ref{cor:exp}, computing the $m$-th moments of $(\rand X_i\mid \rand Y_{1:n}=y_{1:n})$ for $i=1,\ldots,n$ comes down to computing 
$\heights_{ji}^{\funcf_m}$ whenever particle $(j,i)$ exists.
Since $\heights_{ji}^{\funcf_m}=\Zcc_{ji}^m\slash\Zcc_{ji}^0$ (see Equation (\ref{eq:laplaceexp})),
we achieve this by means of Corollary \ref{cor:intlaplace}.

Figure \ref{fig:laplace_marginals}(a) depicts the well-log data in black and in green it shows the expectations $\mcc_i=\expec{\rand X_i\mid \rand Y_{1:n}=y_{1:n}}$. 
We see that within a segment, the expectations do hardly vary even in the presence of outliers.
However, they move away very rapidly at \cp locations.

We take a simple and fairly arbitrary approach to visualize the standard deviations $\scc\dcc_i=\sqrt{\expec{\big(\rand X_i-\mcc_i\big)^2\mid \rand Y_{1:n}=y_{1:n}}}$ and skewnesses $\scc\kcc_i=\expec{\big(\rand X_i-\mcc_i\big)^3\slash\scc\dcc_i^3\mid \rand  Y_{1:n}=y_{1:n}}$ in a joint fashion.
For each timepoint $i=1,\ldots,n$, (a) depicts the interval
\begin{align*}
\big[\mcc_i-2\cdot(\scc\dcc_i-\scc\kcc_i\ind{\scc\kcc_i< 0}),\ \mcc_i+2\cdot(\scc\dcc_i+\scc\kcc_i\ind{\scc\kcc_i\geq 0})\big]
\end{align*} 
in orange.

This shows that within a segment, the variations of the heights are fairly small.
In turn, at \cp locations their variations become much larger due to their ambiguity towards two different segment heights.

Whilst the mean doesn't follow the outliers, for the example at around 1200, the model expresses its tendencies through its skewness.
Skewness also arises around \cp locations.
This is partially a consequence of the mixed form of the pointwise segmental distributions.
However, as we have seen in Section \ref{chap:laplacian}, thanks to the Laplace distribution, single components of the mixture may also be skewed.

Thus, in the middle of longer segments, where only very few components dominate the posterior height distribution, the Laplace distribution is capable of maintaining skewness.
In contrast, the normal distribution wouldn't allow this.

Figure \ref{fig:laplace_marginals}(b) shows the marginal \cp probabilities $\tqcc_i$ for each timepoint $i=1,\ldots,n$.
Within the Bayesian \cp literature, these quantities are mainly used to justify \cp models or to perform model selection in a manual fashion.
By this means, we can obtain a rough impression of the spots where the model is inclined to place \cps.
According to Remark \ref{rem:expcp}, the expected number of \cps is 17.8.

Due to the location uncertainty, single variations in the data can usually be explained through \cps w.r.t. more than one location.
Within such a set of locations we usually see only one \cp at a time.
As a consequence, the probability of seeing \cps there is shared among the corresponding $\tqcc_i$'s.
For this reason, the heights in Figure \ref{fig:laplace_marginals}(b) appear to be quite disorganized and not very informative in terms of the sensitivity of the model to place \cps.

Consider, for example, the dip at around 3800.
The marginal \cp probabilities appear suspiciously low and broad there.
In fact, the (empirical) probability of seeing at least one \cp between 3600 and 3900, inferred from $10^7$ \cp samples (see Section \ref{chap:sampling}), amounts to 0.76, which is a highly unexpected result. 
We address this issue by means of a novel visualization scheme in Section \ref{chap:simcredreg}.

Instead of the moments, we could also compute median and interquartile range.
Even though this appears to be more suitable for the Laplace distribution, these quantities cannot be represented  in terms of simple expectations and their implementation is in particular not straightforward for mixture distributions.

\subsection{Parameter estimation via likelihood and the EM algorithm}
\label{chap:param}

In this section, we discuss means to estimate the parameters of a \cp model.
Thereby, we focus on the marginal loglikelihood function (Remark \ref{rem:marglikelihood}).
Examples for parameters are $\mu,\sigma,\tau,q$ and $r$ in our well-log example (Section \ref{chap:well_log_forward}).

Assume that $\theta$ represents the underlying parameter (or parameters) of interest and $\Theta$ a real coordinate  space that represents all possible values of $\theta$.
For the purpose of choosing a value for $\theta$, we have to take into account that the latent \cp process remains unknown.
This is important because the model depends on the parameters and the segmentation jointly.

A primal approach could be to estimate the set of parameters, \cp locations and segment heights simultaneously.
However, this is very susceptible towards overfitting and therewith paves the way for degenerated models.

A workaround is to fit the parameters in accordance to a predetermined segmentation or if appropriate, a single representative segment.
This takes advantage of prior knowledge, but doesn't guarantee that the model will later be sensitive towards this choice of segmentation or segment.
More importantly, it ignores the uncertainty behind the choice of the segmentation.

Bayesian techniques are capable of incorporating uncertainty into the estimation.
A common approach is to model $\theta$ through a random variable $\rand \theta$ and obtain a Bayes estimator, e.g. the posterior mean, $\expec{\rand \theta\mid \rand Y_{1:n}=y_{1:n}}$.

Such an estimator absorbs all possible segmentations through a marginalization over $\rand X_{1:n}$ and $\rand C_{1:n}$.
This nicely takes care of the segmentation uncertainty.
Furthermore, computing the expectation over $\rand \theta$ addresses the uncertainty behind the parameters.

In this course, however, we focus on the segmentation uncertainty.
Thus, we leave the parameters non-random and consider the likelihood marginalized over all possible segmentations
\begin{align*}
L_\theta=\prob{\rand Y_{1:n}\in dy_{1:n}\semic \theta}\densslash \psi(dy_{1:n})
\end{align*}
The sign $\semic$ introduces the list of non-random parameters that are to be estimated.
In Remark \ref{rem:marglikelihood}, we have already explained how to compute $L_\theta$.

Under some conditions, finding a maximum of $L_\theta$ yields sensible estimates for $\theta$. 
One of these conditions is that the distribution behind the data and the distribution employed in the \cp model match sufficiently.

This can be an awkwardly strong obligation, especially, if it implies that the model is build from mixture distributions, e.g. where outliers are modeled separately.
In fact, for \cp inference, a tendency to organize the datapoints correctly into segments, is mostly sufficient.
Thus, \cp models just need to prove a  tool not to reflect reality, a perspective that lies more in the heart of a Bayesian approach than a frequentist one.
In this context, we may additionally resort to prior knowledge to steer the parameter estimation, by means of the marginal likelihood, into the right direction.
Section \ref{chap:well-log-para} demonstrates this for the well-log example.

A simple grid search poses a first heuristic to find a maximizer of $L_\theta$.
However, if $\theta$ is of higher dimension, the required grid can easily become too huge to be scanned through.
Moreover, the grid may be to rough at important sights, which promotes inaccurate results.

To this end, we could demand for more sophisticated algorithms like gradient descent \citep{cauchy1847methode}, but we are usually unable to compute derivatives of $L_\theta$.
This is because $L_\theta$ is the product of an intricate marginalization that entangles $\theta$ and the observations in manifold ways.

Fortunately, the EM algorithm is able to overcome this issue by exploiting the fact that our \cp models factorize into much smaller and simpler distributions.

In the following, we will first discuss the EM algorithm for \cp models in general (Section \ref{chap:em}).
Afterwards, we discuss the EM algorithm for segment length for the geometric as well as the negative binomial distribution (Section \ref{chap:geomem}).
Section \ref{chap:well-log-para} elaborates a parameter estimation approach for the well-log example.
Finally, we investigate our EM algorithm for segment height and observation and apply it to the Laplacian change in median model followed by a brief simulation study (Section \ref{chap:emheiobs} and \ref{chap:emsigmalap}).

\subsubsection{An EM algorithm for \cp models}
\label{chap:em}
The expectation-maximization (EM) algorithm \citep{emalgorithm} is a popular way to gain estimators for the parameters of a statistical model that involves latent or unobserved variables.
To this end, it provides local maximizers of the marginal likelihood.
The great advantage of the EM algorithm is that it exploits the factorized structure many statistical models including \cp models are subject to.

The EM algorithm works on the basis of densities. 
Therefore, assume that $(\lat, \latsigma)$ is endowed with a measure $\xi$ such that $\jprior(dx)\densslash \xi(dx)$ exists.

For $c_{1:n}\in\mathbb{N}^n$ and $x_{1:n}\in\lat^n$, let $\cloglik_\theta(c_{1:n}, x_{1:n})$ be the complete loglikelihood of the \cp model w.r.t. $c_{1:n}$ and $x_{1:n}$, i.e.
\begin{align}
\label{eq:emdecompose}
\cloglik_\theta(c_{1:n}, x_{1:n})=&\sum_{i=1}^n\ind{c_i=i}\funcg(x_i\semic \theta)+\funcf(y_i\mid x_i\semic \theta)\\
&+\ln(\prob{\rand C_1=c_1\semic \theta})+\sum_{i=2}^n\ln(\prob{\rand C_i=c_i\mid \rand C_{i-1}=c_{i-1}\semic \theta})
\nonumber
\end{align}
whereby $\funcg(x\semic \theta)=\ln(\jprior(dx\semic \theta)\densslash \xi(dx))$ and $\funcf(y_i\mid x\semic \theta)=\ln(\dens(\rand Y_i=y_i\mid \rand X_i=x\semic \theta))$.
For simplicity, we assume that $\dens(\rand Y_i=y\mid \rand X_i=x\semic \theta)$ is of the same form for all $i=1,\ldots,n$ and agree that $\ln(0)=-\infty$ and $0\cdot \ln(0)=0$.

\begin{algo}[EM algorithm]
\label{alg:em}
Choose a $\theta_0\in\Theta$ to start with.
In step $\ell=1,2,\ldots$ derive $\theta_\ell$ as an element of 
\begin{align}
\label{eq:em}
\argmax_{\theta\in\Theta}\Big\{\expec{\cloglik_\theta(\rand C_{1:n}, \rand X_{1:n})\mid \rand Y_{1:n}=y_{1:n}\semic \theta_{\ell-1}}\Big\}
\end{align}
Keep iterating over $\ell$ until $\theta_\ell$ converges.
\end{algo}

This convergence is justified theoretically if $L_\theta$ is bounded (see, for example, \cite{bishop2014}).
However, there is no guarantee that a global maximum of $L_\theta$ will be reached.
Instead, we can only guarantee for a local maximum.

In practice, convergence is reached if the distances of successive values of $\theta_\ell$ fall below a self defined threshold.
If pruning is in place, instead of converging, the values of $\theta_\ell$ may start to circulate.
This needs to be taken care of in order to avoid unnecessary u-turns.
\cite{hoffman2014no} describe such an approach in another context.

%If we apply pruning, due to inaccuracies, the values of $\theta_\ell$ may tend to circle and not to converge. 
%In order to obtain a safe stopping rule for the EM algorithm, we have to take this into account.
%This can be achieved in the fashion of the No-U-Turn Sampler of \cite{hoffman2014no}.

Since the expectation of a sum decomposes into a sum of expectations, we only need to compute the expectations w.r.t. to the single summands in (\ref{eq:emdecompose}).
Furthermore, the desired parameters do usually solely affect either the distribution of time, observation or segment height.
Thus, we consider these cases separately in order to focus on their particular characteristics.

The EM algorithm is an abstract algorithm, where each statistical model requires its own implementation.
Several instances for particular models have been developed, 
most notably the Baum-Welch algorithm for (discrete) hidden Markov models \citep{rabiner1986introduction}.
Lately, interesting progress has been made in connection with continuous time hidden Markov models \citep{liu2015efficient, roberts2008algorithm}. 
This concerns the estimation of the entries of the underlying generator matrix, given the observations.

Intrinsic to this estimation is the exploitation of endpoint-conditioned distributions (see also \cite{hobolth2009simulation}).
This is a decomposition of the posterior distribution, based on the discrete grid of observational timepoints.
Even though, in some cases, this is applicable to \cp models as well, in the following, we elaborate algorithms that require a more sophisticated approach.
That is, a decomposition based on segmentations.

Until today, there are some particular instances of the EM algorithm for \cp models available.
It has been shown that an EM algorithm can be developed to detect \cp locations \citep{keshavarz2014expectation}.
\cite{bansal2008application}, in turn, estimate the prior probabilities $\prob{\rand C_i=i}$ in the single \cp case by means of an EM algorithm.
Furthermore, stochastic EM algorithms for parameter estimation in an online fashion have been developed \citep{yildirim2013online}.

We demonstrate implementations of our EM algorithm, which is concerned with parameter estimation, for a certain fundamental class of \cp models and thereby provide a basic guide.
At the same time, we discuss its limits in regard to feasibility.

We assume that we are in the middle of the iteration process of Algorithm \ref{alg:em}.
The current value of $\theta$ is $\tho$ and our task is to conduct the next EM step
to provide a new value, say $\thn$.
In order to keep the notation uncluttered, we will mostly omit and not mention $\tho$.
However, it is important to keep in mind that the employed posterior distributions are build on this particular assignment of $\theta$.

\subsubsection{An EM algorithm for segment length}
\label{chap:geomem}
At first, we elaborate the case where the sojourn time is geometrically distributed.
Through the iteration procedure of the EM algorithm, we want to compute the success probability $q$ as part of $\theta$, whereby $q$ is not employed in any other distribution.
We assume that $\qold$ is the value of $q$ within $\tho$ and $\qnew$ the value of $q$ within $\thn$.
Subsequently, we explain how to obtain estimates for the success probability of the negative binomial distribution when the number of successes is given.

In the geometric case we have
\begin{align}
\nonumber
&\ln(\prob{\rand C_i=c_i\mid \rand C_{i-1}=c_{i-1}})&&=\ind{c_i=i}\ln(q)+\ind{c_i<i}\ln(1-q)\\
\label{eq:loggeomprob}
&\ln(\prob{\rand C_1=c_1})&&=\ind{c_1=1}\ln(q)+\ind{c_1=0}\ln(1-q)
\end{align}
for $i=2,\ldots,n$.

\begin{lemma}
\label{lem:geomem}
We may set 
\begin{align*}
&\qnew=\frac{\tqcc_1+\ldots+\tqcc_n}{n}
\end{align*}
\end{lemma}

\begin{proof}
We just need to find a global maximum of
\begin{align*}
\sum_{i=1}^n\tqcc_i\ln(q)+(1-\tqcc_i)\ln(1-q)
\end{align*}
w.r.t. to $q$.
\end{proof}
\noindent
According to Remark \ref{rem:expcp}, in each step, this EM algorithm derives the expected number of \cps and divides it by the number of timepoints.

Similar to \cite{bansal2008application}, we may assume that the $\qcc_{ji}$ are autonomous probabilities and independent of $j$, say $\qcc^i=\qcc_{ji}$.
In this case, each step of the EM algorithm simply yields $\qcc^{i}_{new}=1-\tqcc_{i}$ for $i=1,\ldots,n$.

%
%A generalization of the geometric distribution that allows for a limitation of the minimal length of a segment is $\ind{\ell<T}+\ind{\ell\geq T}(1-q)^{\ell-T}q$.
%$\ell$ represents the segment length and $T\in\mathbb{N}$ its minimum length.
%$q$ serves as the probability of success for this generalized geometric distribution.
%In this case the $\qcc_{ji}$'s from Equation (\ref{eq:ciprobs}) read
%\begin{align*}
%\qcc_{ji}=\begin{cases}
%(1-q)& j=0 \text{ or } i-j\geq T\\
%1& \text{ else}
%\end{cases}
%\end{align*}
%for $0\leq j<i\leq n$.
%The following corollary shows why this is a handy distribution.
%
%
%
%\begin{corollary}
%For the generalized geometric distribution we may apply the result from Lemma \ref{lem:geomem} directly. 
%In turn, the choice of $\gamma$ in Lemma \ref{lem:choiceofgamma} becomes
%\begin{align*}
%\gamma=\ln\big((1-q)\slash q\big)+T\cdot \ln(1-q)
%\end{align*}
%whereby in Equation (\ref{eq:funcgeomem}) we minimize over all \cp configurations with segment lengths of at least $T$ time points.
%\end{corollary}
%

Now, we discuss the EM algorithm for the success probability, $q$, of the negative binomial distribution with a predetermined $r$.
We employ $\qnew$ and $\qold$ as before.

This time, the expectation in Equation (\ref{eq:em}) will be composed based on all possible segments from $j$ to $i$ with $0<i\leq n$ and $0\leq j\leq i$.
In this case, the negative binomial distribution contributes as $r\ln(q)+(i-j)\ln(1-q)$, whereby we have skipped the binomial coefficient since it does not rely on $q$.

In the following, we distinguish three kinds of segments: introductory segments with $0=j<i\leq n$, intermediate segment with $0<j\leq i<n$ and final segments with $0<j\leq i=n$.

Each segment accounts for a single summand and all summands together constitute the expectation in Equation (\ref{eq:em}).
The intermediate segments are the easiest to deal with.
\begin{lemma}
	The summand for the intermediate segment from $j$ to $i$ with $0<j\leq i< n$ can be represented as
	\begin{align*}
	\tccc_{j i}\tqcc_{i+1}\big(r\ln(q)+(i-j)\ln(1-q)\big)
	\end{align*}	
\end{lemma}
\begin{proof}
	For the segment from $j$ to $i$ we get
	\begin{align*}
	\prob{\rand C_i=j, \rand C_{i+1}=i+1\mid \rand Y_{i:n}=y_{1:n}}\ln\big(q^r(1-q)^{i-j}\big)
	=\tccc_{j i}\tqcc_{i+1}\ln\big(q^r(1-q)^{i-j}\big)
	\end{align*}
	whereby we have used Remark \ref{rem:probseg}.
\end{proof}

We assume that the data observation was stopped in the middle of the data generating process that goes on forever.
In order to facilitate the computation of the summands w.r.t. the final segments, we incorporate them as a whole and not only until timepoint $n$.
Let $\rand S$ be the length of the final segment, i.e. the segment that contains timepoint $n$.

\begin{lemma}
We conclude that for $0< j\leq n$ and $n-j\leq i$
\begin{align*}
\prob{\rand S=i, \rand C_n=j\mid \rand Y_{1:n}=y_{1:n}}=\frac{\tccc_{j n}\cdot \text{NB}(i\semic \qold,r)}{\sum_{\ell=n-j}^\infty\text{NB}(\ell\semic \qold,r)}
\end{align*}
\end{lemma}
\begin{proof}
This follows from $(\rand S\indep \rand C_{1:n-1}, \rand Y_{1:n}, \rand X_{1:n}\mid \rand C_n)$.
\end{proof}

\begin{corollary}
We may use the following summand for the final segment that starts at $0<j\leq n$
\begin{align*}
&\sum_{i=n-j}^\infty \prob{\rand S=i, \rand C_n=j\mid  \rand Y_{1:n}=y_{1:n}}\big(r\ln(q)+i\ln(1-q)\big)\\
&=\tccc_{j n}r\ln(q)+\tccc_{j n}\Bigg(\sum_{\ell=n-j}^\infty\ell\cdot\text{NB}(\ell\semic \qold,r)\Bigg)\Bigg\slash \Bigg(\sum_{\ell=n-j}^\infty\text{NB}(\ell\semic \qold,r)\Bigg)\ln(1-q)
\end{align*}
\end{corollary}

We further assume that the data was observed in the middle of the process which started infinitely long ago.
Let now $\rand S$ represent the random length from 0 to the end of the first segment.
We get that $\prob{\rand S=i}= \frac{q}{r(1-q)}\sum_{\ell=i+1}^\infty\text{NB}(\ell\semic q,r)$, whereby
$\rand S=0$ corresponds to $\rand C_1=1$ and $\rand S=i$ with $i>0$ corresponds to $\rand C_{i}=0$ and $\rand C_{i+1}=i+1$.
This distribution incorporates $q$ in a fairly difficult manner rendering an exact computation of the desired summand for the EM algorithm infeasible.

While the negative binomial distribution may heavily concentrate its mass to specific lengths, the distribution of $\rand S$ allows a much larger freedom.
Intuitively speaking, it puts the most of its mass to all the lengths that are not longer than those lengths which are likely under the corresponding negative binomial distribution.
Therefore, it is inappropriate to simply consider $1$ as the first \cp.
Instead, we will use minor modifications of the length distribution to address these difficulties.

The distribution of $\rand S$ has its mode at 0 and is monotonously decreasing.
Thus, a geometric distribution appears to be a reasonable replacement.
Its success probability, say $q'$, can be chosen as $q\slash r(1-q)$.
We come to this conclusion by looking at
\begin{align*}
\prob{\rand S=i}=\frac{q}{r(1-q)}\Bigg(1-\sum_{\ell=0}^{i}\binom{r-1+\ell}{\ell}q^r(1-q)^{\ell}\Bigg)
\end{align*}
Considering $i=0$, this means that $\big(q(1-q^r)\big)\slash\big(r(1-q)\big)$ appears to be the right candidate for $q'$.
However, in order to keep things simple, we neglect $q^r$ and set $q'=q\slash\big(r(1-q)\big)$ instead.

Please note that this approach precludes values of $q$ that are larger than $r\slash(r+1)$, because in this case, $q'$ doesn't represent a probability anymore.
Nevertheless, since for $q=r\slash(r+1)$ we already get a nonsensical expected segment length of 1 under the negative binomial distribution, this restriction is insignificant.

\begin{corollary}
By means of the above geometric approximation, the summand for the segment from $0$ to $i$ with $0< i\leq n$ can be represented as
\begin{align*}
\tccc_{0 i}\tqcc_{i+1}\big(\ln(q)-\ln(1-q)+i\ln\big(1-q\slash\big(r(1-q)\big)\big)\big)
\end{align*}
\end{corollary}

Having computed the summands for all segments, the expectation in (\ref{eq:em}), which is subject to maximization, finally equates to
\begin{align}
\label{eq:emnegbinexp}
c_1\ln(q)+c_2\ln(1-q)+c_3\ln\big(1-q\slash\big(r(1-q)\big)\big)
\end{align}
with constants $c_1,c_2$ and $c_3$.
\begin{lemma}
Let
\begin{align*}
z_{1, 2}=\frac{r(c_1+c_2)+c_1(r+1)+c_3\pm\sqrt{\big(r(c_1+c_2)+c_1(r+1)+c_3\big)^2+(r+1)(c_1+c_2)c_1r}}{4(r+1)(c_1+c_2)}
\end{align*}
A maximizer for $q$ in Equation (\ref{eq:emnegbinexp}) can be chosen uniquely from $\{z_1,z_2\}$.
\end{lemma}
\begin{proof}
Differentiating (\ref{eq:emnegbinexp}) and setting to 0 yields
\begin{align}
\nonumber
&c_1\slash q-c_2\slash(1-q)-\frac{c_3}{r(1-q)^2-q(1-q)}=0\\
%&\Leftrightarrow c_1(r(1-q)^2-q(1-q))-c_2(r(1-q)-q)q-c_3q=0\\ 
%&\Leftrightarrow c_1rq^2+c_1r-2c_1rq-c_1q+c_1q^2-c_2rq+c_2rq^2+c_2q^2-c_3q=0\\ 
%&\Leftrightarrow q^2(c_1r+c_1+c_2r+c_2)+q(-2c_1r-c_1-c_2r-c_3)+c_1r=0\\ 
%&\Leftrightarrow q^2(r+1)(c_1+c_2)-q\big(c_1(2r+1)+c_2r+c_3\big)+c_1r=0\\
&\Leftrightarrow q^2(r+1)(c_1+c_2)-q\big(r(c_1+c_2)+c_1(r+1)+c_3\big)+c_1r=0
\label{eq:quadeq1}
\end{align}
whereby $q\not\in\{0,1, r\slash (r+1)\}$.
We can easily verify that $c_1>0$ and $c_3\geq 0$.
Since $0<\qcc_{01}<1$, we even get $c_3>0$.
Substituting $q=0$ into the l.h.s. of Equation (\ref{eq:quadeq1}) yields $c_1r$, which is positive, and
for $q=r\slash(r+1)$ we get $-c_3$, which is negative.
Therefore, the l.h.s. of Equation (\ref{eq:quadeq1}) exhibits one single root between $0$ and $r\slash(r+1)$.
Finally, $q\leq r\slash (r+1)$ implies that this is the only solution.  
\end{proof}

The EM algorithm for the negative binomial distribution essentially parses through the existing $\tccc_{ji}$'s and touches each of them exactly once.
Thereby, it employs the marginal \cp probabilities $\tqcc_i$ for $i=1,\ldots,n$.
Assume that we have computed $\tccc_{ji}$ for $0\leq j\leq i$ and $0< i\leq n$.
By virtue of Theorem \ref{theo:expcomp}, we may state that the time and space complexity of each EM step for $q$ in the geometric as well as in the negative binomial case is $\mathcal{O}$(number of particles).

\subsubsection{Parameter estimation in the well-log example}
\label{chap:well-log-para}

\begin{figure}[ht]
	\includegraphics[width=\linewidth]{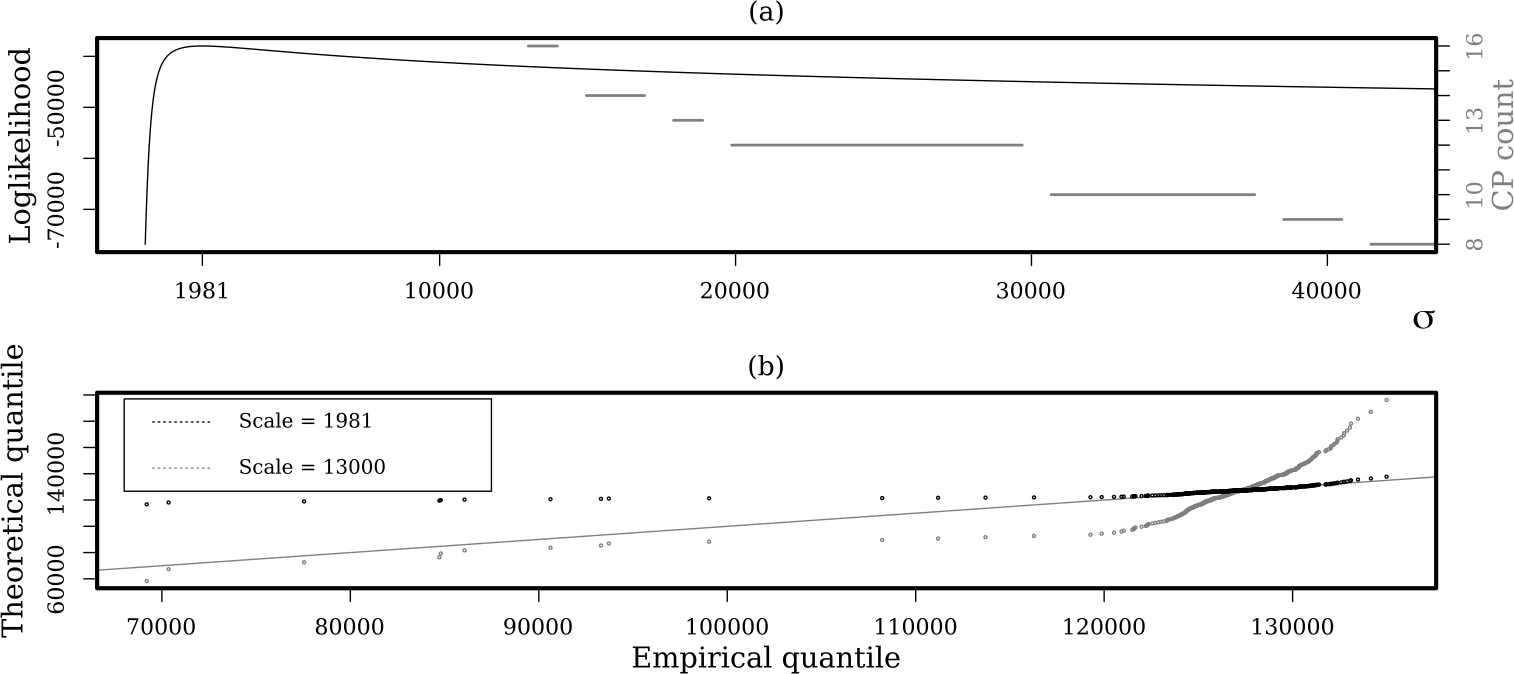}
	\caption[Loglikelihood with MAP and Q-Q plot]{(a) Illustration of the loglikelihood (left axis) and number of \cps inferred by the MAP (right axis) as a function of $\sigma$. (b) Q-Q plot of the empirical distribution of the datapoints from 1100 to 1500 against two Laplace distributions w.r.t. two different scale parameters. The straight line represents the diagonal.		
	}
	\label{fig:log-likelihood_laplace_sigma}
\end{figure}

In this section, we are concerned with determining the parameters $\sigma$, $q$ and $r$ in the well-log example of Section \ref{chap:well_log_forward}.
$\mu$ and $\tau$ remain fixed here.
Estimating the parameters of a statistical model is commonly referred to as training or learning.
Once trained, the finalized model can be employed straight away whenever new data is available.
To this end, we elaborate a practically adapted approach that employs the EM algorithm for segment length (Section \ref{chap:geomem}) and the MAP (Section \ref{chap:condlikandmap}).

We start with discussing $\sigma$.
Figure \ref{fig:log-likelihood_laplace_sigma}(a) depicts the loglikelihood, $\ln(L_\sigma)$, of our \cp model and the number of \cps inferred by the MAP, both w.r.t. $\sigma$ and with fixed $q=0.01430724$ and $r=3$.

$L_\sigma$ exhibits a single maximum at 1981, which yields an unreasonable \cp count of far more than 16.
Thus, an estimator for $\sigma$, solely based on maximizing the likelihood, doesn't provide reasonable estimates here.

The seemingly right-skewed shape of $\ln(L_\sigma)$ suggests that the posterior mean w.r.t. $\sigma$ may provide more suitable estimates.
For this sake, we put a prior on $\sigma$, which is the uniform distribution from 50 to 100000.
Furthermore, we use a piecewise linear approximation of $\ln(L_\sigma)$ to compute $\expec{\rand \sigma\mid\rand Y_{1:n}=y_{1:n}}$ in the fashion of Corollary \ref{cor:intlaplace}. 
Surprisingly, this yields a similarly useless value of approximately 1984.
This is because, those values of $\sigma$ near the peak dominate the whole calculation.

Another naturally arising estimation procedure is to infer $\sigma$ by fitting a manually chosen representative segment.
This has the advantage that it can be performed independently of $q$.

We take the segment from timepoint 1100 to 1500.
Figure \ref{fig:log-likelihood_laplace_sigma}(b) depicts two Q-Q plots and the diagonal.
One Q-Q plot shows the empirical distribution against the Laplace distribution with scale  equal to 1981 and location equal to the segment's median.
The other Q-Q plot depicts the same distributions but this time with scale equal to 13000.

The fact that 1981 maximizes $L_\sigma$ is reflected by the good fit of the intermediate quantile pairs to the diagonal.
However, the lower quantile pairs are situated far above the diagonal, which indicates that the corresponding theoretical quantiles are much larger than the empirical ones.
Thus, the large deviations present in the data are very unlikely for this particular Laplace distribution, which provokes many undesirable \cps.
In turn, $\sigma=13000$ appears to be a more suitable choice, since the corresponding lower quantile pairs stay below the diagonal.

What we can learn thereof, is that the data within the segments does not stem from a Laplace distribution.
Apparently, the data is strongly asymmetric, but more importantly, the portion of large deviations is simply too small. 
A likelihood based estimation does, however, rely on the assumption that the empirical and theoretical distributions match sufficiently.
As a result, maximizing $L_\sigma$ yields a far too small value for $\sigma$.

Having set $\sigma$ to 13000, we would now choose $q$.
$0.0088$ yields an average length of 339 for the negative binomial distribution with $r=3$ and thus around twelve \cps on average within the 4049 datapoints.

Therewith, the MAP estimator does not detect any \cps within the considered segment, but 2 extra \cps, one at 1034 and one at 3744.
The latter is highly undesirable, since it belongs to a perturbation in the data that takes the form of a dip.

However, the same reasoning, based on Q-Q plots, conducted on the segment that contains 3744 doesn't lead to a larger $\sigma$.
Thus, in order to get rid of the additional \cps, we either have to increase $\sigma$ or decrease $q$ based on a different approach.

Since the additional segments concerning the \cp at 3744 have reasonable lengths, we shouldn't address this by decreasing $q$.
Furthermore, the probability of segments with a length of at most 50 is only 1\% under the negative binomial distribution with $r=3$ and $q=0.0088$.
Such segments are readily present in the data suggesting that we should not only increase $\sigma$, but also $q$.

Hence, our final approach is to estimate $q$ and $\sigma$ in a simultaneous fashion.
Thereby, we resort to the likelihood to obtain appropriate values for $q$ and the MAP to adjust $\sigma$.
The MAP provides a most likely \cp configuration and is therefore well-suited for controlling the tendencies of the posterior distributions.
Fortunately, those MAP estimates having twelve \cps seem to match sufficiently with our desired locations depicted in Figure \ref{fig:laplace_forward}(a).

It is reasonable to assume that our well-log data represents an average example for this application.
Therefore, we aim at selecting a set of parameters that lies centrally within a region of suitable values for $\sigma$ and $q$.
This induces a sense of robustness, which also supports the reusability of the estimated parameters in case new data is observed.

Additionally, we need to take numerical aspects into account.
The roles of $q$ and $\sigma$ are mutually intertwined.
We can, for example, obtain longer segments by either decreasing $q$ or increasing $\sigma$.
Hence, smaller values of $q$ might be combined with smaller values of $\sigma$ and vice versa.
This collinearity paves the way to nonsensical, extreme parameter choices, which may seriously impair the accuracy during the calculations.
Aiming at high likelihood values wherever possible mitigates the risk concerning this issue.

To this end, we employ the function $\sigma\mapsto \hat \qcc$ that returns the EM estimate for $q$ given $\sigma$.
By this means, we choose $\sigma$ such that the MAP solidly yields twelve \cps given that $q$ is set to $\hat \qcc$ and $r$ is set to 3.

\begin{figure}[ht]
	\includegraphics[width=\linewidth]{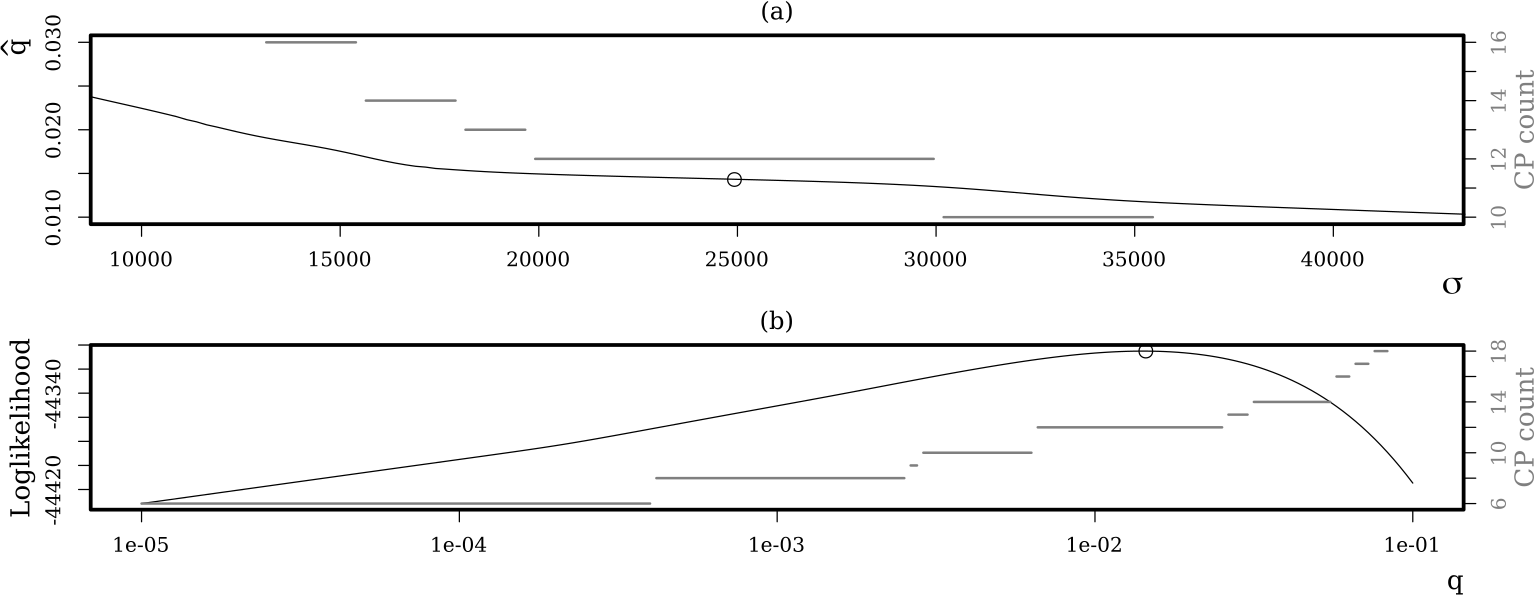}
	\caption[EM estimate for $q$ with MAP and loglikelihood w.r.t. $q$]{(a) Illustration of $\hat \qcc$ (left axis) and the number of \cps inferred by the MAP (right axis) w.r.t. $\sigma$.
		The circle marks the chosen set of parameters for $\sigma$ and $q$.
		(b) Illustration of the loglikelihood (left axis) and the number of \cps inferred by the MAP (right axis) w.r.t. $q$ and $\sigma=25000$.
		The circle marks the global maximum of the likelihood.
	}
	\label{fig:em-log-likelihood_laplace}
\end{figure}

Figure \ref{fig:em-log-likelihood_laplace}(a) demonstrates this.
It shows the graph of $\hat \qcc$ w.r.t. $\sigma$ and the corresponding \cp count inferred by the MAP.
$\sigma=25000$ is roughly centered within those values that yield twelve \cps.
The corresponding value of $\hat \qcc$ reads $0.01430724$.

Figure \ref{fig:em-log-likelihood_laplace}(b) shows the loglikelihood w.r.t. $q$ while having set $\sigma$ to $25000$, and the corresponding \cp count inferred by the MAP.
We see that the value $0.01430724$, indicated by the circle, is also roughly centered within those values that yield twelve \cps.
Moreover, the loglikelihood exhibits only one single global maximum being correctly inferred by the EM algorithm.

Therewith, we have considered two orthogonal directions of the $(\sigma, q)$-plane.
Although unlikely, this would be misleading if the MAP estimates spread very irregularly in the neighborhood of our selected set of parameter values.

Figure \ref{fig:cp_art}, however, suggests that this is not the case.
It depicts an excerpt of the $(\sigma, q)$-plane and colored tiles rendered in accordance with the \cp counts inferred by the MAP.
The upper and lower boundaries of the tiles w.r.t. $\sigma$ were computed by a simple binary search heuristic.
The cross points to $\sigma=25000$ and $q=0.01430724$.

As we see, our set of parameters is situated centrally within a comparably thick and regular region.
The unsteady or even incoherent nature of other regions indicates that there are dissimilar \cp configurations with equal \cp counts.
The aforementioned collinearity between $\sigma$ and $q$ is reflected through the beveled appearance of the regions.

Finally, we discuss the choice of $r$.
With increasing $r$, the mode of the length distribution shifts away from zero and thereby potentially increases the minimal size of the segments.
This mitigates the influence of outliers and thus allows for a slight decrease of the segmental scale.
Therewith, the sensitivity of the model is improved without impairing its specificity towards outliers.

However, we cannot increase $r$ arbitrarily in the presence of comparably small segments.
We found that $r$ should not be larger than 3.
This is because for $r>3$ the negative binomial distribution, when fitted to the segment lengths, turns out to be too sharply peaked.
As a result, the sensitivity towards unreasonable \cps soars.
We found that for $r=5,6$ and $q=\hat \qcc$, we can't even obtain a $\sigma$ that yields a MAP with twelve \cps.

\subsubsection{An EM algorithm for segment height and observation}
\label{chap:emheiobs}
Now, we elaborate the EM algorithm for segment height and observation by means of three exponential family alike families of distributions.
They cover a good range of practical cases and convey the basic idea behind this nicely.
We will further contemplate complexity and feasibility concerns.

We assume that $\funcf$ and $\funcg$ solely rely on their very own parameters $\theta^g$ and $\theta^f$.
$\theta^g$ and $\theta^f$ are defined on sets $\Theta^g$ and $\Theta^f$, respectively, and are part of $\theta$, but not linked to each other or any parameter that concerns the distribution of segment length.

Thus, we maximize the expectations of $\sum_{i=1}^n\funcf(y_i\mid \rand X_i\semic \theta^f)$ and $\sum_{i=1}^n\ind{\rand C_i=i}\funcg(\rand X_i\semic \theta^g)$ separately (see Equation (\ref{eq:emdecompose}) and Algorithm \ref{alg:em}).
We use $\tho^f$ and $\tho^g$ as well as $\thn^f$ and $\thn^g$ to represent the values of $\theta^f$ and $\theta^g$ within $\tho$ and $\thn$. 

%\begin{align}
%\label{eq:expeclengthobs}
%\expecbb{\sum_{i=1}^n\funcg(y_i\mid \rand X_i\semic \theta)+\ind{\rand C_i=i}\funcf(\rand X_i\semic \theta)\bmid \rand Y_{1:n}=y_{1:n}\semic \theta_{\ell-1}}
%\end{align}
%w.r.t. $\theta$.

\begin{lemma}
\label{lem:emheight}
Assume that $\funcg(x\semic \theta^g)=\eta(\theta^g)^tT(x)+A(\theta^g)$ with functions $\eta:\Theta^g\rightarrow\mathbb{R}^u$, $A:\Theta^g\rightarrow \mathbb{R}$ and measurable $T:\lat\rightarrow\mathbb{R}^u$.
Any value of $\theta^g$ that maximizes
\begin{align}
\label{eq:lememheight}
(\tqcc_1+\ldots+\tqcc_n)A(\theta^g)+\sum_{\ell=1}^u\eta_\ell(\theta^g)\sum_{i=1}^n \expec{T_\ell(\rand X_i)\ind{\rand C_i=i}\mid \rand Y_{1:n}=y_{1:n}}
\end{align}	
is a candidate for $\thn^g$.
\end{lemma}
\begin{proof}
This follows directly from the particular form of \funcg.
\end{proof}
The functional form of \funcg in Lemma \ref{lem:emheight} states that $\jprior$ is a member of the exponential family w.r.t. $\theta^g$.
Since $\theta^g$ and $x$ are employed in separate factors within $\funcg$, $\theta^g$ does not concern the expectations drawn from the $T_\ell(\rand X_i)\ind{\rand C_i=i}$'s.
Thus, the ensuing maximization w.r.t. $\theta^g$ can fortunately be performed based on constant expressions for these expectations.
If this was not the case, maximization in closed form would easily become a hopeless task.

Typically, we obtain maximizers of (\ref{eq:lememheight}) through differentiating and zeroing. 
In this case, we require that the $\eta_\ell$'s and $A$ are (partially) differentiable w.r.t. $\theta^g$. 
The posterior expectation of $T_\ell(\rand X_i)\ind{\rand C_i=i}$ for $\ell=1,\ldots,u$ and $i=1,\ldots,n$ needs to be feasible.
This feasibility relies on the ability to compute all the $\heights_{ji}^{T_\ell}$'s (see Equation (\ref{eq:hoff1})) as the following lemma shows.
\begin{lemma}
\label{lem:expTF}
For $\ell=1,\ldots,u$ and $0<i\leq n$ we get
\begin{align*}
\expec{T_\ell(\rand X_i)\ind{\rand C_i=i}\mid \rand Y_{1:n}=y_{1:n}}= \sum_{j=i}^{n-1}\heights_{ij}^{T_\ell}\tccc_{ij}\tqcc_{j+1}+\heights_{in}^{T_\ell}\tccc_{in}
\end{align*}
\end{lemma}
\begin{proof}
This follows from
\begin{align*}
&\prob{\rand X_i\in\LargerCdot, \rand C_i=i\mid \rand Y_{1:n}=y_{1:n}}\\
=&\sum_{j=i}^{n-1}\prob{\rand X_j\in\LargerCdot, \rand C_{j+1}=j+1, \rand C_j=i\mid \rand Y_{1:n}=y_{1:n}}
+\prob{\rand X_n\in\LargerCdot, \rand C_n=i\mid \rand Y_{1:n}=y_{1:n}}\\
=&\sum_{j=i}^{n-1}\heights_{ij}\tccc_{ij}\tqcc_{j+1}+\heights_{in}\tccc_{in}
\end{align*}	
and a subsequent integration.
\end{proof}

Lemma \ref{lem:expTF} touches each particle only once.
This paves the way for an EM algorithm for segment height where each step has a time and space complexity of $\mathcal{O}$(number of particles).
Lemma \ref{lem:expTF} and Theorem \ref{theo:expcomp} clarify that the main bottle-neck here is the feasibility and complexity of computing $\tccc_{ji}$ and $\heights_{ij}^{T_\ell}$ for $0<i\leq n$, $0\leq j\leq i$ and $\ell=1,\ldots,u$.

\begin{table}[ht]
	\centering
	\begin{tabular}{ c | c}
		$\begin{array}{c}
		\text{Distribution}\text{ and }\funcg
		\end{array}$ & $T$, $\eta$ and $A$\\
		\hline
		$\begin{array}{c}
		\text{Gaussian} \\
		\funcg(x\semic \mu,\sigma)=
		-\frac{1}{2}\big(\frac{x-\mu}{\sigma}\big)^2-\frac{1}{2}\ln(2\pi\sigma^2)
		\end{array}$&$\begin{array}{c}
		T(x)=-\frac{1}{2}(x^2, x, 1)^t\\
		\eta(\mu, \sigma)=(1\slash \sigma^2, -2\mu\slash \sigma^2, \mu^2\slash \sigma^2)^t\\
		A(\mu, \sigma)=-\frac{1}{2}\ln(2\pi\sigma^2)
		\end{array}$\\
		\hline
		$\begin{array}{c}
		\text{Laplacian} \\
		\funcg(x\semic \sigma)=
		-\frac{|x-\mu|}{\sigma}-\ln(2\sigma)
		\end{array}$&$\begin{array}{c}
		T(x)=-|x-\mu|\\
		\eta(\sigma)=1\slash \sigma\\
		A(\sigma)=-\ln(2\sigma)
		\end{array}$\\
		\hline
		$\begin{array}{c}
		\text{Inverse gamma }\\\funcg(x\semic \alpha, \beta)=
		\alpha\ln(\beta)-\ln(\Gamma(\alpha))\\-(\alpha+1)\ln(x)-\beta\slash x
		\end{array}$&$\begin{array}{c}
		T(x)=-(\ln(x), 1\slash x)^t\\
		\eta(\alpha, \beta)=\big(1+\alpha , -\beta)\big)^t\\
		A(\alpha, \beta)=\alpha\ln(\beta)-\ln(\Gamma(\alpha))
		\end{array}$\\
		\hline
		$\begin{array}{c}
		\text{Beta }\\
		\funcg(x\semic \alpha,\beta)=(\alpha-1)\ln(x)\\
		+(\beta-1)\ln(1-x)-\ln(\text{B}(\alpha, \beta))
		\end{array}$&$\begin{array}{c}
		T(x)=(\ln(x), \ln(1-x))^t\\
		\eta(\alpha, \beta)=(\alpha-1 , \beta-1)^t\\
		A(\alpha, \beta)=-\ln(\text{B}(\alpha, \beta))
		\end{array}$
	\end{tabular}
	\caption[EM estimates for segment height]{A brief collection of univariate distributions for segment height and the corresponding realizations of $T$, $\eta$ and $A$.}
	\label{fig:emsegheight}
\end{table}

Table \ref{fig:emsegheight} lists some common univariate distributions for segment height and the corresponding form of $T$, $A$, and $\eta$ in Lemma \ref{lem:emheight}.
While the normal distribution is a conjugate prior in the Gaussian change in mean model,  the Laplace distribution (Section \ref{chap:laplacian}) is suitable for the Laplacian change in location model.
The inverse gamma distribution comes into place in the Gaussian and Laplacian change in scale model and the beta distribution is involved in all sorts of change in success rate models.

The Gaussian change in mean model is completely benign.
It just requires us to compute the first and second moment, and some simple derivatives w.r.t. $\sigma$ and $\mu$.
In the Laplacian change in location model, where $\sigma$ is to be estimated, we need to compute the first absolute moments of $\heights_{ji}$.
This is elaborated in Section \ref{chap:emsigmalap}.

In the Gaussian change in variance model where the inverse gamma distribution is used, the $\heights_{ji}$'s are again of the inverse gamma type. 
Computing expectations w.r.t. these kind of simple distributions doesn't pose a problem and could be performed numerically.
In turn, the maximization involves the gamma function, which might foil a closed form maximization.
However, we may again resort to a numeric method to achieve this.
The same conclusions hold for the beta distribution, where the $\heights_{ji}$'s are typically beta distributions again.

The next lemma concerns the EM algorithm for observation with regards to a particular class of distributions.
\begin{lemma}
	\label{lem:emobs}
	Assume that $\funcf(y\mid x\semic\theta^f)=\eta'(\theta^f, y)^tT'(x)+A'(\theta^f)$ with functions $\eta':\Theta^f\times \obs\rightarrow\mathbb{R}^v$, $A':\Theta^f\rightarrow \mathbb{R}$ and measurable $T':\lat\rightarrow\mathbb{R}^v$. Any value of $\theta^f$ that maximizes
	\begin{align}
	\label{eq:lememobs}
	nA'(\theta^f)+\sum_{\ell=1}^v\sum_{i=1}^n \eta'_\ell(\theta^f, y_i)\expec{T'_\ell(\rand X_i)\mid \rand Y_{1:n}=y_{1:n}}
	\end{align}	
	is a candidate for $\thn^f$.
\end{lemma}
\begin{proof}
	This follows directly from the particular form of \funcf.
\end{proof}

In the context of Lemma \ref{lem:emobs}, we may again be able to achieve a space and time complexity of $\mathcal{O}$(number of particles) for each EM step.
Theorem \ref{theo:expcomp} and Corollary \ref{cor:exp} show that the main bottle-neck here is the feasibility and complexity of computing $\tccc_{ji}$ and $\heights_{ij}^{T_\ell'}$ for $0<i\leq n$, $0\leq j\leq i$ and $\ell=1,\ldots,v$.

However, the requirement of Lemma \ref{lem:emobs} is relatively strict.
It does, for example, not hold in the Laplacian case where we want to estimate $\sigma$, since $|y-x|$ does not decompose into two separate factors.
Therefore, we would need to compute $\expec{|y_i-\rand X_i|\mid \rand Y_{1:n}=y_{1:n}}$ for all $i=1,\ldots,n$, which might increase the overall complexity by a factor of $n$ or more.
We may thus, be more successful by other means as the next lemma shows.

\begin{lemma}
\label{lem:emobs2}
Assume that $\funcf(y\mid x\semic\theta^f)=\eta'(\theta^f)^tT'(x, y)+A'(\theta^f)$
with functions $\eta':\Theta^f\rightarrow\mathbb{R}^v$, $A':\Theta^f\rightarrow \mathbb{R}$ and measurable $T':\obs\times \lat\rightarrow\mathbb{R}^v$. 
For $0\leq j\leq i$, $0<i\leq n$ and $\ell=1,\ldots,v$ let 
\begin{align*}
\occ_{ji}^\ell=\int \sum_{k=\max(j,1)}^iT_\ell(x,y_k) \heights_{ji}(dx)
\end{align*}
Any value of $\theta^f$ that maximizes
\begin{align}
\label{eq:lememobs2}
nA'(\theta^f)+\sum_{\ell=1}^v\eta_\ell'(\theta^f)\Bigg(\sum_{i=1}^{n-1}\sum_{j=0}^i\tccc_{ji}\tqcc_{i+1}\occ_{ji}^\ell+\sum_{j=0}^n\tccc_{jn}\occ_{jn}^\ell\Bigg)
\end{align}		
is a candidate for $\thn^f$.
\end{lemma}
\begin{proof}
The  proof follows from
\begin{align*}
\expec{\sum_{i=1}^nT_\ell'(\rand X_i, y_i)\bmid \rand Y_{1:n}=y_{1:n}}=
\sum_{i=1}^{n-1}\sum_{j=0}^i\tccc_{ji}\tqcc_{i+1}\occ_{ji}^\ell+\sum_{j=0}^n\tccc_{jn}\occ_{jn}^\ell
\end{align*}
This, in turn, can be seen as follows
\begin{align*}
&\expec{\sum_{i=1}^nT_\ell'(\rand X_i, y_i)\bmid \rand Y_{1:n}=y_{1:n}}\\
=&\sum_{c_{1:n}\in\mathbb{N}^n}\prob{\rand C_{1:n}=c_{1:n}\mid \rand Y_{1:n}=y_{1:n}}\sum_{\scalethis{\begin{array}{c}j\leq i\text{ with } c_i=j \text{ and}\\
	\big(i=n\text{ or } c_{i+1}=i+1\big)
	\end{array}}}\occ_{ji}^\ell\\
=&\sum_{i=1}^{n}\sum_{j=0}^i\occ_{ji}^\ell\sum_{\scalethis{\begin{array}{c}c_{1:n}\in\mathbb{N}^n\text{ with } c_i=j \text{ and }\\
	 \big(i=n \text{ or } c_{i+1}=i+1\big)
	\end{array}}}\prob{\rand C_{1:n}=c_{1:n}\mid \rand Y_{1:n}=y_{1:n}}\\
=&\sum_{i=1}^{n-1}\sum_{j=0}^i\tccc_{ji}\tqcc_{i+1}\occ_{ji}^\ell+\sum_{j=0}^n\tccc_{jn}\occ_{ji}^\ell
\end{align*}
\end{proof}

The advantage of Equation (\ref{eq:lememobs2}) in combination with Lemma \ref{lem:emobs2} is that we just need to compute one expectation per particle.
Thus, the overall complexity of each EM step is $\mcal{O}$(number of particles).

\begin{table}[ht]
	\centering
	\begin{tabular}{ c | c | c }
		$\begin{array}{c}
		\text{\Cp problem}\\\text{ and }\funcf
		\end{array}$ & $T'$, $\eta'$ and $A'$ & $\begin{array}{c}\text{Apply}\\ \text{ Lemma}\end{array}$\\
		\hline
		$\begin{array}{c}
		\text{Gaussian change in mean}\\
		\funcf(y\mid x\semic \sigma)=\\
		-\frac{1}{2}\big(\frac{y-x}{\sigma}\big)^2-\frac{1}{2}\ln(2\pi\sigma^2)
		\end{array}$&$\begin{array}{c}
		T'(x)=-\frac{1}{2}(x^2, x, 1)^t\\
		\eta'(\sigma, y)=(1\slash \sigma^2, -2y\slash \sigma^2, y^2\slash \sigma^2)^t\\
		A'(\sigma)=-\frac{1}{2}\ln(2\pi\sigma^2)
		\end{array}$&\ref{lem:emobs}\\
		\hline
		$\begin{array}{c}
		\text{Gaussian change in variance}\\
		\funcf(y\mid x\semic \mu)=\\
		-\frac{1}{2}\frac{(y-\mu)^2}{x}-\frac{1}{2}\ln(2\pi x)
		\end{array}$&$\begin{array}{c}
		T'(x)=-\frac{1}{2}(1\slash x, 1\slash x, 1\slash x, \ln(2\pi x))^t\\
		\eta'(\mu, y)=\big(y^2 , -2y \mu, \mu^2, 1\big)^t\\
		A'(\mu)=0
		\end{array}$&\ref{lem:emobs}\\
		\hline
		$\begin{array}{c}
		\text{Laplacian change in location}\\
		\funcf(y\mid x\semic \sigma)=\\
		-\frac{|y-x|}{\sigma}-\ln(2 \sigma)
		\end{array}$&$\begin{array}{c}
		T'(x, y)=-|y-x|\\
		\eta'(\sigma)=1\slash\sigma\\
		A'(\sigma)=-\ln(2\sigma)
		\end{array}$&\ref{lem:emobs2}\\
		\hline
		$\begin{array}{c}
		\text{Laplacian change in scale}\\
		\funcf(y\mid x\semic \mu)=\\
		-\frac{|y-\mu|}{x}-\ln(2 x)
		\end{array}$&$\begin{array}{c}
		T'(x)=-(1\slash x, \ln(2x))^t\\
		\eta'(\mu, y)=(|y-\mu|, 1)^t\\
		A'(\mu)=0
		\end{array}$&\ref{lem:emobs}
	\end{tabular}
	\caption[EM for data distributions]{A brief collection of univariate observational distributions for specific \cp problems. It states the functions $T'$, $\eta'$ and $A'$, and the lemma that should be applied in order to build an EM algorithm for the unknown parameters.}
	\label{fig:emobs}	
\end{table}

Table \ref{fig:emobs} provides a brief list of univariate observational distributions for certain \cp problems. 
We can also read the corresponding form of $T'$, $\eta'$ and $A'$, and the lemma that should be applied in order to develop an EM algorithm.

As before, if we exploit conjugacy, we will obtain simple posterior distributions.
Thus, computing logarithmic or inverse moments as in the Gaussian change in variance and Laplacian change in scale examples can be done numerically.

\subsubsection{EM algorithm for the scales in the well-log example}
\label{chap:emsigmalap}
In the following, we want to apply the results of Section \ref{chap:emheiobs} to the well-log example (Sections \ref{chap:laplacian}, \ref{chap:well_log_forward}, \ref{chap:well-log-pointw} and \ref{chap:well-log-para}).
Our task is to develop an EM algorithm for $\tau$ and $\sigma$.
$\mu$ does not meet the requirements of Lemma \ref{lem:emheight} since $|x-\mu|$ does not decompose into two separate factors.

In order to build an EM algorithm for $\tau$, we apply Lemma \ref{lem:emheight}.
Therefore, we need to compute
\begin{align*}
\mcc_0=\sum_{i=1}^n\expec{|\rand X_i-\mu|\cdot \ind{\rand C_i=i}\bmid \rand Y_{1:n}=y_{1:n}}
\end{align*}
According to Lemma \ref{lem:expTF}, this entails the computation of $\heights_{ji}^T$ with $T(x)=|x-\mu|$ for all $0<i\leq n$ and $0\leq j\leq i$.
Finally, we obtain $\tau_{new}=\mcc_0\slash (\tqcc_1+\ldots+\tqcc_n)$ through maximizing $-(\tqcc_1+\ldots+\tqcc_n)\ln(2\tau)-\mcc_0\slash\tau$ w.r.t $\tau$.

Hence, it remains to show how to compute $\heights_{ji}^T=\int |x-\mu|\heights_{ji}(dx)$.
Therefore, consider that $\heights_{ji}^T\cdot \Zcc_{ji}^0=\int |x-\mu|\euler{-\frac{|x-\mu|}{\tau}-\sum_{\ell=j}^i\frac{|x-y_\ell|}{\sigma}}dx$.
This elementary integration is explained in the following corollary.

\begin{corollary}
	\label{cor:intlaplaceabs}
	In the context of Corollary \ref{cor:intlaplace} the following holds
	\begin{align}
	\label{eq:laplacesumabs}
	&\int |x-\mu|\euler{-\sum_{\ell=1}^r\frac{|x-z_\ell|}{\sigma_\ell}}dx
	=\sum_{\ell=1}^r\euler{\dcc_\ell z_\ell+\ecc_\ell}(\bcc_{\ell-1}^\ell-\bcc_{\ell}^\ell)
	\end{align}
with
\begin{align*}
&\bcc_j^k=\text{sgn}(z_k-\mu)\cdot\begin{cases}
\frac{z_k-\mu}{\dcc_j}-\frac{1}{\dcc_j^2} &\dcc_j\neq 0\\
z_k^2\slash 2-\mu z_k &\dcc_j=0	
\end{cases}
\end{align*}
for $j=0,\ldots,r$ and $k=1,\ldots,r$.
\end{corollary}
\begin{proof}
The proof is similar to the proof of Corollary \ref{cor:intlaplace}.
\end{proof}

An EM algorithm for $\sigma$ can be derived by means of Lemma \ref{lem:emobs2}.
To this end, we need to compute 
$\occ_{ji}=\int \sum_{\ell=\max(j,1)}^i|x-y_{\ell}| \heights_{ji}(dx)$
for all $0<i\leq n$ and $0\leq j\leq i$.
With this, we further compute 
\begin{align*}
\mcc_1=\sum_{i=1}^{n-1}\sum_{j=0}^i\tccc_{ji}\tqcc_{i+1}\occ_{ji}+\sum_{j=0}^n\tccc_{jn}\occ_{jn}
\end{align*}
Finally, we obtain $\sigma_{new}=\mcc_1\slash n$ by maximizing $-n\ln(2\sigma)-\frac{\mcc_1}{\sigma}$ w.r.t. $\sigma$.

The integrand in $\occ_{ji}$ represents a cumbersome function we haven't dealt with yet.
We can avoid it by considering that
\begin{align*}
\occ_{ji}=\sigma\int\frac{|x-\mu|}{\tau} + \sum_{\ell=j}^i\frac{|y_\ell-x|}{\sigma} \heights_{ji}(dx)-\frac{\sigma}{\tau}\heights_{ji}^T
\end{align*}
In this case, the integrand and the exponent within $\heights_{ji}(dx)\densslash\xi(dx)$ differ only in terms of their sign.
The following lemma explains this more elementary integration.

\begin{corollary}
	\label{cor:intlaplacesum}
	In the context of Corollary \ref{cor:intlaplace} the following holds
	\begin{align}
	\label{eq:laplacesumsum}
	&\int \sum_{\ell=1}^r\frac{|x-z_\ell|}{\sigma_\ell}\euler{-\sum_{\ell=1}^r\frac{|x-z_\ell|}{\sigma_\ell}}dx
	=\sum_{\ell=1}^r\euler{\dcc_\ell z_\ell+\ecc_\ell}(\bcc_{\ell-1}^\ell-\bcc_{\ell}^\ell)
	\end{align}
	with 
	\begin{align*}
	&\bcc_j^k=\cdot\begin{cases}
	z_k-\frac{1-\ecc_j}{\dcc_j} &\dcc_j\neq 0\\
	\ecc_jz_k &\dcc_j=0	
	\end{cases}
	\end{align*}
	for $j=0,\ldots,r$ and $k=1,\ldots,r$.
\end{corollary}
\begin{proof}
The proof is similar to the proof of Corollary \ref{cor:intlaplace}.
\end{proof}

Equipped with this knowledge, we now conduct a brief simulation study to examine the EM algorithm for $q$ and $\sigma$ in the Laplacian change in median model that employs a negative binomial distribution for segment length.
Therefore, we compare different values of $r$ with regard to the simultaneous estimation of the two parameters.

In frequentist settings, the \cp locations are usually not considered as random.
Thus, frequentists deem the prior distribution for segment length just as a necessary evil of the Bayesian approach.
In this context, the main interest lies in estimates for $\sigma$.

However, the choice of the length distribution has to be taken with care, because it impacts the estimation of $\sigma$.
Once again, this can be explained by the collinearity of the segment lengths and the observational scale (see Section \ref{chap:well-log-para}).
In the following, we are particularly interested in understanding this impact.

We repeatedly sample 4050 datapoints from a Laplace distribution with scale equal to 1 and changing median.
The segment heights (the medians) are sampled from a Laplace distribution with median equal to 0 and scale equal to 10.
Each data set contains 12 random \cps drawn uniformly.

\begin{figure}[ht]
	\includegraphics[width=\linewidth]{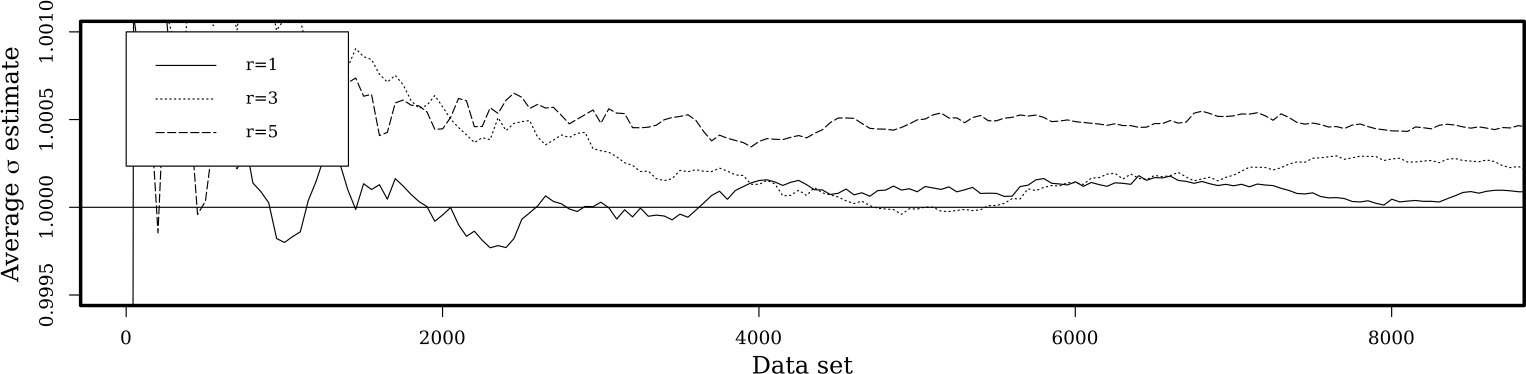}
	\caption[Average EM estimates for variances]{The graphs show how the average EM estimate for $\sigma$ evolves for different values of $r$ at an increasing number of data sets. The horizontal line indicates the true value of $\sigma$.
	}
	\label{fig:expsigma}
\end{figure}

Figure \ref{fig:expsigma} displays the average EM estimate for $\sigma$ w.r.t. to different choices of $r$ and an increasing number of data sets.
The horizontal line marks the true value of $\sigma$.
Of course, convergence towards certain limits can not be derived safely hereof, but tendencies are apparent.
All graphs indicate a positive bias, which seems to increase with $r$.
Anyway, the three biases concern at most the fourth digit of the true value of $\sigma$ and are thus tiny.

The geometric distribution, i.e. $r=1$, places the \cps uniformly regardless of $q\in(0,1)$ if their number is given.
It is thus the most natural choice for the distribution of segment length here.
This explains why $r=1$ yields the smallest (or perhaps no) bias from the true value of $\sigma$.

On the other hand, a uniform distribution for \cps tends to produce very short segments.
Those segments combined with comparably small height changes can barely be distinguished from random perturbations in the data.
Since the number of \cps is uncertain during the estimation, segments are overlooked occasionally.
This may result in an overestimation of $\sigma$.
Furthermore, increasing $r$ inhibits short segments enhancing this effect.

\begin{table}[ht]	
	\centering
	\begin{tabular}{l|c|c|c|c}
		$r\quad$ & MSE & 		$\begin{array}{c}
		\text{Average expected}\\\text{\cp count}
		\end{array}$ & Average MAP & $\begin{array}{c}
		\text{Average expected}\\\text{segment length}
		\end{array}$ \\
		\hline
		1&  0.000243&  12.0178&  11.562&  336.99\\
		3&  0.000249&  11.8754&  11.513&  334.34\\
		5&  0.000251&  12.1006&  11.471&  327.28
	\end{tabular}
	\caption[Several statistics derived from EM estimates]{For different choices of the negative binomial distribution, the table shows several statistics derived from the EM algorithm w.r.t. to 10000 data sets.}
	\label{fig:tablesigma}
\end{table}

Table \ref{fig:tablesigma} displays different statistics, derived from the EM estimation of $\sigma$ and $q$ w.r.t. 10000 data sets.
The mean squared error of the estimation for $\sigma$ appears to be lowest for $r=1$ and increases with $r$.
Furthermore, the expected \cp count seems to be a good (roughly unbiased) estimator for the true \cp count when the length distribution is chosen appropriately.
In accordance with the above reasoning, $r=1$ is appropriate here.
The MAP estimator overlooks a considerable portion of \cps for all choices of $r$.
This reflects its robustness towards uncertain \cps.

A potentially contra intuitive observation is that the expected segment length gets smaller with increasing $r$.
This is because the negative binomial distribution, for larger values of $r$, needs to be tightened more and more, in order to adjust to the entirety of possible segment lengths. 
At the same time, the expected \cp counts exhibit an irregular and odd picture.

Therefore, we can conclude that too large values of $r$ may provoke overfitting towards strongly peaked length distributions.
This is highly undesirable and calls for a careful determination of $r$.
On the other hand, increasing $r$ helps to mitigate the effect of sudden outliers by inhibiting short segments (see also  Section \ref{chap:well-log-para}).

\subsection{Discussion}
\label{chap:discussion_spike}

In this section, we developed a wide range of exact inference strategies for a certain class of \cp models.
The central achievement of this work is Theorem \ref{theo:expcomp}.
It shows that under good conditions all the considered algorithms exhibit a linear complexity w.r.t. the data size.
Even though this requires an approximation through pruning, the results are still convincing enough to be referred to as virtually exact.

In order to perform exact inference over \cp locations, it is not necessary that all involved computations are feasible in closed form.
Instead, a sufficient requirement is that certain elementary integrals, i.e.
\begin{align}
\label{eq:elint}
\int \prod_{\ell=j}^{i}\dens(\rand Y_\ell=y_\ell\mid \rand X_\ell=x)\jprior(dx)
\end{align}
for $0<j\leq i\leq n$, can be computed either analytically or numerically.
%In this context, (\ref{eq:elint}) has to be positive.
Therewith, the $\ccc_{ji}$'s and $\heights_{ji}$'s can be derived in a purely arithmetic fashion.

In turn, for pointwise inference over segment heights by means of expectations w.r.t. functions $\funcf$, we require that 
\begin{align*}
\int \funcf(x)\prod_{\ell=j}^{i}\dens(\rand Y_\ell=y_\ell\mid \rand X_\ell=x)\jprior(dx)
\end{align*}
is computable for $0<j\leq i\leq n$.

Even though these calculations determine the desired posterior quantities uniquely, the observational densities are generally not unique, i.e. they may differ on null sets.
This is related to the Borel-Kolmogorov paradox \citep{kolmogorov2018foundations} and gives rise to ambiguous inference results.
A common workaround is to impose topological continuity to the employed densities.

Furthermore, there might be applications where (\ref{eq:elint}) is not applicable, for example, if it yields zero.
This may occur when the observations depend somehow deterministically on segment heights that are continuous.
Under such circumstances, we would refrain from our density based construction and resort to the measure theoretic notion of conditional distributions instead \citep{kolmogorov2018foundations}.

Our inference algorithms are implemented in discrete time exclusively.
Even though an extension to continuous time is in principle viable, it involves dealing intensively with the backward residual times \citep{stewart2009}.
These quantities are trivial for exponentially distributed segment lengths, but generally of a delicate form.

The main constrained in this section is the imposed independence induced by \cps.
Unfortunately, there are very basic applications that do not meet this constrained, like a continuous piecewise linear regression.
The problem here is the exponential growth of the number of particles in the forward algorithm, which requires a sophisticated adaption of the pruning algorithm to this far more demanding case.
This topic lies beyond the scope of this thesis, but \cite{fearnhead_dependency} may provide a good starting point for such an algorithm.
They elaborate an SMC approach \citep{doucet2001introduction, fearnhead2003line} for online \cp detection.

As a general remark, applying mixtures to Lemma \ref{lem:bayesstepfor} is prone to fail since it yields mixtures for the $\heights_{ji}$'s with an exponential growth in the number of components.
If mixtures are involved, we should instead consider resorting to numerical integration by means of Lemma \ref{lem:segheight}.

A further, rather inconspicuous restriction is that the observations are supposed to be mutually independent given the segment heights.
Some applications may require instead that the observations follow a joint segmentwise and not a pointwise distribution.

This can be incorporated into our algorithms straight away if the segmentwise distributions are consistent in regards to the length of their segment.
The consistency condition comes into play in the forward algorithm since it works on partial segments.
If it doesn't hold, any partial segment would have to be represented explicitly as a mixture w.r.t. all possible segments containing this partial segment.

However, the forward algorithm may be sidestepped by means of a purely segmentwise approach like the sampling method for \cp models described in \cite{exact_fearnhead}.
In essence, it provides the $\tccc_{ji}$'s.
All our backward and pointwise algorithms may build on top of this since they rely on the $\tccc_{ji}$'s directly.

A famous modeling approach, which hasn't been considered in this thesis, is to carry along the number of \cps at each timepoint.
For this purpose, we add a third index to $\ccc_{ji}$, i.e. the number of \cps that occurred hitherto.

\cite{Rigaill2012} and \cite{aston2012implied} employ this method to infer the pointwise probabilities of seeing certain \cps.
It is further necessary for the use of a prior for \cp locations that pins the \cp count down to a constant.
Section \ref{chap:unicps} elaborates such a uniform prior, which is ready to use within our algorithms.
By assuming that $k$ is the maximum number of \cps considered, the space and time complexity is then increased by a factor of $k$ regardless of pruning.
Furthermore, the additional split of the probabilities is numerically more challenging.

At several occasions, we saw that our methods share strong similarities with the popular Kalman-filter (and Kalman smoothing).
The trajectory of the pointwise posterior expectations of the latent process poses a MAP for the Kalman-filter.
Although this does not hold for \cp models, Figure \ref{fig:laplace_marginals}(a) showed that we may still gather valuable information from such expectations.

Moreover, Figure \ref{fig:laplace_marginals}(a) reveals that the expected segment heights behave fairly steady amid segments even in the presence of strong outliers.
This is because \cp models adjust single segment heights closely to whole sequences of neighboring observations.
At the same time, the expected segment heights can move very rapidly at \cp locations.

In contrast, the volatility of the expectations in the Kalman-filter rely heavily on the variances of single observations.
They can only move away abruptly if the corresponding variances are sufficiently small, i.e. the measurement is precise.
Conversely, a steady behavior w.r.t. strong perturbations requires large variances.

Thus, \cp models introduce a very specific behavioral pattern, which could perhaps be of significant practical interest.
This is subject to further research.

The mathematics developed for the estimation of the success probability of the geometric distribution provides interesting links to the tremendously important topic of variable selection in statistical models \citep{guyon2003introduction}.
It can be stated as follows.
Assume that we deal with a family of statistical models and $n$ parameters with values in $\Theta$.
Each model is parameterized by a subsets of these $n$ parameters.
Given an observed data set, the task is to find one or more parsimonious models that represent the data well.

Therefore, we utilize their likelihoods, say $L_{I, \theta}$, where $I\sub\{1,\ldots,n\}$ represents a subset of the parameters and $\theta\in \Theta^{\#I}$ their values.
The famous Bayesian information criterion (BIC) \citep{schwarz1978estimating} and the Akaike information criterion (AIC) \citep{akaike1973information} are commonly used here.
In this primal context, they would be defined as follows
\begin{align*}
&\text{BIC}_I=\ln(n)\cdot\#I-2\ln\Bigg(\int L_{I,\theta} \mathcal{P}_I(d\theta)\Bigg)\\
&\text{AIC}_I=2\cdot \#I-2\ln\Big(\max_{\theta\in\Theta^{\#I}}\big\{L_{I,\theta}\big\}\Big)
\end{align*}
whereby $\mathcal{P}_I$ is a prior over a measurable space build from $\Theta^{\#I}$.
Among all $I\sub\{1,\ldots,n\}$, we would choose one that maximizes either $\text{BIC}_I$ or $\text{AIC}_I$.

Both approaches perform a maximization that penalizes the cardinality of $I$.
Therewith, they promote small and reasonable sets of parameters.
While the BIC employs a prior in order to marginalize the parameters out, the AIC incorporates their maximum.

Unfortunately, the coefficient of $\#I$ is in both cases static and thus, these approaches appear to be a bit inflexible.
This gets clearer through a natural analogy that uses the Bernouilli distribution as a prior for the elements of $I$, i.e. $\prob{\rand I=I}=q^{\#I}(1-q)^{n-\#I}$.
The negative loglikelihood of this Bayesian model reads
\begin{align}
\label{eq:loglikelihoodbic}
-n\ln(1-q)+\ln\big((1-q)\slash q\big)\cdot\#I-\ln\big(L_{I,\theta}\big)
\end{align}
This yields an optimization problem similar to BIC and AIC, but with an arbitrary penalization for $\#I$.

In this context, BIC uses a prior for $I$ that employs less than 1 parameter on average.
This steers the estimation of the number of parameters to a parsimonious choice and thus focuses on specificity.
However, this is inappropriate for \cp problems and also for regression problems, where sensitivity towards individual parameters is of high value.

Hence, it might be reasonable to adjust the coefficient, say $\gamma$, of $\#I$ in (\ref{eq:loglikelihoodbic}) to the observed data in order to allow for a decent sensitivity.
This could be done by means of the EM algorithm of Section \ref{chap:geomem} based on samples of $\rand I$ \citep{spike_slab_mitchel, george_mccullogh_gibbs, green1995}.
According to Lemma \ref{lem:geomem}, in each EM step, we would just need to approximate the expected size of $\rand I$.

Having estimated $\gamma$, we may then maximize 
\begin{align*}
\gamma\cdot\#I-\ln\Bigg(\int L_{I, \theta} \mathcal{P}_I(d\theta)\Bigg)\quad\text{ or }\quad \gamma\cdot \#I-\ln\Big(\max_{\theta}\big\{L_{I, \theta}\big\}\Big)
\end{align*}
w.r.t. $I$, a procedure which is sometimes referred to as an empirical Bayes method \citep{casella1985introduction}.
This is subject to further research.

This whole section is accompanied by a Laplacian change in median example.
It demonstrates the special case where the observational distribution is not part of the exponential family w.r.t. its segment height.
In this case, simple recursive Bayes formulas are not available and thus, Lemma \ref{lem:segheight} combined with plain integration becomes the method of choice.
Therewith, we have conveyed the ideas behind our algorithms and how to implement them in more challenging scenarios.
Section \ref{chap:well-logsimcred} will conclude our Laplacian change in median example by a comparison with the model proposed by \cite{exact_fearnhead}.

\newpage

\newpage\null\newpage
\section{Simultaneous credible regions for multiple \cp locations}
\label{chap:simcredreg}

\subsection{Outline}
In this last big section, we present a novel approach for Bayesian \cp models, that is based on \cp samples.
It facilitates the examination of the distribution of \cps through a novel set estimator.
For a given level $\alpha$, we aim at smallest sets that cover all \cps with a probability of at least $1-\alpha$.

These so-called smallest simultaneous credible regions, computed for certain values of $\alpha$, provide parsimonious representations of the possible \cp locations.
In addition, combining them for a range of different $\alpha$'s enables very informative yet condensed visualizations. 
Therewith we allow for the evaluation of model choices and the analysis of \cp data to an unprecedented degree.

This approach exhibits superior sensitivity, specificity and interpretability in comparison with highest density regions, marginal inclusion probabilities and confidence intervals inferred by \stepR. 
Whilst their direct construction is usually intractable, asymptotically correct solutions can be derived from posterior samples. 
This leads to a novel NP-complete problem.
Through reformulations into an Integer Linear Program we show empirically that a fast greedy heuristic computes virtually exact solutions.

\subsection{Introduction}

%Detecting \cps in time series is an important task.
%For example, while observing the gating behavior of ion channels \citep{struct_change_ion_channel,modal_gating_ivo}.
%There are many algorithms for, and scientific publications on, detecting multiple \cps in time series, 
%such as frequentist approaches \citep{friedrich, frick_munk_sieling} and Bayesian approaches \citep{exact_fearnhead, ocpd, fearnhead_online}. 
%A more exhaustive overview of existing methods can be found in \cite{eckley2011analysis}.

Detecting \cps in a time series usually comes down to deciding on a set of \cp locations.
Thus, Bayesian frameworks aim to infer a set valued random variable that gives a reasonable representation of this decision \citep{exact_fearnhead, ocpd, lai2011simple}.
The non-deterministic nature of these so-called \emph{posterior random \cps} expresses the uncertainty of their location.

\cite{Rigaill2012} illustrates this uncertainty by means of a Bayesian model with exactly two \cps.
They plot for all  possible pairs of timepoints  the posterior probability of being these \cps.
The results indicate both that posterior random \cps are highly dependent and that generally more than one combination is likely. 
Unfortunately, this approach is not suitable to monitor the distribution of more than two \cps.
It is a crucial fact that  the space of possible \cp locations is very high-dimensional even for time series of moderate size. 
Thus, an extensive exploration is a nontrivial task.

In Bayesian research, summaries of \cp locations, uncertainty measurements or model selection criteria are often provided by means of marginal \cp probabilities.
\cite{Rigaill2012} gives a general consideration of this approach, but it has always enjoyed great popularity in the \cp community. 
See, for example, \cite{perreault2000retrospective, lavielle2001application, tourneret2003bayesian, exact_fearnhead, hannart2009bayesian, fearnhead_dependency, lai2011simple, aston2012implied, nam2012quantifying}.
Marginal \cp probabilities, as shown in \cp histograms, simply consist of the probabilities for a \cp at each timepoint and thus, they are pointwise statements.
However, due to the uncertainty of their location, (even single) \cps cannot be explored comprehensively by pointwise statements.
On these grounds, we present a novel approach that incorporates all possible \cp locations simultaneously.

Let $y=(y_1,\ldots,y_n)$ be a time series and let $\rand C\sub\{1,\ldots,n\}$ be the posterior random \cps.
A region that contains all \cps simultaneously with a probability of at least $1-\alpha$ is an $A\sub\{1,\ldots,n\}$ with $\prob{\rand C\sub A}\geq 1-\alpha$.
We call such an $A$ a \emph{simultaneous $\alpha$ level credible region} or simply \emph{credible region} if there is no risk of confusion.
It provides a sensitive assessment of the whole set of possible \cp locations.
However, in order to obtain specific assessments as well, we seek a smallest such region.
This means we seek an element of
\begin{align}
\label{eq:minprob}
\primeprob{\alpha,\rand C}\defeq\argmin\limits_{A\sub \{1,\dots,n\}}\Big\{\cardi{A}\bmid \prob{\rand C\sub A}\geq 1-\alpha\Big\}
\end{align} 
whereby \cardi{A} is the cardinality of $A$.

\begin{figure}[ht]
\begin{center}
\includegraphics[width=\linewidth]{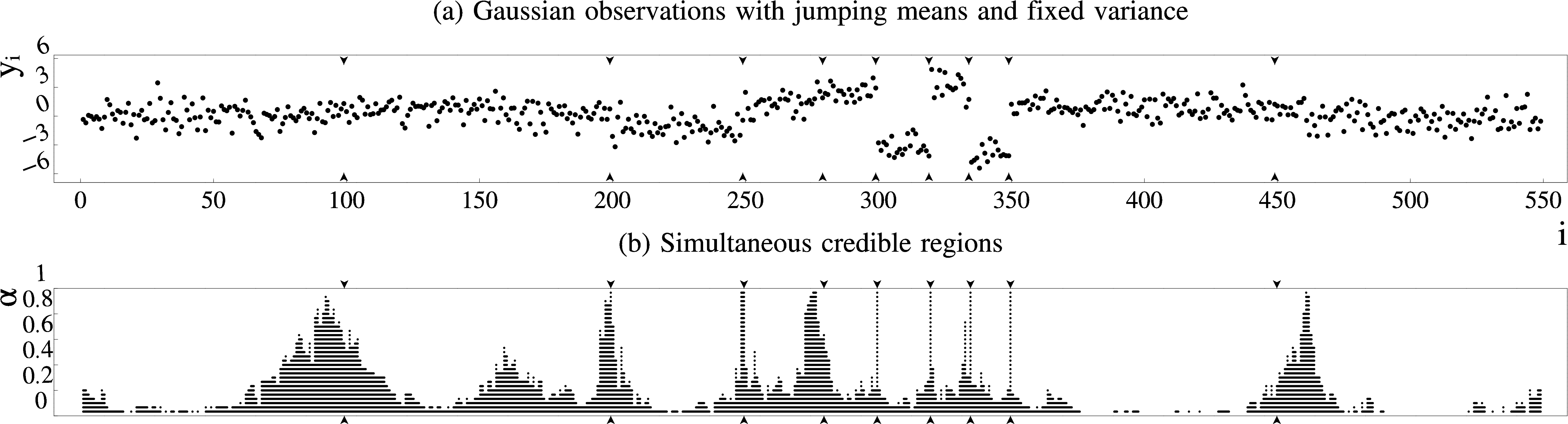} 
\caption[Gaussian data and credible regions]{(a) A time series.
(b) 29 approximated smallest simultaneous $\alpha$ level credible regions, each of them   drawn as a broken horizontal line.}
\label{fig:introduction}
\end{center}
\end{figure}

Figure \ref{fig:introduction} demonstrates smallest credible regions by means of an example. 
The data points in (a) were drawn independently from a normal distribution having a constant variance of 1 and mean values that are subject to successive changes.  
The true \cps are marked by small vertical arrows.
To build an exemplary Bayesian model here, we  choose a prior for the \cp locations and mean values.
The time from one \cp to the next is geometrically distributed with success probability $\frac{3}{550}$ and the mean values are distributed according to $\mathcal{N}(0,25)$.

Smallest credible regions are visualized in (b). For each $\alpha\in\Big\{\frac{1}{30}, \ldots, \frac{29}{30}\Big\}$ the plot shows one such (approximated) region as a collection of horizontal lines.
The value of $\alpha\in[0,1]$ controls the size of the regions and thus, governs the trade of between specificity and sensitivity.

The broadness of the credible regions around a true \cp expresses the uncertainty of its locations.
At the same time, a pointed shape reveals that the model is in favor of certain \cp locations.
Existing visualization techniques, like \cp histograms, are unable to go beyond these two estimates.
However, it is also crucial to get an impression of the sensitivity of the model towards a true \cp.
We can examine this by looking at the height of the peak that relates to the true \cp.
The higher the peak, the higher the model's sensitivity.
We refer to this height as the importance of the true \cp.
The importances of the true \cps in Figure \ref{fig:introduction} are always larger than 0.9.

Of course, importance (as well as broadness and shape) can also be used for \cp data without true \cps.
There, we look at the peaks that belong to the features of interest.
Examples for features, which may occur in practice, are changes in mean, variance, slope or any other change in distribution. 
Furthermore, anomalies like outliers are of concern as well although they are usually supposed to be skipped by the model.
Figure \ref{fig:introduction} shows a nuisance feature in form of a small irregularity at around 170 with an importance of around 0.4.

By means of importance we can conveniently evaluate if the model is sensitive towards the desired features but skips the nuisance ones.
Conversely, we can also detect the relevant features in \cp data on the basis of a given model.
Most notably, this does not require any previous knowledge about the number of features or their positions.
This novel concept is one of the main outcomes of the present work. 
It allows for a much more detailed analysis of \cp models and \cp data than hitherto possible.
Figure \ref{fig:introduction}(b) demonstrates this.
It shows that the \cp model in use is sensitive towards the desired true \cps and specific towards random distortions in the data.

The outline of this work is as follows.
In Section \ref{chap:statapp} we consider the above problem from a general statistical and algorithmic viewpoint. 
We examine alternative approaches in Section \ref{chap:alt_approc}.
Section \ref{chap:simcredchanpoi} starts with an overview over sampling strategies in \cp models and deals with importance in a broader and more formal way.
Afterwards, we compare our results with the existing approaches.
We examine Dow Jones returns and demonstrate how credible regions can be used in order to perform model selection in Section \ref{chap:gaussian_change_in_variance}.
Finally, we discuss our results in Section \ref{chap:discussion_sim}.

\subsection{The Sample Based Problem}
\label{chap:statapp} 

In this section we investigate the computational and mathematical foundations of our approach.
We show how to approximate credible regions in an asymptotic manner, examine the complexity and elaborate suitable algorithms.

We are given an arbitrary random set $\rand C\sub\{1,\ldots,n\}$.
Let $\distributionemp$ be the distribution of $\rand C$, i.e. $\distribution{A}=\prob{\rand C=A}$ for all $A\sub\{1,\ldots,n\}$. 
In this section, we are mainly concerned with the task of finding an element of \primeprob{\alpha, \rand C} (see Equation (\ref{eq:minprob})), where $\alpha$ can take any value in $[0,1]$.

Deriving $\prob{\rand C\sub A}$ is generally intractable since it arises from a summation of $2^{\#A}$ probabilities. 
Likewise, finding an element of \primeprob{\alpha,\rand C} requires a search over the $2^n$ subsets of $\{1,\dots,n\}$. 
Hence, we address this problem in an approximate manner by using relative frequencies instead of probabilities.
\begin{definition}
For $A\sub\{1,\dots,n\}$ and $s_1,\dots,s_m\sub\{1,\ldots,n\}$, let 
\begin{align*}
\rel{A,s_{1:m}}\defeq \frac{1}{m}\sum_{i=1}^m\ind{s_i\sub A}
\end{align*}
where $\ind{\dots}\in \{0,1\}$ is the indicator function, which is equal to $1$ iff the bracketed condition is true.
$\rel{A,s_{1:m}}$ corresponds to the relative frequency by which the subsets are covered by $A$.
\end{definition}
In the following theorem, we will show that having independent samples $s_1,\dots,s_m\sub\{1,\ldots,n\}$ from $\distributionemp$,
we can approximate an element of \primeprob{\alpha,\rand C} by finding an element of
\begin{align*}
\SampleProb{\alpha}{s_{1:m}}\defeq\argmin\limits_{A\sub \{1,\dots,n\}}\Big\{\cardi{A}\bmid \rel{A,s_{1:m}}\geq 1-\alpha\Big\}
\end{align*}
We denote this problem as the \emph{Sample Based Problem (\SBP)}.

\newcommand*{\solutiontolem}{A_\infty}
\begin{theorem}
\label{lem:convergence}
Let $\rand S_1,\rand S_2,\dots\sub\{1,\ldots,n\}$ be independent random sets distributed according to
\distributionemp and $\alpha\in[0,1]$. 
If there exists an $\solutiontolem\in\primeprob{\alpha,\rand C}$ with $\prob{\rand C\sub
\solutiontolem}>1-\alpha$, then $\SampleProb{\alpha}{S_{1:m}}\sub\primeprob{\alpha,\rand C}$ eventually almost surely.
\end{theorem}

\begin{proof}[Proof]
Let $A\sub\{1,\ldots,n\}$, since the $\ind{\rand S_i\sub A}, i=1,2,\dots$ are
i.i.d. with finite expectations, the strong law of large numbers states that
$\limi{m}\rel{A,\rand S_{1:m}}=\prob{\rand C\sub A}$ almost surely. 
Since $2^{\{1,\ldots,n\}}$ is finite, $\rel{A,\rand S_{1:m}}$ even converges to $\prob{\rand C\sub A}$ for all $A\sub\{1,\ldots,n\}$ almost surely.

Let $s_1,s_2,\ldots\sub\{1,\ldots,n\}$ be an arbitrary but fixed sequence with $\limi{m}\rel{A,s_{1:m}}=\prob{\rand C\sub A}$ for all $A\sub\{1,\ldots,n\}$.
For $A\sub\{1,\dots,n\}$ with $\prob{\rand C\sub A}< 1-\alpha$ we can pick an $m_A\in\mathbb{N}$ such that
$\rel{A,s_{1:m}}<1-\alpha$ for all $m>m_A$.
Hence,
$\SampleProb{\alpha}{s_{1:m}}\sub\{A\sub\{1,\dots,n\}\mid\prob{\rand C\sub A}\geq 1-\alpha\}$
holds for all $m>m_1\defeq\max\{m_A\mid \prob{\rand C\sub A}<1-\alpha\}$.
Furthermore, since $\prob{\rand C\sub \solutiontolem}>1-\alpha$, we can pick an $m_0\in\mathbb{N}$ with
$\rel{\solutiontolem,s_{1:m}}>1-\alpha\ $ for all $m>m_0$. 
Thus, by choosing $\ell=\max\{m_0,m_1\}$, we obtain $\SampleProb{\alpha}{s_{1:m}}\sub\primeprob{\alpha,\rand C}$ for all $m>\ell$.
\end{proof}

Since the exact value of $\alpha$ is usually irrelevant, the condition that there exists a certain solution can in most cases be neglected.

\subsubsection{Reformulation of the Sample Based Problem and its complexity}
\label{chap:impl_and_complexity}
In this section we consider the complexity of \SBP, formulate it as an \emph{Integer Linear Program} (ILP)
and introduce a fast and fairly accurate approximation by a greedy method. 
The next theorem shows that there is no hope
to find a polynomial time algorithm to solve \SBP unless $P= NP$ \citep{Garey_johnsons}.
The corresponding proof is provided in the appendix (Section \ref{chap:minkcov}). 
\begin{theorem}[Marc Hellmuth]
\label{theo:np}
The (decision version of) \SBP is NP-complete.
\end{theorem}

\SBP is NP-hard and thus, there is no polynomial-time algorithm to 
optimally solve a given \SBP-instance. Nevertheless, \SBP 
can be formulated as an ILP and thus, \SBP becomes accessible to highly efficient 
ILP solvers \citep{meindl2013analysis}. Such ILP solvers can be employed 
to optimally solve at least moderately-sized \SBP-instances. 
The ILP formulation is as follows. 
\begin{ILP_problem}[Marc Hellmuth]
\label{rem:simplesbp}
Following the notion of Section \ref{chap:statapp}, 
we are given a set of samples $s_1,\ldots,s_m\sub\{1,\ldots,n\}$ and an $\alpha\in[0,1]$.
We declare binary variables $U_i,F_j\in\{0,1\}$ for all $i\in\{1,\ldots,n\}$ and $j\in\{1,\ldots,m\}$.
For $i\in\{1,\ldots,n\}$ let $\vertcov(s_{1:m},i)\defeq\{j\mid i\in s_j\}$ denote the set of samples that contain $i$.
Now define the following constraints
\begin{align*}
&\text{\textbf{\emph{(I)}}}\quad\sum\limits_{j=1}^m F_j\geq m\cdot(1-\alpha)\\
&\text{\textbf{\emph{(II)}}}\quad\forall i\in\{1,\ldots,n\}:\ \sum\limits_{j\in \vertcov(s_{1:m},i)} (1-F_j)\geq  \cardi{\vertcov(s_{1:m},i)}\cdot (1-U_i)
\end{align*}
under which the objective function $\sum_{i=1}^n U_i$ needs to be minimized. 

Having computed an optimum, $A=\{i\mid U_i=1\}\in\SampleProb{\alpha}{s_{1:m}}$, i.e. this set is  a solution of \SBP.
\end{ILP_problem}

The binary variable $F_j$ represents the $j$-th sample. 
If  $F_j = 1$ then $s_j\sub A$.
Thus, the constraint \textbf{(I)} states that $A$ covers at least $m\cdot(1-\alpha)$ samples, which is equivalent to $\rel{A,s_{1:m}}\geq 1-\alpha$ in the definition of $\SampleProb{\alpha}{s_{1:m}}$.
The constraint \textbf{(II)} states that timepoint $i$ can only be dropped from $A$, if $i\in s_j $ implies that $F_j=0$.
From the perspective of \SBP these constraints are obviously necessary.
The proof of sufficiency is provided in the appendix (Section \ref{chap:ilpproof}).

A benchmark of several ILP solvers can be found in \cite{meindl2013analysis}.
Following this advice, we use CPLEX V12.6.3 for Linux x86-64 \cite{cplex} on a Lenovo Yoga 2 Pro (8GB Ram, 4 x 1.8GHZ Intel CPU) to solve our ILP instances.

\subsubsection{A greedy heuristic}
\label{chap:greedy}
To provide an alternative way to address \SBP, 
%In order to address the NP-hardness of \SBP, 
we now resort to a simple greedy heuristic.
This approach starts with the whole set of timepoints and 
greedily removes all timepoints iteratively. 
The greedy rule used here chooses a timepoint that is contained in the smallest number of samples.
In the subsequent steps, these samples will be ignored.

\begin{heuristic}
Let $A_0\defeq \{1,\ldots,n\}$. Compute $A_{\ell+1}=A_\ell\backslash\{k_{\ell+1}\}$ iteratively where
$k_{\ell+1}\in\argmax_{i\in A_\ell}\big\{\rel{A_\ell\backslash\{i\},s_{1:m}}\big\}$ until $A_{\ell+1}=\emptyset$.
Use $A_\ell$ as a solution proposal for \SBP for any $\alpha$ with $\rel{A_{\ell+1},s_{1:m}}<1-\alpha\leq \rel{A_\ell,s_{1:m}}$.
\end{heuristic}

\begin{algorithm}[hbt]
\begin{algorithmic}[1]
%\scriptsize
\footnotesize 
\State{\textbf{Input:} List $s[j]$ that represents the sample $j$, $1\leq j\leq m$}
\State{\hspace{1cm} Sorted lists $\mathcal{C}[i]$ that represents the samples that contain timepoint $i$, $1\leq i\leq n$}
\State{\textbf{Output:}} $A_1,\ldots,A_n$ \Comment{ $A_{\ell+1}=A_\ell\backslash\{k_{\ell+1}\}$ with $k_{\ell+1}\in\argmax_{i\in A_\ell}\big\{\rel{A_\ell\backslash\{i\},s_{1:m}}\big\}$ }
\State{$A_0=\{1,\ldots,n\}$}
\For{$\ell=0,\ldots,n-1$}\Comment{Discard all timepoints iteratively}
	\State{$i=\argmin\big\{\mathcal{C}[k]$.length$\ \mid\ k\in A_\ell\big\}$} \label{alg:ass1}\Comment{Complexity $\mathcal{O}(n)$}
	\State{$A_{\ell+1}=A_\ell\backslash\{i\}$}
	\For {$j=1,\ldots,\mathcal{C}[i]$.length and $k=1,\ldots,s[j]$.length}  \Comment{Remove all samples that contain $i$ from $\mathcal{C}$}
		\State{$\quad\mathcal{C}\Big[s[j][k]\Big]$.remove($j$)}\label{alg:rem1}\Comment{Removing has a complexity of $\mathcal{O}(\ln(m))$} 
	\EndFor
%\State{dropSamples($i$)}\Comment{Complexity $\mathcal{O}(nm\ln(m))$}
\EndFor
%\EndFunction
\end{algorithmic}
\caption{\greedy Algorithm for \SBP}\label{alg:greedy}
%{Implementation of \greedy with a runtime complexity of $\mathcal{O}(n^2m\ln(m))$.}
\end{algorithm}
 
Pseudocode \ref{alg:greedy} provides a possible implementation of \greedy.
The runtime of this algorithm is  $\mathcal{O}(n^2m\ln(m))$.
To see this, observe that the first for-loop runs $n$ times
and the assignment in Line \ref{alg:ass1} needs $\mathcal{O}(n)$ time. 
The second for-loop runs $\mathcal{O}(mn)$ times and the removal of
an element in Line \ref{alg:rem1} can be done in  $\mathcal{O}(\ln(m))$ time. 
Thus, the  overall time-complexity is 
$\mathcal{O}(n^2 + n^2m\ln(m)) = \mathcal{O}(n^2m\ln(m))$.

Although, the worst case runtime-complexity is more or less cubic, Algorithm \ref{alg:greedy} can be well applied in practice for the following reasons.
Assume that every sample has exactly $k$ elements uniformly distributed over $\{1,\ldots,n\}$.
In this case, the number of samples containing $i$ is binomially distributed with parameters $\frac{k}{n}$ and $m$.
Thus, their average number is $m\cdot \frac{k}{n}$.
Finally, the average complexity reduces to $\mathcal{O}(k^2m)$ if we use a hash table to store elements of the $\mathcal{C}[i]$'s, a sorted list to store their lengths and assume that $n\ln(n)<k^2m$.
In turn, if the elements of the samples are non-uniformly distributed and $k$ only represents their maximal cardinality the runtime will be shortened even further.

\greedy provides credible regions for all $\alpha$ values at once. 
However, in contrast to solutions to the ILP, the resulting credible regions will be nested w.r.t. increasing $\alpha$ values.

\subsection{Alternative approaches}
\label{chap:alt_approc}
In this section we explain two alternative Bayesian approaches that can be derived from $\rand C$ as well.
In the first approach we join highest density regions in order to infer credible regions. 
Highest density regions can be considered as a straightforward simultaneous tool to explore distributions of interest \citep{Held_sim_post}.
The second approach, marginal \cp probabilities, is a straightforward pointwise tool in the context of \cps.

Besides this, it should be noted that \cite{guedon2015segmentation} addresses uncertainty of \cp locations through the entropy of $\rand C$.

\subsubsection{Highest density regions}
\label{chap:hdr}
A highest density region (HDR) is a certain subset of a probability space with elements having a higher density value than elements outside of it. 
Such a subset can be utilized to characterize and visualize the support of the corresponding probability distribution \citep{hyndman_computing}. 
In a Bayesian context, $(1-\alpha)$-HDR's are often used as simultaneous $\alpha$ level credible regions \citep{Held_sim_post}. 
In this section, we examine HDR's in general and for the case of random subsets of $\{1,\ldots,n\}$.

Let $\rand X$ be a random variable with a density $\dens$. 
For $\alpha\in[0,1]$, let $q_\alpha$ be an $\alpha$-quantile of $p(\rand X)$, i.e.
$\prob{\dens(\rand X)\leq q_\alpha}\geq \alpha$ and $\prob{\dens(\rand X)\geq q_\alpha}\geq
1-\alpha$. 
\begin{definition} The set
$\{x\mid \dens(x)\geq q_\alpha\}$ is referred to as the \emph{$(1-\alpha)$-HDR} of $\rand X$.
\end{definition}

The $(1-\alpha)$-HDR is a smallest subset of the state space with a probability of at least $1-\alpha$ \citep{box73}.

Now we consider $\rand X=\rand C$.
Our credible regions are subsets of $\{1,\dots,n\}$, whereas the $(1-\alpha)$-HDR of $\rand C$ would be a subset of $2^{\{1,\ldots,n\}}$. 
For the purpose of comparison, we join all successes in the HDR:
\begin{definition} 
Let $s_1,\dots,s_\ell\sub \{1,\ldots,n\}$ be the
$(1-\alpha)$-HDR of $\rand C$. We refer to $\bigcup_{i=1}^\ell s_i$ as
the \emph{joined $(1-\alpha)$-HDR} of $\rand C$. 
\end{definition}
Let $A$ be the joined $(1-\alpha)$-HDR and $\distributionemp$ the distribution of $\rand C$. Then
\begin{align*}
\prob{\rand C\sub A}=\distributionemp\big(2^{A}\big)\geq \sum_{i=1}^\ell\distribution{s_i}\geq 1-\alpha
\end{align*}
Even though, $s_1,\dots,s_\ell$ have high probabilities, $2^{\cardi{A}}$ may be much larger than $\ell$. 
Likewise, $\prob{\rand C\sub A}$ might be considerably larger than $1-\alpha$. 
Thus, $A$ might be substantially larger than elements of $\primeprob{\alpha,\rand C}$. 
In Section \ref{chap:simcredchanpoi} we will see that this happens, especially for small
$\alpha$.

Unfortunately, in many cases we cannot compute HDR's directly and therefore, we use an approximation scheme \citep{Held_sim_post}.
Let now $s_1,\dots,s_m\sub\{1,\ldots,n\}$ be independent samples from the distribution of $\rand C$ that are arranged in descending order according to their values under $\distributionemp$, i.e. $i>j\Rightarrow \distributionemp(s_i)\leq \distributionemp(s_j)$.
We use $\bigcup_{i=1}^\ell s_i$ with $\ell\defeq {\lceil m\cdot(1-\alpha)\rceil}$ as an approximate joined $(1-\alpha)$-HDR.

\subsubsection{Marginal \cp probabilities}
\label{chap:pointwise_sim}
Here we consider they marginal \cp probabilities, i.e. $\prob{i\in \rand C}$ for $i=1,\ldots,n$. 
They can be used to derive subsets of $\{1,\ldots,n\}$ that are closely related to credible regions.

\begin{lemma}
$\{i\mid \prob{i\in \rand C}> \alpha\}$ is a subset of all elements of $\primeprob{\alpha,\rand C}$.
\end{lemma}
\begin{proof}[Proof]
Since $\prob{i\in \rand C}> \alpha$ implies $\prob{\rand C\sub \{1,\ldots,n\}\setminus i}< 1-\alpha$, $i$ has to be part of any $\alpha$ level credible region.
\end{proof}
The Bonferroni correction \citep{dunnett55} can be applied to construct a credible region. 
Let
$\supmarginal{\alpha, \rand C}\defeq\Big\{i\bmid \prob{i\in \rand C}>\frac{\alpha}{n}\Big\}$.
\begin{lemma}
$\prob{\rand C\sub\supmarginal{\alpha, \rand C}}\geq 1-\alpha$
\end{lemma}
\begin{proof}[Proof]
Let $A\defeq\supmarginal{\alpha, \rand C}^c$. 
We conclude that $1-\prob{\rand C\sub A^c}=\prob{\rand C\cap A\not=\emptyset }\leq \sum_{i\in A} \prob{i\in \rand C}\leq \alpha$. 
Thus, $\prob{\rand C\sub\supmarginal{\alpha, \rand C}}=\prob{\rand C\sub A^c}\geq 1-\alpha$ applies.
%Thus, $\mathbb{P}(C\subseteq\mathfrak{B}(\alpha, C))=\mathbb{P}(C\subseteq A^c)\geq 1-\alpha$ applies.
\end{proof}

Pointwise statements suffer from their inability to reflect dependencies. 
To see this, we assume that $\prob{\#\rand C=1}=1$ and  $\prob{i\in \rand C}=\frac{1}{n}$ for all $i=1,\ldots,n$. 
Single timepoints are strongly dependent, e.g. $i\in \rand C$ implies $j\not\in \rand C$ for all $j\neq i$.
Clearly, $\supmarginal{\alpha, \rand C}=\{1,\ldots,n\}$ for all $\alpha<1$.
Moreover, $\{i\mid \prob{i\in \rand C}> \alpha\}=\emptyset$ for $\alpha\geq \frac{1}{n}$.
In contrast to this, elements of $\primeprob{\alpha, \rand C}$ become smaller if $\alpha$ becomes larger. More precisely, $\#A=\lceil (1-\alpha) n\rceil$ for $A\in\primeprob{\alpha,\rand C}$. 
This is due to the fact that pointwise statements ignore dependencies, whereas simultaneous statements incorporate them. 
Hence, in practice $\supmarginal{\alpha, \rand C}$ may be much broader than elements of $\primeprob{\alpha, \rand C}$ and $\{i\mid \prob{i\in \rand C}> \alpha\}$ may be much smaller. 
In Section \ref{chap:simcredchanpoi} we will give an example that underpins this claim.

\subsection{Multiple changepoints and credible regions in more detail}
\label{chap:simcredchanpoi}

In this section we discuss sampling strategies in \cp models and we consider the concept of importance in a more formal way.
Afterwards we reconsider the aforementioned example in order to compare existing approaches with the credible region approach.
Finally, we investigate the performance of \greedy and the convergence speed of \SBP empirically.

\subsubsection{Importance revisited}
\label{chap:importance}

Characterizing the features in the data that should or should not trigger a \cp enables the formulation of \cp problems in the first place.
An intrinsic property of virtually every \cp problem is that the location of the \cp, placed as a consequence of a feature, is uncertain.
Hence, a feature in the data is usually related to a set of locations instead of a single location.
In order to deal with this fact, we require a simultaneous approach that considers all \cp locations that belong to the same feature in a coherent fashion.

To this end, as intimated already, importance provides an estimate of the sensitivity of \cp models towards features in \cp data.
However, it remains to substantiate this claim mathematically.
We start with our definition of sensitivity.
The sensitivity of a \cp model towards a feature in the data is defined as the probability by which posterior random \cps occur as a consequence of this feature.
Thus, if $A\sub\{1,\ldots,n\}$ represents the set of timepoints which are related to the feature, the sensitivity equates to $\prob{\rand C\not\sub A^c}$.
In general, determining the set $A$ is a fuzzy task. 
However, since \cp data usually exhibits the relevant features only occasionally, in most cases we can easily specify these sets to a sufficient extent.

In the introduction the importance of a feature is roughly defined as the height of the peak that relates to this feature.
This means we derive importance from a given range of credible regions.
Assume that we are given a whole range of credible regions $R_\alpha\in\primeprob{\alpha,\rand C}$ for all $\alpha\in[0,1]$.
We now define the importance of a feature as $\hat\alpha=\inf\{\alpha\mid R_\alpha\sub A^c\}$, where $A$ is again the set of locations that belongs to the feature.
With this the following lemma applies.
\begin{lemma}
\label{lem:exclude}
The importance of a feature is an upper bound for the sensitivity of the \cp model towards this feature.
\end{lemma}
\begin{proof}
If there exists an $\epsilon>0$ with $R_{\alpha}\sub A^c$ for all $\alpha\in]\hat\alpha, \hat\alpha+\epsilon[$, we conclude
\begin{align*}
&1-\alpha\leq\prob{\rand C\sub R_{\alpha}}\leq \prob{\rand C\sub A^c}\text{ for all } \alpha\in]\hat\alpha, \hat\alpha+\epsilon[\\
&\Leftrightarrow\quad 1-\hat\alpha\leq \prob{\rand C\sub A^c}
\quad\Leftrightarrow\quad\hat\alpha\geq\prob{\rand C\not\sub A^c}
\end{align*}
In the case where $R_{\hat\alpha}\sub A^c$, we may apply the second part of the above equation directly.
\end{proof}

Since the credible regions are not necessarily nested, in some rare cases, the importance read by means of the above definition may not correspond to a global peak in the credible regions.
Furthermore, if credible regions are only drawn for an incomplete range of $\alpha$'s, we may slightly underestimate the heights of the peaks.
In exchange for these negligible inaccuracies, we obtain an intuitive, visual estimate that provides great insights into the distribution of changepoints and is thus of high practical value.

\begin{figure}[h!t]
\includegraphics[width=\linewidth]{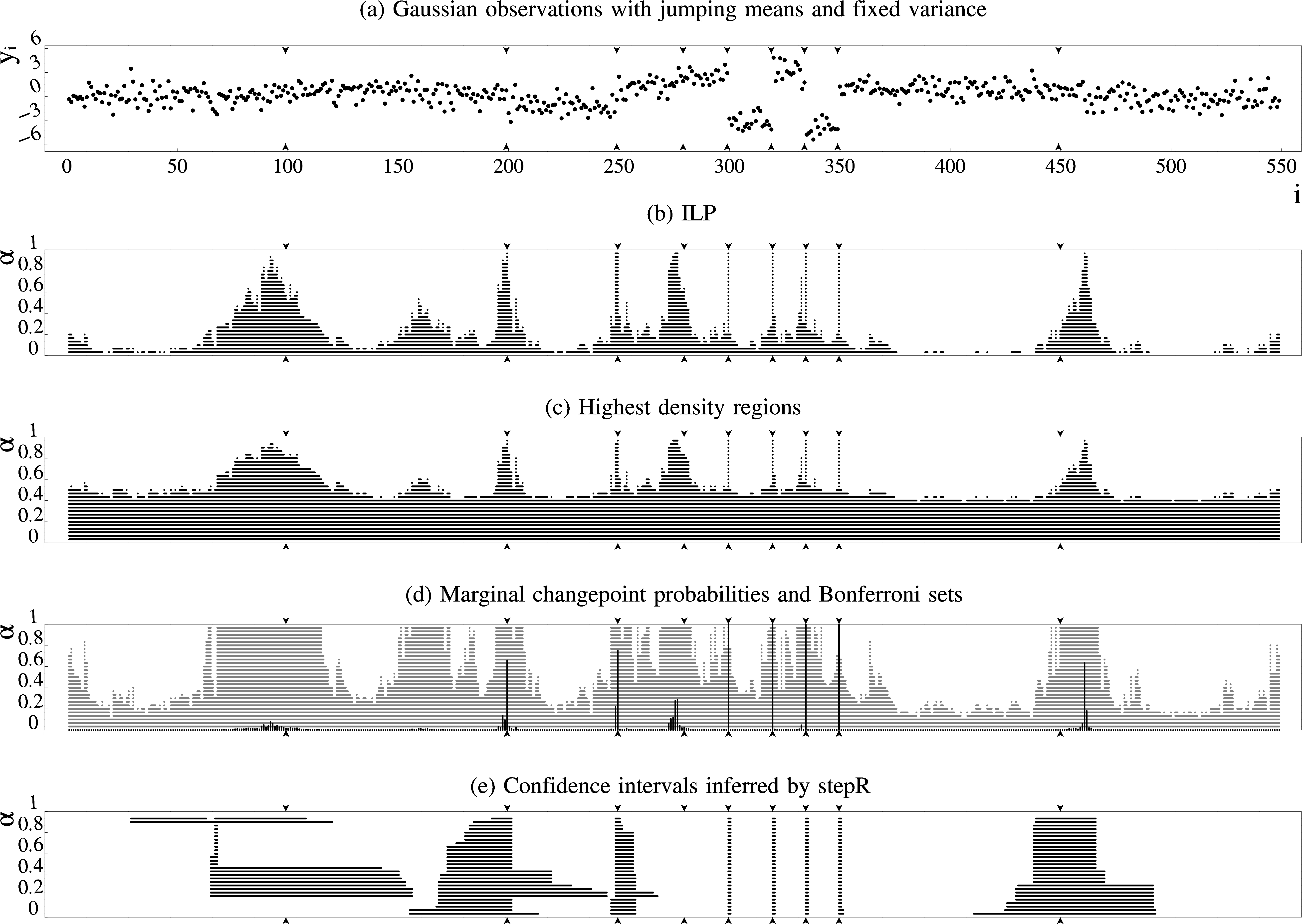}
\caption[Gaussian data and credible regions from ILP, HDR, Bonferroni sets, \stepR]{(a) Simulated independent Gaussian observations with variance 1 and successive changes in mean. 
The true \cps are marked by vertical arrows.
(b) Solutions derived by solving the ILP using $10^5$ samples.
(c) Joined $(1-\alpha)$-HDR's using $10^7$ samples. 
(d) Marginal \cp probabilities in black and $\supmarginal{\alpha, \rand C}$ in gray.
(e) Joined confidence intervals inferred by \stepR.
}
\label{fig:data}
\end{figure}

\subsubsection{Exemplary comparison}
\label{chap:gaus_change_in_mean}

Figure \ref{fig:data}(a) displays the same dataset as Figure \ref{fig:introduction}(a).
There are 6 obvious true \cps at 200, 250, 300, 320, 335 and 350 having a large jump height and two true \cps at 280 and 450 with smaller jump heights of 1 and 0.8, respectively. 
Besides, there is a true \cp at 100 with an even  smaller jump from 0 to 0.5.

We now want to compare the credible regions approach with the highest density regions approach, marginal jump probabilities and confidence intervals inferred by the R package \stepR.
We use the same model as in the introduction.
Figure \ref{fig:data}(b) displays the credible regions with respect to $\alpha=\frac{1}{30},\ldots, \frac{29}{30}$ inferred by solving the ILP.

In the same fashion as before, Figure \ref{fig:data}(c) displays several approximated joined HDR's derived from $10^7$ samples. 
This confirms that joining the elements of an HDR leads to larger subsets of $\{1,\dots,n\}$ compared to elements of \primeprob{\alpha,\rand C}. 
At an $\alpha$ level less than 0.4, the joined HDR already covers most of the timepoints and thus, holds no special information about the \cp locations anymore.

Figure \ref{fig:data}(d) displays the marginal \cp probabilities, i.e. $\prob{i\in \rand C}$ highlighted in black. 
Unfortunately, they don't reflect the possible set of \cps very well since they appear to be too sparse. 
Furthermore, due to a dispersal of the probabilities they are not able to express the sensitivity towards the true \cps (see, for example, the \cps at 100 and 280).
This leads to the conclusion that this approach suffers from a lack of interpretability and is less sensitive than the credible regions approach.

Additionally, Figure \ref{fig:data}(d) shows several credible regions corresponding to $\supmarginal{\alpha, \rand C}$.
As prognosticated in Section \ref{chap:pointwise_sim} these regions are very broad compared to smallest credible regions and thus, less specific.

Figure \ref{fig:data}(e) displays several joined confidence intervals inferred by \stepR \citep{stepR}. 
\stepR first estimates the number of \cps and produces one confidence interval for each \cp.
The plot shows the union of these confidence intervals.
Although confidence sets and credible regions are different by definition, they intend to make similar statements. 

Unfortunately, \stepR does not forecast a confidence interval for the true \cp at 280.
Furthermore, the disappearance of certain \cp locations at decreasing $\alpha$ values seems somewhat confusing and hard to make sense of.
So, the shapes, broadnesses and importances read from the confidence sets are if at all only of little informative value.
Therefore, the authors of this work consider the confidence sets inferred by \stepR as unsuitable for this purpose, because they suffer from a grave lack of interpretability.

The supplement contains a collection of pictures similar to Figure \ref{fig:data}.
There, the data was repeatedly generated using the same \cp locations and mean values.
It becomes apparent that these kind of plots, produced with \stepR, are almost consistently of poor quality.

In Section \ref{chap:importance} we show that the importance of a feature in the data is an upper bound for the sensitivity of the model towards this feature. 
We can examine the differences between these two estimates empirically by means of the above example.
Sensitivity and importance match for all the true \cps except the first one.
The importance of the first \cp is approximately 0.94, whereas its sensitivity with regard to the interval from 0 to 130 is approximately 0.72.
That's a deviation of around 0.2.
The small irregularity at around 170 results in an even bigger deviation of around 0.3.

\begin{figure}[ht]
\begin{center}
\includegraphics[width=\linewidth]{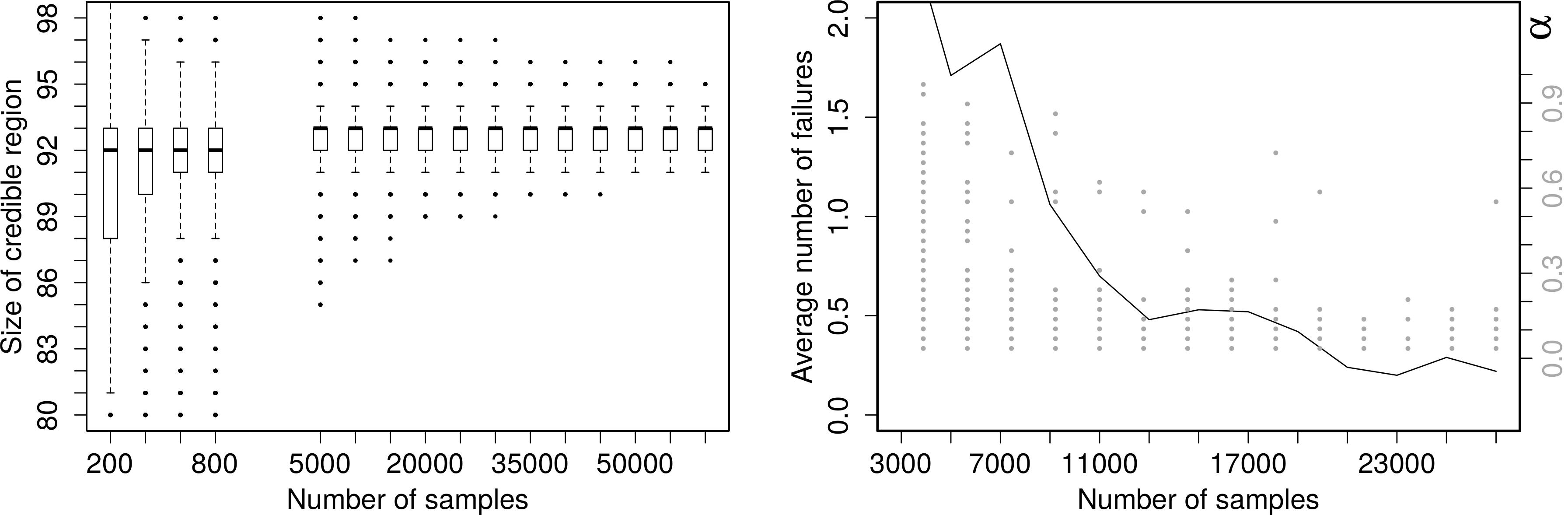} 
\caption[\SBP convergence and accuracy of \greedy]{(left) An empirical convergence proof of \SBP. Boxplots over the sizes of the elements of $\SampleProb{0.3}{s_{1:m}}$ where $m$ is varying according to the x-axis.
(right) Illustration of the accuracy of solutions inferred by \greedy. 
The black graph illustrates the average number of failures within 29 different credible regions (left y-axis). 
The gray dots mark the $\alpha$ values (right y-axis) where failures occurred.}
\label{fig:boxplots}
\end{center}
\end{figure}

\subsubsection{Empirical proof of convergence and the accuracy of \greedy}
\label{chap:conv_and_accuracy}

Figure \ref{fig:boxplots} (left) demonstrates how solutions to \SBP evolve with increasing sample size.
For several sample sizes $m$ it shows boxplots (using 1000 repetitions) over the size of elements of $\SampleProb{0.3}{s_{1:m}}$.
It can be observed that the sizes of the credible regions increase with an increasing sample size. 
This is due to the fact that we need a certain number of samples to cover the possible \cp locations satisfactorily.
However, at the same time, we observe an increasing preciseness, which is a result of Theorem \ref{lem:convergence}.

For different sample sizes we compared for each $\alpha\in\Big\{\frac{1}{30}, \ldots, \frac{29}{30}\Big\}$ the solutions provided by solving the ILP with those provided by \greedy.
In Figure \ref{fig:boxplots} (right), the black graph illustrates the average number of $\alpha$'s (using 100 repetitions) where \greedy failed to provide an optimal solution.
At higher sample counts, there are less than 0.3 of the 29 credible regions wrong. 
The gray points represent the $\alpha$ values (right axis) where the ILP was able to compute smaller regions than \greedy.

Hence, \greedy performs virtually exact on this \cp problem.
However, at smaller sample counts it gets more frequently outwitted by random.
Fortunately, this shows that if \greedy does not compute an ideal region for a certain $\alpha$, it will compute correct ones later again.
This is due to the fact that credible regions for \cp locations are roughly nested (like the solutions provided by \greedy). 
Only in some cases it is possible through solving the ILP to remove certain timepoints a little bit earlier than \greedy. 
Furthermore, because \greedy removes samples having a high number of \cps fast, it works well together with \cp models, since they prefer to explain the data through a small number of \cps.
Thus, we can conjecture that \greedy will perform just as well in most \cp scenarios.

\subsection{An example of use for model selection}
\label{chap:gaussian_change_in_variance}
\begin{figure}[ht]
\includegraphics[width=\linewidth]{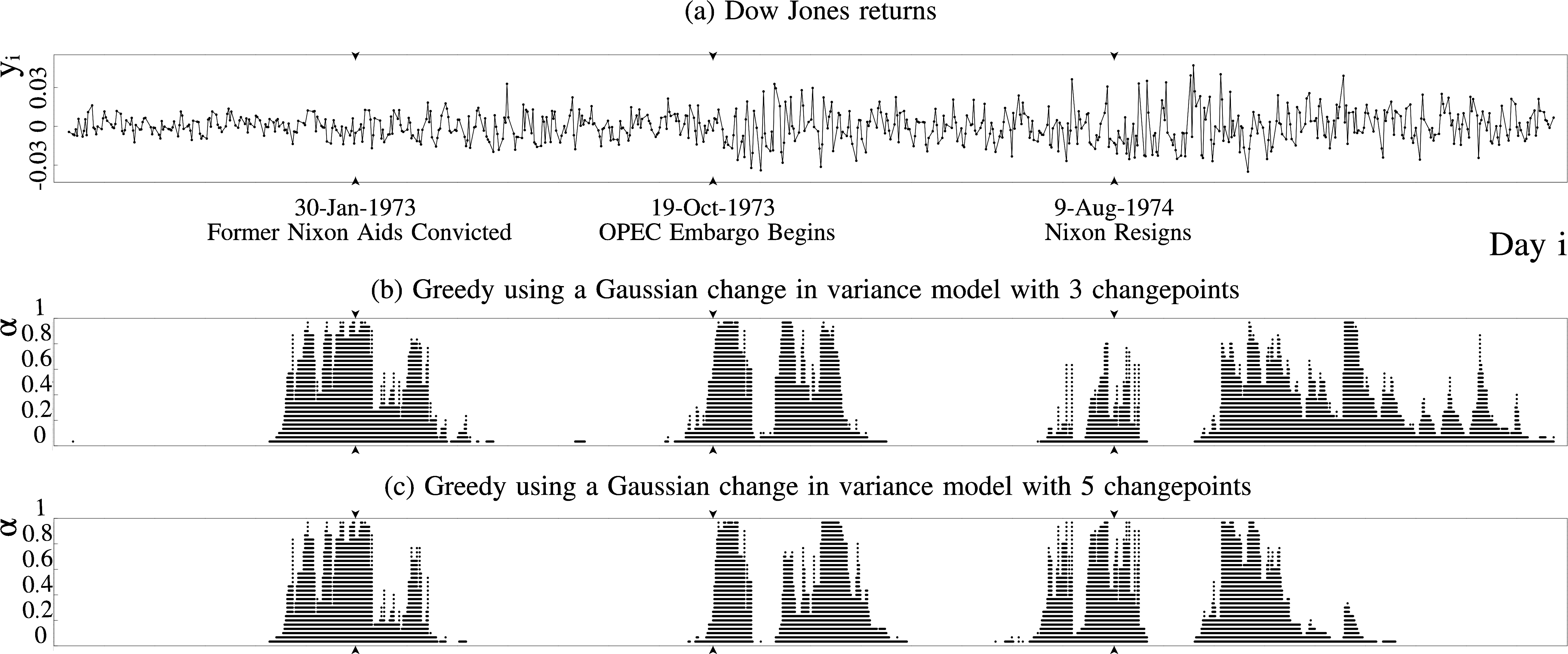}
\caption[Dow Jones data and credible regions]{(a) Dow Jones returns. (b) and (c) Different regions provided by \greedy.
}
\label{fig:dow_jones}
\end{figure}

Now we examine Dow Jones returns observed between 1972 and 1975 \citep{ocpd}, see Figure \ref{fig:dow_jones}(a). 
There are three documented events highlighted on January 1973, October 1973 and August 1974. 
The data is modeled as normally distributed with constant mean equal to 0 and jumping variances.
The variances are distributed according to an inverse gamma distribution with parameters 1 and $10^{-4}$ (see also \cite{ocpd}).
Instead of allowing a random number of \cps, this time we predetermine the number of \cps to three respectively five and our aim is to compare these two model choices by means of credible regions.

As we can see in Figure \ref{fig:dow_jones}(b), the regions become fairly broad especially in the last third of the picture giving rise to additional, nonsensical \cp locations.
The reason for this is that assuming only one \cp after the second event, yields to a misjudgment of the third or fourth variance. 
Thus, the third \cp becomes superfluous and its location highly uncertain.

In contrast, in the case of five \cps in (c) the illustration turns out to be much more differentiated.
The regions stay fairly narrow even for very small values of $\alpha$.
Just a quick look at the credible regions immediately reveals that fixing the number of \cps may impair the inference dramatically. 
Thus, whenever possible, we should use a model that allows the number of \cps to be determined at random.

\subsection{Simultaneous credible regions in the well-log example}
\label{chap:well-logsimcred}

\begin{figure}[ht]
\includegraphics[width=\linewidth]{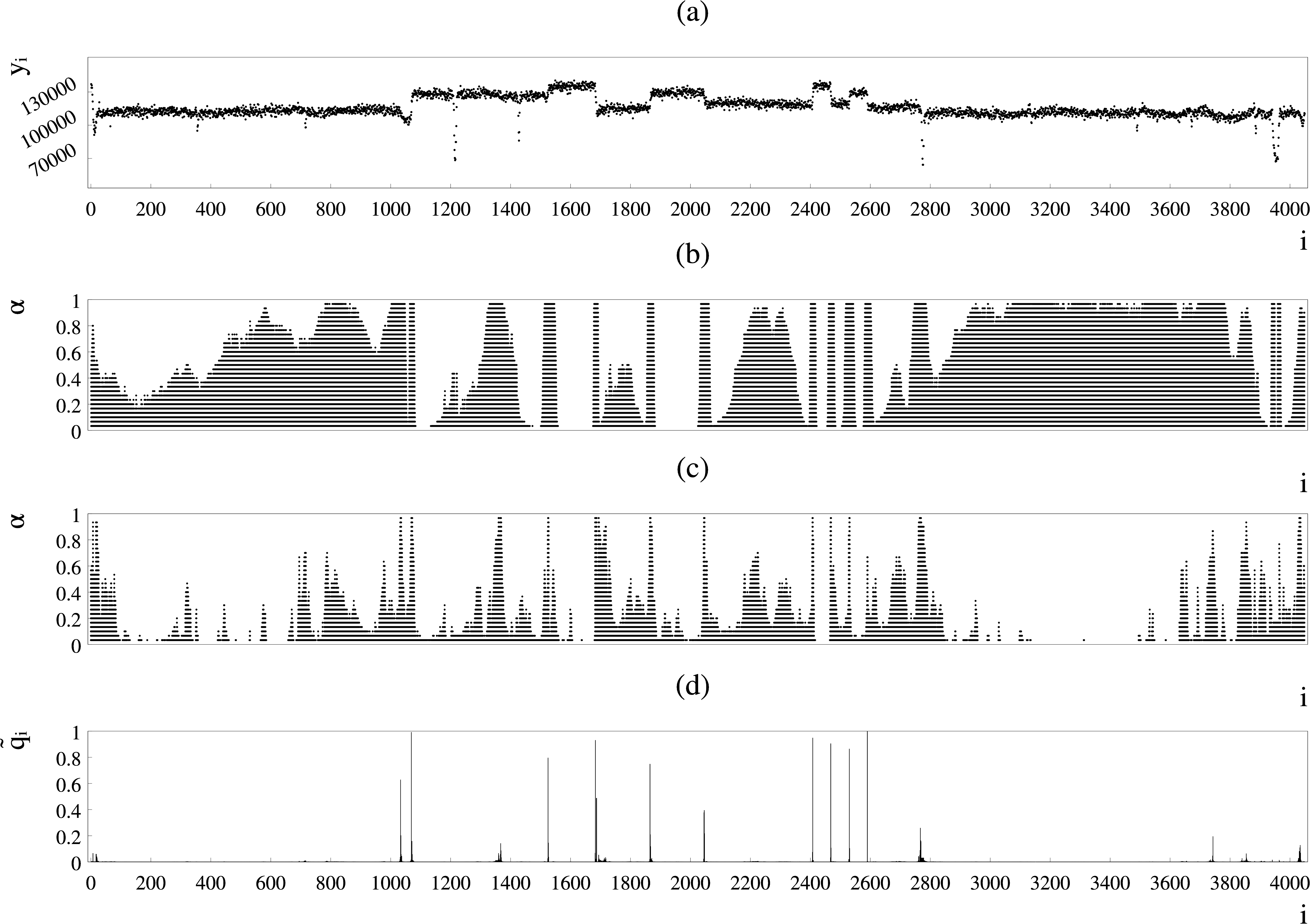}
\caption[Well-log data and credible regions]{(a) Well-Log data. (b) and (c) credible regions inferred by \greedy for the Laplacian change in median model and the model of \cite{exact_fearnhead}, respectively. (d) marginal \cp probabilities for the model of \cite{exact_fearnhead}.
}
\label{fig:well_log_sim}
\end{figure}

One last time, we consider the well-log example of Section \ref{chap:well_log_forward} and the \cp model developed there.
Figure \ref{fig:well_log_sim} (a) depicts the data.
The regions inferred by \greedy from $10^7$ samples are shown in (b).

The credible regions clearly reflect that the model was trained for the sake of \cp inference via MAP estimation and not uncertainty.
In Section \ref{chap:well-log-pointw} we have already seen that the expected number of \cps equates to 17.8.
Thus, within the samples we can find plenty of unnecessary \cps.
The plot shows that these nuisance \cps are mainly located in the first and last part of the data, but also around 1300 and 2200.

We may now try to reduce the uncertainty by increasing the segmental variance in order to obtain a model with a better fit.
%Please note, that decreasing the success probability of the length distribution is not advisable here as the desired segments exhibit comparably short lengths.
However, several experiments have shown that an appropriate reduction of uncertainty gives less then twelve \cps when we use the MAP.

This suggests that a fundamentally different model should be employed in case the focus is more on uncertainty or authenticity.
\cite{exact_fearnhead} proposes a \cp model that employs a piecewise Kalman-filter marginalized over the variance of the latent process.
It is obviously more naturally adapted to the dip at around 3800 and the small shifts in mean.
However, the outliers are challenging since the Kalman-filter is based on the normal distribution.
Therefore, \cite{exact_fearnhead} eliminates the outliers in a preprocessing step, which leaves gaps in the data.

\ref{fig:well_log_sim} (c) shows the corresponding regions inferred by \greedy from $10^7$ samples and (d) the marginal \cp probabilities.
The uncertainty appears to be much lower here.
Especially on the first and last third of the data, there is only a very low tendency to place \cps.
However, at around 1300 there is a high importance to place a nuisance \cp.

In fact, 65\% of the samples exhibit a \cp between 1100 and 1400. 
Our \cp model, in turn, only places a \cp for 36\% of the samples there.
Conversely, our model places a \cp between 2900 and 3900 in 98\% of the samples, whereas the model of \cite{exact_fearnhead} in only 38\%.

As we have already acknowledged, the marginal \cp probabilities (Figure \ref{fig:well_log_sim} (a) and \ref{fig:laplace_marginals} (b)) yield a significant underestimation of sensitivity.
Hence, particularly in such difficult and rather experimental application scenarios, the user should refrain from using marginal \cp probabilities alone to verify \cp models.

Unfortunately, the importances expressed by credible regions tend to overestimate the sensitivity here.
This is because the nuisance \cps spread quite independently over larger sets of timepoints without explaining essential changes in segment height.
As a result, a considerable fraction of the samples contain a mix of the desired and also haphazard \cps.
This is challenging for our credible regions, since their computation is based on a sophisticated rejection of samples.

Even though overestimated, the importances still point to areas of substantial uncertainty, which indicates that both models are not perfectly adapted to the data and application.
To fix this, the model of \cite{exact_fearnhead} might just need an adjustment of the parameters in order to reduce the sensitivity towards the nuisance \cps at around 1300.
However, since \cite{exact_fearnhead} used marginal \cp probabilities in order to justify the parameter choices, this small imperfection remained unnoticed.

%Unfortunately, the undesirable effect caused by spontaneously appearing \cps increases when the data set grows, since there is an increasing chance of seeing haphazard \cps within the samples.
%Thus, when the data becomes very large and the credible regions appear too high, it might be advisable to split the data generously into overlapping chunks and to compute credible regions separately on these chunks.

\todo{}

\subsection{Discussion}
\label{chap:discussion_sim}

In this section we develop a novel set estimator in the context of Bayesian \cp analysis. 
It enables a new visualization technique that provides very detailed insights into the distribution of \cps.
The resulting plots can be analyzed manually just by considering the concepts of
broadness to assess the uncertainty of the \cp locations in regards to a feature,
shape to explore the \cp locations the model is in favor of to explain a feature, and
importance to get an idea about the sensitivity of the model towards a feature.
This greatly facilitates the evaluation and the adjustment of \cp models on the basis of a given dataset.
However, by means of these three concepts, we are also able to conveniently analyze \cp datasets on the basis of a predetermined \cp model.

We say that a credible region points to a feature in the data if the region contains at least one of the \cp locations that are triggered by the model as a consequence of this feature.
An $\alpha$ level credible region will point to all features having a sensitivity larger than $\alpha$.

Against this backdrop, we may derive a single $\alpha$ level credible region from a given model in order to analyze the features in a dataset. 
Thus, we need to choose $\alpha$ in such a way that the corresponding credible region points to the features of interest, but skips the uninteresting ones.
A good \cp model will represent the degree of interest through its sensitivity.
On that basis, we can choose $\alpha$ such that the requirements regarding accuracy and parsimony are met.
\cite{frick_munk_sieling} recommend to choose $\alpha=0.4$ for their method.
This is a reasonable choice.
However, it becomes apparent in Section \ref{chap:gaus_change_in_mean} that the difference between sensitivity and importance can be around 0.3 in difficult cases.
Thus, we recommend to choose $\alpha=0.7$ instead.

An essential quality of the \cp model is its willingness to jump since it highly affects the model's  sensitivity towards features.
In the supplement you can find several video files illustrating this through the success probability for the distribution of the sojourn time between successive \cps.
In this context, if \cp samples are available but the necessary modifications of the model are not feasible anymore,
a sloppy but perhaps effective way to obtain a representative credible region is to increase or decrease $\alpha$ until a reasonable coarseness comes about.
However, systematic defects as in the model depicted in Figure \ref{fig:dow_jones}(b) cannot be tackled in this way.

Credible regions provide valuable knowledge about groups of \cp locations that explain single features jointly.
Besides, they reveal a little bit about combinations of \cps with respect to more than one feature.
The shapes of the credible regions shown in Figure \ref{fig:data}(b) and \ref{fig:dow_jones}(c) suggest that the data can be explained very well through combinations of nine respectively five \cps, which are limited to very few locations.
On the other hand, the medium importance of the feature around 170 in Figure \ref{fig:data}(b) shows that there is a notable alternate representation, which could be quite different to that of the true \cps.

Our theory so far is built on random sets that represent the posterior random \cps.
However, since we can specify a bijective function between subsets of $\{1,\ldots,n\}$ and binary sequences in $\{0,1\}^n$, all the theory in this section applies equally to sequences of binary random variables of equal length.
Thus, we may apply credible regions to ``Spike and Slab'' regression \citep{spike_slab_mitchel, george_mccullogh_gibbs, ishwaran2005} as well.
There, each covariate is assigned to a binary variable that determines if the covariate is relevant for explaining the responses.
However, covariates can generally not be described in terms of a time series and this could hinder a graphical analysis through credible regions.
%It is important to note that in contrast to \cp problems, the covariates can exhibit very intricate dependency structures.
%As a consequence, the accuracy of \greedy and the interpretability of the credible regions may be affected.

Even though the construction of credible regions evoke acceptance regions from statistical testing, we do not intent to create a method in this direction.
As we can see in the example in Section \ref{chap:gaus_change_in_mean}, credible regions can be fairly broad for say $\alpha=0.05$.
While this would interfere with statistical testing purposes, it provides valuable insights into \cp models.

The ILP can be improved by introducing a constraint for each single \cp in the samples, instead of having constraints for larger sets of samples that share the same \cp location (Constraint II).
This is because for an ILP solver, many small constraints are easier to handle than a big one involving many variables.
However, due to high runtimes we do not recommend using the ILP.
While we are able to compute solutions to the Gaussian change in mean example, it is not possible to compute all 29 solutions to the Dow Jones example within a week.

To conclude, the authors of this work highly advice the use of \greedy's credible regions to evaluate and justify \cp models. 
To the same extent, we recommend to use them for the analysis of \cp data.
To this end, the R Package \textit{SimCredRegR} provides a fast implementation of \greedy and plotting routines that can be applied to \cp samples without further ado.

\newpage

\section{Final Remarks}
\label{chap:conclusion}

In this thesis, we elaborated intensively upon three very important topics that crop up in \cp research: sampling, inference and uncertainty.
While \cp samples derived from a bare likelihood open up one way to conduct \cp inference, the computation of certain quantities derived from a \cp model unbar another.
The user of a particular \cp model has to decide which of the two ways or if even a mix of both is used.
This thesis provides the right tools to make and implement these decisions.

Numerous existing and novel inference approaches that go deep into Bayesian \cp models and their characteristics are discussed.
Therewith, the tremendous complexity that arises from random segmentations applied to observed data are demonstrated.
This complexity is mainly owed by the simultaneous discrete and continuous nature of \cp problems.

The resulting mixed state spaces also pose the source of big confusions within the statistics community.
Hence, a sense of measure theory is required to gain a comprehensive understanding here.
However, with such a sophisticated notational framework at hand, until then seemingly difficult issues become readily manageable and straightforward again.

Nevertheless, runtime challenges like polynomial growth and NP-completeness limits the practical feasibility.
Consequently, heuristic approaches like pruning and greedy are of utter importance.
Their skillful use allows for the computation of virtually accurate solutions to otherwise unfeasible problems.

Throughout my time as a PhD student, I worked intensively on these problems.
To this end, I utilized state of the art computing technologies to develop and run R and C++ programs, and I applied sophisticated mathematical tools like measure theory.
With this thesis, I intend to convey my main findings and to help others in solving their own problems.

I finally reckon that spike and slab regression \citep{spike_slab_mitchel, ishwaran2005} is a promising related topic that may benefit considerably from my research.

\newpage
\section{Related own papers and significant contributions of other researchers to this thesis}
\label{chap:contributions}

The content of Section \ref{chap:mcmcfinite} is part of a paper that was accepted in \emph{The Mathematical Gazette} and will be available in the issue July 2020 (see \cite{siemsfinite} for a preprint).
Section \ref{chap:transdim} accounts for a paper that is currently available as a preprint \citep{siemstrans}.
Section \ref{chap:simcredreg} reflects the content of \cite{siems2019simultaneous}.

Professor Marc Hellmuth laid the foundations for Theorem \ref{theo:np} and came up with the corresponding proof idea.
He also pointed the way for the ILP of Section \ref{rem:simplesbp}.
Lisa Köppel indicated that Metropolis-within-Gibbs might be an expedient method to approach post-hoc translations (see Section \ref{chap:metgibs}).

\newpage
\appendix
\section{Supplementary material}
\label{chap:supplement}

The Sections \ref{chap:transdim} and \ref{chap:inference} are each accompanied by a C++/Qt project providing the relevant source code that was used to compute the presented results. 
Both projects contain a README file, which explain how to compile, run and browse through the programs on Ubuntu.

The R Package \textit{SimCredRegR} that comes with the supplementary material of Section \ref{chap:simcredreg}, provides all the applied datasets and sampling algorithms.
It further enables the computation of credible regions according to \greedy, joined highest density regions, marginal \cp probabilities, Bonferroni sets  and \stepR's joined confidence intervals.
Besides, you can find the R Package \textit{SimCredRegILPR} to compute credible regions according to the ILP.
This package can only be installed if IBM's ILP solver CPLEX \citep{cplex} is in place.
Several video (``.mp4'') files, which demonstrate how credible regions evolve at different parameter choices, can be found.
By means of different realizations of the data in Figure \ref{fig:data}, we also provide a comparison of our credible regions and \stepR's confidence intervals.
This can be found in the file ``collection\_of\_different\_data\_simulations.pdf''.

\section{Appendix}
\label{chap:appendix}

\subsection{How to draw a fixed number of \cps uniformly}
\label{chap:unicps}

Within a set of timepoints $1,\ldots,n$, we want to uniformly place exactly $k$ \cps, where $0< k\leq n$.
The following lemma shows how to achieve this in an iterative fashion by deciding for every timepoint separately if a \cp is placed or not.

\begin{lemma}
	By iterating through $i=1,\ldots,n$ and applying the following rule, we obtain $k$ \cps in an uniform manner.
	At timepoint $i$, with $\ell<k$ previous \cps, place a \cp with probability $\frac{k-\ell}{n-i+1}$ or stop if $\ell=k$.
\end{lemma}
\begin{proof}
	For $0\leq i\leq n$, let $\rand K_i\in\{0,\ldots,i\}$ be a random number with
	\begin{align*}
	\prob{\rand K_i=\ell\mid \rand K_{i-1}=\ell'}=\begin{cases}
	\frac{1}{2}&\ell=\ell' \\
	\frac{1}{2}&\ell=\ell'+1
	\end{cases}
	\end{align*}
	for $0<i\leq n$.
	If we look at the random sets $\{i\mid \rand K_i\neq \rand K_{i-1}\}$, we see that any two such sets with the same number of elements exhibit the same probability.
	Thus, by conditioning all $\rand K_i$ on $\rand K_n=k$, we obtain $k$ uniformly distributed timepoints.

	For $0\leq k-\ell\leq n-i$ and $0<i\leq n$, we get $\prob{\rand K_i=\ell, \rand K_n=k\mid \rand K_{i-1}=\ell'}\propto\binom{n-i}{k-\ell}$ 
	since there are $\binom{n-i}{k-\ell}$ options to pick the $k-\ell$ remaining \cps from $\{i+1,\ldots,n\}$.
	Therefore, $\prob{\rand K_i=\ell\mid \rand K_n=k, \rand K_{i-1}=\ell}=\frac{\binom{n-i}{k-\ell}}{\binom{n-i}{k-\ell}+\binom{n-i}{k-\ell-1}}=1-\frac{k-\ell}{n-i+1}$.
\end{proof}

\subsection{Hypergraphs and the complexity of the \SBP}
\label{chap:minkcov}
In this section we introduce the $k$ minimum edge union problem (\minkcov), which is shown to be equivalent to the \SBP.
We show the NP-completeness of the \minkcov. 
In order to establish the theory, we introduce hypergraphs and some of its basic notions. 
Indeed, there is a close connection between hypergraphs and the statistical theory established in this research. 
As we shall see, a hypergraph can be defined in terms of a family of subsets of a finite set and vice versa.
We further provide an Integer Linear Program (ILP) formulation to \minkcov that allows to compute exact solutions and we proof its correctness.

\subsubsection{Hypergraphs}

Here, we briefly discuss (multi-)hypergraphs and their structure
and refer to \cite{berge_hypergraphs,voloshin} for the interested reader. 
Before we start with the formal definitions, we recall that
multisets are a natural generalization of usual sets \citep{calude2001multiset}.
Whereas a usual set contains each element only once, a multiset can contain each
element arbitrary often. 
Therefore, a multiset over a set $A$ is defined in terms of a mapping from $A$ to $\mathbb{N}$ that assigns to each  $a \in A$ the number of occurrences of $a$ in the multiset.

\begin{definition}
	A \emph{hypergraph} $\hyperg$ is a pair $(V,h)$ where $V$ is a finite nonempty set and $h$ is a multiset over $\pows{V}$, i.e. $h:\pows{V}\rightarrow\mathbb{N}$.
\end{definition}

\begin{definition}
	We say that  $e\in \pows{V}$ is an \emph{edge} of the hypergraph $\hyperg=(V,h)$ iff  $h(e)>0$ and write $e\in \hyperg$.
	Moreover, the \emph{cardinality} $\cardi{\hyperg}$ of $\hyperg$ is given as the number of edges $\sum_{e\in\hyperg} h(e)$.
	The \emph{edge union} of a hypergraph $\hyperg$ is defined as the set $\eunion(\hyperg)\defeq\bigcup_{e\in\hyperg} e$. 
	\emph{Removing} an edge $e$ of $\hyperg$ is achieved by setting $h(e)$ to $\max\{0,h(e)-1\}$, 
	whereas \emph{completely removing} an edge $e$ from $\hyperg$ means to set $h(e) =0$.
\end{definition}

Therefore, a set of \cp samples $s_1,\ldots,s_m\sub\{1,\ldots,n\}$ represents a practical example of a hypergraph. 
There, the hypergraph is constructed through $V=\{1,\dots,n\}$ and $h(e)\defeq\cardi{\{i\mid e=s_i\}}$.
Since we can also build a hypergraph from a set of subsets of $\{1,\ldots,n\}$, hypergraphs and families of subsets of a finite set are equal.

\begin{definition}
	A hypergraph $\hypergg=( V, g)$ is a \emph{sub-hypergraph} of a 
	hypergraph $\hyperg=(V, h)$ iff $g(e)\leq h(e)$ for all $e\in\pows{ V}$ and we write \hypergg\sub\hyperg.
	%\end{definition}
	%We will also refer to inclusion of usual sets by $\sub$. 
	Moreover, for a given vertex $x\in V$ and a hypergraph $\hyperg=( V, h)$ 
	we define the sub-hypergraph  $\vertcov(\hyperg,x)=( V, g)$ through 
	\begin{align*}
	g(e)=\begin{cases}      h(e)&,\text{if }x\in e\\
	0&, \text{otherwise}
	\end{cases}
	\end{align*}
\end{definition}
$\vertcov(\hyperg,x)$ represents the sub-hypergraph of $\hyperg$ that consists of all edges that contain vertex $x$. 

%\subsection{Problem statement}

\subsubsection{Computational complexity of the \SBP}
\label{chap:proofnp}

For a given hypergraph $\hyperg$ and an integer k with $0\leq k\leq \#\hyperg$, we want to solve the problem of finding a sub-hypergraph $\hypergg\sub\hyperg$ that has at least $k$ edges but an edge union of minimum cardinality.
Equivalently, we want to find at least one element of the set
% minimum number of vertices. To be more precise:
%\namedproblem{$\boldsymbol{k}$ minimum edge union (\minkcov)}{\newline \
%\emph{Input:}  A hypergraph $\hyperg=( V,f)$ and an integer $0<k\leq \cardi{\hyperg}$ $L$. \\
%\emph{Question:}  Is there  a sub-hypergraph $\hypergg$ of $\hyperg$ with  
%						$\cardi{\hypergg}\geq k$ and 
%						 $k$ edges	that has an edge union with a minimal number of vertices compared to all sub-hypergraphs of $\hyperg$ with at least $k$ edges. 	
%Mathematically, we are searching for 
$\argmin\limits_{\hypergg\sub \hyperg}\Big\{\cardi{\eunion(\hypergg)} \bmid \cardi{\hypergg}\geq k\Big\}$
We refer to this task as the  \emph{$k$ minimum edge union (optimization) problem} (\minkcov).
\begin{myremark}
	\label{bem:eqsoltom}
	Provided that a solution $\hypergg=( V,g)$ to a \minkcov instance $(\hyperg, k)$ with $\hyperg=(V,h)$ is known, 
	one can easily determine further 
	solutions $\hypergg'=( V,g')$ by choosing for all edges $e\in \hyperg$ an 
	arbitrary integer $\ell_e$ with $g(e)\leq \ell_e\leq h(e)$ and setting 
	\begin{align*}
	g'(e)\defeq \begin{cases}\ell_e&,\ \text{if } g(e)>0\\0&,\ \text{otherwise}\end{cases}
	\end{align*}
	Clearly, $\cardi{\eunion(\hypergg)}=\cardi{\eunion(\hypergg')}$ and $\cardi{\hypergg'}\geq \cardi{\hypergg}\geq k$.
\end{myremark}

\begin{theorem}
	\label{theo:minksbpeq}
	The problems \minkcov and \SBP are equivalent. 
	\label{thm:equi}
\end{theorem}

Before proving this theorem, we need to consider the following lemma.
\begin{lemma}
	\label{lem:minksbp}
	Given a hypergraph $\hyperg=(\{1,\ldots,n\},h)$ and a sequence $s_1,\ldots,s_m\sub\{1,\ldots,n\}$ with $m=\#\hyperg$ and $h(e)=\#\{i\mid e=s_i\}$ for all $e\sub \{1,\ldots,n\}$, the following applies
	\begin{flalign}
	\label{eq:minksbp}
	\min\limits_{A\sub\{1,\ldots,n\}}\Big\{\#A\bmid \sum\limits_{i=1}^m\ind{s_i\sub A}\geq k\Big\}=\min\limits_{\hypergg\sub \hyperg}\Big\{\#\eunion(\hypergg)\bmid \#\hypergg\geq k\Big\}
	\end{flalign}
\end{lemma}
\begin{proof}[Proof of Lemma \ref{lem:minksbp}]
	For $A\sub\{1,\ldots,n\}$ with $\sum_{i=1}^m\ind{s_i\sub A}\geq k$, we construct the hypergraph $\hypergg=(\{1,\ldots,n\},g)$ with
	$g(e)=\#\{i\mid e=s_i, s_i\sub A\}$. 
	Since $\#\hypergg\geq k, \hypergg\sub\hyperg$ and $\#\eunion(\hypergg)\leq\#A$ we conclude that $\geq$ holds in Equation (\ref{eq:minksbp}).
	
	Given a hypergraph $\hypergg\sub\hyperg$ with $\#\hypergg\geq k$, we can chose an $I\sub\{1,\ldots,m\}$ with $\#I\geq k$ and $\bigcup_{i\in I}s_i\sub\eunion(\hypergg)$.
	Since $\sum_{i\in I}\ind{s_i\sub \eunion(\hypergg)}\geq k$ we can also conclude that $\leq$ holds in Equation (\ref{eq:minksbp}).
\end{proof}

\begin{proof}[Proof of Theorem \ref{theo:minksbpeq}]
	Let $(\hyperg,k)$ be a \minkcov instance with $\hyperg=(\{1,\ldots,n\},h)$.
	We construct an equivalent \SBP instance with $s_1,\ldots,s_m\sub\{1,\ldots,n\}$ and $\alpha\in[0,1]$ such that 
	a solution to this \SBP instance provides a solution to the \minkcov instance.
	Therefore, choose $m=\#\hyperg,\ \alpha=1-\frac{k}{m}$ and $s_1,\ldots,s_m$ so that $h(e)=\#\{i\mid e=s_i\}$.
	If $A$ is a solution to this \SBP instance, then 
	the hypergraph $\hypergg=(\{1,\ldots,n\},g)$ 
	with $g(e)=\#\{i\mid e=s_i, s_i\sub A\}$ is a solution to the \minkcov instance $(\hyperg,k)$.
	Indeed, since $\#\eunion(\hypergg)= \#A$ and $\#\hypergg\geq k$, Lemma \ref{lem:minksbp} implies 
	that $\hypergg$ is a solution to the \minkcov instance.
	
	Conversely, given an \SBP instance with $\alpha\in[0,1]$ and $s_1,\ldots,s_m\sub\{1,\ldots,n\}$, we construct an 
	equivalent \minkcov instance $(\hyperg,k)$
	such that a solution to this \minkcov instance provides a solution to the \SBP instance.
	Therefore, let $k=\lceil m\cdot(1-\alpha)\rceil$ and $\hyperg=(\{1,\ldots,n\},h)$ with $h(e)=\#\{i\mid e=s_i\}$. 
	If $\hypergg$ is a solution to this \minkcov instance, then $\eunion(\hypergg)$ is a solution to the \SBP instance.
	Indeed, since $\sum_{i=1}^m\ind{s_i\sub \eunion(\hypergg)}\geq m\cdot(1-\alpha)$, Lemma \ref{lem:minksbp} implies 
	that $\eunion(\hypergg)$ is a solution to the \SBP instance.
\end{proof}

We now show that (the decision version of) \minkcov (and hence of \SBP) is NP-complete \citep{Garey_johnsons}. 
Thus, there is no polynomial time algorithm to solve this problem, unless $P=NP$. 

The decision version of \minkcov is as follows:

\begin{problem}[{\bf Decision Version of \minkcov}] \ \\
	\setlength{\tabcolsep}{5pt}
	\begin{tabular}{ll}
		\textit{Input:} &  Hypergraph $\hyperg=( V,h)$ and $k,l\in\mathbb{N}$ with  $0\leq k\leq \cardi{\hyperg}$ and $0<l\leq \cardi{\eunion(\hyperg)}$. \\
		\textit{Question:} & Is there a $\hypergg=(V,g)\sub \hyperg$  
		such that $\cardi{\hypergg}\geq k$ and  $\cardi{\eunion(\hypergg)}\leq l$ ? 
	\end{tabular}
\end{problem}

In order to prove the NP-completeness of \minkcov, we use the well-known 
NP-complete \knaps-problem \citep{karp1972reducibility, Garey_johnsons}. 

\begin{problem}[{\bf \knaps}] \ \\
	\setlength{\tabcolsep}{5pt}
	\begin{tabular}{ll}
		\textit{Input:} &  A finite set $U$ and for each $u\in U$ a weight 
		$w(u)\in \mathbb{N}$, a value $v(u) \in \mathbb{N}$ and  \\
		&positive integers $a$ and $b$.  \\
		\textit{Query:} & Is there a subset $U'\sub U$ such that 
		$\sum_{u\in U'} w(u)\leq b$ and $\sum_{u\in U'} v(u)\geq a$?
	\end{tabular}
\end{problem}

%b=W, a=V

\begin{theorem}
	\label{thm:NP-complete1}
	\minkcov is NP-complete.
\end{theorem}
\begin{proof}[Proof]
	We begin with showing that \minkcov $\in  \mathtt{NP}$. 
	To this end, it suffices to demonstrate that a candidate solution to \minkcov
	can be verified in polynomial time.
	However, this is easy to see, since we only need to check whether 
	for a possible solution $\hypergg\sub \hyperg$ 
	it holds that $\cardi{\eunion(\hypergg)} \leq l$ and $\cardi{\hypergg} \geq k$. 
	Both tasks can be done in linear time in the number of edges of $\hypergg$.

	We proceed to show by reduction from \knaps that \minkcov is NP-hard. Thus, let
	us assume we are given an arbitrary instance of \knaps, that is, a finite set
	$U$, for each $u\in U$ the weight $w(u)\in \mathbb{N}$ and the value $v(u) \in
	\mathbb{N}$, as well as positive integers $b$ and $a$.  
	Now, we construct an instance of \minkcov as follows. 
	For each $u\in U$ we set an edge $e_u\defeq \{(u,1),\dots,(u,v(u))\}$ and $h(e_u)=w(u)$. 
	The vertex set of the hypergraph $\hyperg=(V,h)$ is then $V=\cup_{u\in U} e_u$. 
	Note, the edges in $\hyperg$ are pairwise disjoint. 
	Clearly, this reduction can be done in polynomial time in the number 
	of elements in $U$ and the values $v(u)$. 
	
	In what follows, 
	we show that \knaps has a solution for given integers $b, a$ if and only if
	\minkcov has a solution with $k=\cardi{\hyperg}-b$ and $l=\cardi{\eunion(\hyperg)}-a$.
	
	Let $U'=\{u_1,\dots,u_n\}\subseteq U$ such that 
	$\sum_{i=1}^n w(u_i)\leq b$ and $\sum_{i=1}^n v(u_i)\geq a$.
	Completely remove all corresponding edges $e_{u_i}$, $1\leq i\leq n$
	from $\hyperg$ to obtain the sub-hypergraph $\hypergg$. 
	Hence, $\cardi{\hypergg} = \cardi{\hyperg}-\sum_{i=1}^n w(u_i) \geq \cardi{\hyperg}-b = k$
	and $\cardi{\eunion(\hypergg)} = \cardi{\eunion(\hyperg)} - \sum_{i=1}^n v(u_i) \leq
	\cardi{\eunion(\hyperg)} - a = l$.
	
	Conversely, assume that $\hypergg = (V,g)$ is a valid solution for 
	the hypergraph $\hyperg=(V,h)$ (as constructed above) and given
	integers $k\geq \cardi{\hyperg}$ and $l\leq \cardi{\eunion(\hyperg)}$.
	Thus, we can write $k= \cardi{\hyperg}-b$ and $l= \cardi{\eunion(\hyperg)}-a$.
	Hence, $\cardi{\hypergg}\geq \cardi{\hyperg}-b$ and 
	$\cardi{\eunion(\hypergg)}\leq \cardi{\eunion(\hyperg)}-a$. 
	Therefore, at least $b$ edges must have been removed from $\hg$
	resulting in $\hgg$ where the edge union of $\hgg$ has at least
	$a$ fewer vertices than  $\eunion(\hyperg)$.
	Note, to obtain fewer vertices in $\eunion(\hyperg)$ one needs to 
	completely remove edges from $\hg$. 
	Let $E_0=\{e\in \hg \mid g(e)=0\}$ be 
	the set of all edges that have been completely removed from $\hg$.
	By construction, each edge $e_u$ is uniquely identified with an
	element $u\in U$. 
	We show that $U' = \{u\in U\mid e_u\in E_0\}$ provides a solution
	for \SBP.
	To this end, observe that 
	$\cardi{\hg}-\sum_{e\in E_0} h(e)\geq \cardi{\hgg}\geq \cardi{\hg}-b $, which implies that
	$\sum_{e\in E_0} h(e) = \sum_{u\in U'} w(u) \leq b$, as desired. 
	Moreover, by construction we have 
	$\cardi{\eunion(\hypergg)} = \cardi{\eunion(\hyperg)} - \sum_{e\in E_0} \cardi{e} = 
	\cardi{\eunion(\hyperg)} - \sum_{u\in U'} v(u) \leq  \cardi{\eunion(\hyperg)} - a$. 
	Thus, $\sum_{u\in U'} v(u)\geq a$, which completes the proof. 
\end{proof}

Combining Theorem \ref{theo:minksbpeq} and \ref{thm:NP-complete1} we obtain the following
\begin{corollary}
	The decision version of \SBP is an NP-complete problem. 
\end{corollary}

\subsection{Proof of correctness of the ILP}
\label{chap:ilpproof}

We introduce for a hypergraph $\hyperg=(V,h)$ and an integer $k$ the following binary variables $U_x, F_e\in \{0,1\}$:
\begin{description}
	\item[$U_x = 1$] if and only if vertex $x$ of $\hg$ is contained in 
	the edge union $\eunion(\hgg)$ of $\hgg\sub \hg$.  
	\item[$F_e = 1$]  if and only if the edge $e\in \hg$ is contained in $\hgg$. 
	%						and thus, not completely removed from $\hg$. 
	%  \item[$E_{x,e}=1$] if and only if edge $e\in \hg$ contains vertex $x$. 
\end{description}
Moreover, the number of edges that contain a vertex $x$ is given by 
the constants $\contx \defeq \cardi{\vertcov(\hyperg,x)}$ for all $x\in V$.

To find a solution for the \minkcov problem, we need to minimize the number of vertices in the edge union of $\hgg=(V,g)\sub \hg$, which is achieved by minimizing the objective function 
\begin{align}
&\sum\limits_{x\in\eunion(\hyperg)}U_x
\label{eq:objF}
\end{align}
By Remark \ref{bem:eqsoltom} it is always possible to find a sub-hypergraph $\hgg$
with minimum edge union such that $g(e)=h(e)$ for all $e\in \hgg$. 
We will construct such a sub-hypergraph. 

To ensure that
$\cardi{\hgg}\geq k$ we add the constraint
\begin{align}
&\sum\limits_{e\in \hyperg}h(e)\cdot F_e\geq k
\label{eq:kedges}
\end{align}
Note, $e$ is contained in $\hgg$ if and only if $F_e=1$ and by
construction, $\cardi{\hgg}=\sum_{e\in \hyperg}(h(e)\cdot F_e)$. 
Thus, constraint \eqref{eq:kedges} is satisfied if and only if 
$\cardi{\hgg}\geq k$.

Finally, we have to ensure that $U_x=1$ if and only if there is an edge in $\hgg$ that 
contains $x$. To this end, we add for all $x\in\eunion(\hyperg)$ the constraint
\begin{align}
&\sum\limits_{e\in\vertcov(\hyperg,x)}h(e)\cdot (1- F_e)\geq \contx\cdot (1-U_x)
\label{eq:Ux}
\end{align}
Now, if there is no edge containing $x$ in $\hgg$ and thus, 
$F_e=0$ for all $e\in\vertcov(\hyperg,x)$,  then 
Constraint \eqref{eq:Ux} implies that
$\sum_{e\in\vertcov(\hyperg,x)}h(e)\geq \contx(1-U_x)$.
Since $\sum_{e\in\vertcov(\hyperg,x)}h(e) = \contx$, 
we have two choices for $U_x\in \{0,1\}$.
However, the optimization function ensures that
$U_x$ is set to $0$. 
Conversely, assume that there is an edge $e$ that contains $x$ and hence, $F_e=1$.
Thus, $\sum_{e\in\vertcov(\hyperg,x)}h(e)\cdot (1- F_e) < \contx$. 
The only way to satisfy Constraint \eqref{eq:Ux} is achieved by setting 
$U_x = 1$. 

Taken together the latter arguments we can infer the following result. 
\begin{theorem}
	The ILP formulation in Eqs. \eqref{eq:objF} - \eqref{eq:Ux} correctly
	solves the \minkcov problem.
\end{theorem}

\newpage

\newpage
\addcontentsline{toc}{section}{Symbols}
\mbox{}
\nomenclature[A, 01]{$\mathbb{N}$}{Natural numbers including 0}
\nomenclature[A, 02]{$\mathbb{N}_+$}{Natural numbers excluding 0}
\nomenclature[A, 03]{$\mathbb{R}$}{Real numbers}
\nomenclature[A, 04]{$\mathbb{R}_{\geq 0}$}{Non-negative real numbers}
\nomenclature[A, 05]{$\mathbb{R}_{> 0}$}{Positive real numbers}
\nomenclature[A, 06]{$\#$}{Cardinality}

\nomenclature[B, 01]{$\probemp$}{The all-embracing probability measure}
\nomenclature[B, 02]{$\otimes$}{Product operator for either two measures or a measure and a Markov kernel}
\nomenclature[B, 03]{$\densslash$}{$\pi\densslash \mu$ is the density of $\pi$ w.r.t. $\mu$}
\nomenclature[B, 04]{$\delta_x$}{The Dirac or point measure around $x$}
\nomenclature[B, 05]{$\#$}{Counting measure}

\nomenclature[C, 01]{$\eqdist$}{$\rand X\eqdist\rand Y$ means that the random variables $\rand X$ and $\rand Y$ follow the same distribution}
\nomenclature[C, 02]{$\sim$}{$\rand X\sim\pi$ means that $\pi$ is the distribution of $\rand X$}
\nomenclature[C, 03]{$\indep$}{$\rand X\indep\rand Y$ means that the random variables $\rand X$ and $\rand Y$ are independent}
\nomenclature[C, 04]{$\mid$}{Opens a list of conditions}
\nomenclature[C, 05]{$\semic$}{Opens a list of (non-random) parameters}
\nomenclature[C, 06]{$\expecemp$}{Expectation}

\nomenclature[D, 01]{$\euleremp$}{Exponential function in standard form}
\nomenclature[D, 01]{$\ln$}{Natural logarithm}
\nomenclature[D, 02]{$\indemp$}{Indicator function}
\nomenclature[D, 03]{$\phi(x\semic \mu, \sigma^2)$}{Density of normal distribution with expectation $\mu$, variance $\sigma^2$ evaluated at $x$}

\nomenclature[E, 01]{$\mcal{O}$}{Big O notation}
\nomenclature[E, 02]{$\LargerCdot$}{Universal placeholder that represents all possible arguments}

\nomenclature[F, 01]{NB$(\ell\semic q,r)$}{Probability of $\ell$ failures under the negative binomial distribution with success probability $q$ and number of successes $r$}
\nomenclature[F, 02]{$\mathcal{N}(\mu, \sigma^2)$}{Normal distribution with expectation $\mu$ and variance $\sigma^2$}

\printnomenclature[2cm]
\stepcounter{section}
\addcontentsline{toc}{section}{Figures}
{\setlength{\parskip}{0ex}\setlength{\itemsep}{1ex}\listoffigures}
\stepcounter{section}
\stepcounter{section}
\addcontentsline{toc}{section}{Tables}
{\setlength{\parskip}{0ex}\setlength{\itemsep}{0ex}\listoftables}
\addcontentsline{toc}{section}{Pseudocodes}
\listofalgorithms
%\addcontentsline{toc}{chapter}{List of Tables}
\section*{List of Acronyms}
\begin{acronym}
	\setlength{\parskip}{0ex}
	\setlength{\itemsep}{0ex}
	\acro{MAP   }{\quad Maximum A-Posteriori}
	\acro{EM    }{\quad Expectation Maximization}
	\acro{BIC   }{\quad Bayesian Information Criterion}
	\acro{AIC   }{\quad Akaike Information Criterion}
	\acro{MCMC  }{\quad Markov Chain Monte Carlo}
	\acro{SDT   }{\quad Semi Deterministic Translation}	
	\acro{SBP   }{\quad Sampling Based Problem}
	\acro{ILP   }{\quad Integer Linear Programming}
	\acro{HDR   }{\quad Highest Density Region}	
\end{acronym}


\renewcommand\refname{List of References}
\addcontentsline{toc}{section}{References}
\bibliographystyle{chicago}      % basic style, author-year citations
\footnotesize\bibliography{../bibtex}   % name your BibTeX data base
\end{document}